\documentclass[11pt]{amsart}

\usepackage{amsmath, amsthm, amssymb, hyperref, color}
\usepackage{graphicx}
\usepackage{caption}
\usepackage{subcaption}

\usepackage[shortlabels]{enumitem}
\usepackage{graphicx}
\usepackage{makecell}
\usepackage{mhchem}
\usepackage[ruled,vlined,linesnumbered]{algorithm2e}

\usepackage{tikz}
\usepackage{pgfplots} %plot graphs
\usepackage{comment}
\usetikzlibrary{arrows,positioning,matrix,decorations.pathmorphing,calc}

\usepackage{booktabs}

\usepackage{lipsum}
\usepackage{mathrsfs}

\oddsidemargin \evensidemargin
\theoremstyle{plain}
\newtheorem{theorem}{Theorem}[section]
\newtheorem{corollary}[theorem]{Corollary}
\newtheorem{proposition}[theorem]{Proposition}
\newtheorem{lemma}[theorem]{Lemma}

\theoremstyle{definition}
\newtheorem{definition}[theorem]{Definition}
\newtheorem{example}[theorem]{Example}

\newtheorem{remark}[theorem]{Remark}
\newtheorem{proc}[theorem]{Procedure}

\newcommand{\pd}[2]{\ensuremath{ \frac{ \partial #1 }{\partial #2}  }}
\newcommand{\scc}{\mathcal{S}}
\newcommand{\F}{f_{c, \kappa}}
\newcommand{\augmentH}{h_{c, a}}

\newcommand{\Sp}{\mathscr{S}}

\newcommand{\Katau}{\boldsymbol{\tau}}
\newcommand{\uri}{\rightarrow_\circ}
\newcommand{\cond}{\mathcal{C}}
\newcommand{\R}{\mathbb{R}}
\newcommand{\Q}{\mathbb{Q}}
\newcommand{\Lap}{\mathcal{L}}

\newcommand\mapsfrom{\mathrel{\reflectbox{\ensuremath{\mapsto}}}}

\begin{document}

\title[Multistationarity in structured reaction networks]{Multistationarity in structured reaction networks}


\author[A.~Dickenstein, M.~P\'erez Mill\'an, A.~Shiu, and X.~Tang]
{Alicia Dickenstein, Mercedes P\'erez Mill\'an,
Anne Shiu, and Xiaoxian Tang%~\footnote{Corresponding author}
}

\address{A.~D. and M.~P.~M: Departamento de Matem\'atica\\
FCEN, Universidad de Buenos Aires e IMAS (UBA--CONICET)\\
Ciudad Universitaria, Pab.I \\
1428 Buenos Aires, Argentina \\
A.~S. and X.~T.: Department of Mathematics\\ Texas A\&M University\\ College Station TX 77843, USA.} 

\email{alidick@dm.uba.ar,mpmillan@dm.uba.ar,annejls@math.tamu.edu,xiaoxian@math.tamu.edu}

%\thanks{AD and MPM are partially supported by UBACYT 20020170100048BA, 
%CONICET PIP 11220150100473, and ANPCyT PICT 2016-0398, Argentina.} 

%\date{\today}
\date{January 23, 2019}

% \author[1]{Alicia Dickenstein}
% \author[1]{Mercedes P\'erez Mill\'an} 
% 
% \author[2]{\\ Anne Shiu}
% \author[2]{Xiaoxian Tang}
% \affil[1]{Dto.\ de Matem\'atica, FCEN, Universidad de Buenos Aires}%, 
% %    Ciudad Universitaria, Pab.\ I, C1428EGA Buenos Aires, Argentina}
% \affil[2]{Texas A\&M University}

\maketitle

\begin{abstract}
Many dynamical systems arising in biology and other areas exhibit multistationarity  
(two  or  more  positive  steady states with the same conserved quantities).  
Although deciding multistationarity for a polynomial dynamical system
is an effective question in real algebraic geometry, it is in general difficult to 
determine whether a given network can give rise to a multistationary system, and if so, 
to identify witnesses to multistationarity, that is, specific parameter values for which the system exhibits 
multiple steady states.  Here we investigate both problems.
First, we build on work of Conradi, Feliu, Mincheva, and Wiuf, 
who showed that for certain reaction networks 
whose steady states admit a positive parametrization, 
multistationarity is characterized by whether a certain ``critical function'' changes sign.
Here, we allow for more general parametrizations, 
which make it much easier to determine the existence of a sign change. This is particularly simple 
when the steady-state equations are linearly equivalent to binomials;
we give necessary conditions for this to happen, 
which hold for many networks studied in the literature.
We also give a sufficient condition for multistationarity 
of networks whose steady-state equations can be replaced by equivalent triangular-form equations.
Finally, we  present methods for  
finding witnesses to multistationarity, which we show work well 
for certain structured reaction networks, including those common to biological signaling pathways.
Our work relies on results from degree theory, on the existence of explicit rational 
parametrizations of the steady states, and on the specialization of Gr\"obner bases.

  \vskip 0.1cm
%: Keywords
  \noindent \textbf{Keywords:} reaction network, mass-action kinetics,
%steady state, 
multistationarity, parametrization,
%multisite phosphorylation,
binomial ideal, 
Brouwer degree, Gr\"obner basis

\end{abstract}

\section{Introduction}

An important problem in many applications is to determine whether a given dynamical system is multistationary, and if so to find 
witnesses to multistationarity (parameter values for which the system exhibits two or more steady states with the same
conserved quantities).  
Here we resolve these problems for dynamical systems arising under mass-action kinetics from reaction networks with particular structure;  for instance, 
biological signaling networks.
Specifically, our main results are criteria for multistationarity and procedures for obtaining witnesses  
for networks having the following structure:
\begin{enumerate}[(A)]
   	\item networks that admit a rational parametrization of the steady states and where the resulting 
 {\em critical function} changes sign (Theorem~\ref{thm:c-general}), and 
	\item networks for which the steady-state equations are linearly equivalent to binomial equations 
	(Theorems~\ref{thm:NP} and \ref{thm:binom}) or triangular-form equations (Theorem~\ref{prop:triangular}).
\end{enumerate}
The critical function in (A) refers to the composition of the steady-state
parametrization with the determinant of the Jacobian matrix.  
Theorem~\ref{thm:c-general} therefore generalizes recent results,
which rely on degree theory, due to Conradi, Feliu, Mincheva, and Wiuf~\cite{CFMW}.
Specifically, we consider  more general steady-state parametrizations 
 since we allow the rate-constant parameters to depend on the steady-state concentrations, 
 and we replace the parameters with ``effective parameters'', which are usually fewer.
We show that the resulting critical functions can be much simpler to analyze than
 those in~\cite{CFMW}.  In particular, for ``linearly binomial networks'' 
 (see Definition~\ref{def:lbin}), deciding multistationarity from the critical function can be done by inspection (Theorem~\ref{thm:NP}). 
We give sufficient conditions for a ``MESSI network'' (see~\cite{messi}) to be a linearly binomial network, and moreover this transformation of the steady-state equations is explicit (Theorem \ref{thm:binom}).  Such networks include many biological signaling pathways.  Additionally, we explain how critical functions are related to discriminants (Proposition~\ref{prop:discriminant}), 
and give conditions that guarantee that triangular-form equations,
as in~(B) above, exist (Corollary~\ref{cor:tri}). 
This last result relies on prior work on the specialization of 
Gr\"obner bases.
Finally, we illustrate our results on a number of reaction networks arising in biology.
Indeed, our results allow us to investigate multistationarity in biological networks systematically, where previously only ad-hoc methods could be applied.

Our results fit in the context of recent progress
on the problems of deciding multistationarity (reviewed in~\cite{mss-review};
see also \cite{BP,invitation}), obtaining witnesses for multistationarity (e.g., \cite{messi,TSS}), 
and characterizing parameter regions for multistationarity (e.g., \cite{FAM2018, CFMW,  MFA2018, KathaMulti, Smallest,WangSontag}).  
Indeed, we give new criteria for multistationarity and methods for witnesses, 
and also show that the multistationary parameter regions arising from degree theory that we describe are
open sets and thus full-dimensional (Theorem~\ref{thm:openmap}).

As mentioned above, one of our criteria for multistationarity involves examining 
the determinant of the Jacobian matrix.  The first such criterion (without composing with a 
steady-state parametrization) was given by Craciun and Feinberg~\cite{ME1} in the absence of conservation relations, 
and then was extended by many researchers~(e.g., \cite{BP,ME3,signs,ShinarFeinberg2012,WiufFeliu_powerlaw}). %  ME_entrapped,ME2, JiThesis, Simplifying
One version of such a result, a so-called {\em injectivity criterion}, says: If every term in the determinant 
of the Jacobian matrix has the same sign, then the network is not multistationary for any choice of parameters.  
Also, under some hypotheses, the converse holds~\cite{BP,CFMW,ME1,Feliu-inj}; see also the recent paper~\cite{MHR}.  
Here we prove analogous results, after using steady-state parametrizations.

Steady-state parametrizations have already been shown 
to be useful in analyzing reaction networks~\cite{CFMW,perspective,jmp2018,messi,TSS,TG}, %signs
and we build on those prior works.  
Similarly, like many before us,
we use degree theory to decide multistationarity
(see e.g.~\cite{CFMW, CHW08,enciso-fixed-pts}) 
and develop theory attuned to networks with certain structure
 \cite{Feliu-inj,gnacadja2011reachability,messi}, including binomials \cite{translated,signs,TSS}.

Given that our work harnesses several techniques that have already been used for analyzing 
reaction networks -- steady-state parametrizations, degree theory, and structured reaction networks --
we emphasize that our final results rely on new techniques.
Specifically, we introduce the notion of effective parameters (Definition~\ref{def:effective}) to simplify previous approaches for a common class of 
structured networks, and we use results on specializations of Gr\"obner bases~\cite{SunYao2010}.

% NONDEGENERACY
Finally, our work is related to the following open question:
If a network $G$ admits a positive steady state that is degenerate, does this guarantee that $G$ admits multiple positive steady 
states?  (This question is related to the Nondegeneracy 
Conjecture~\cite{Joshi:Shiu:Multistationary}.)
One might hope that perturbing the parameters, i.e., rate constants and conservation-law values, would break apart the degenerate 
steady state into two or more steady states.  
Several prior results answered the above question, under some hypotheses, 
in the affirmative~\cite{CFMW,ME1, Feliu-inj,FSW,shiu-dewolff}.
Some of these results also yield procedures for generating a witness to 
multistationarity.
Here, we add new results 
to this list in Theorems~\ref{thm:c-general},~\ref{thm:NP}, and~\ref{prop:triangular}; see also Procedure~\ref{proc:witness} for linearly binomial networks.

The outline of our work is as follows.  
In Section~\ref{sec:background}, we introduce 
mass-action kinetics systems and recall a well-known result about Newton polytopes.  
In Sections~\ref{sec:degree} and~\ref{sec:binomial}, we consider networks that admit steady-state parametrizations. We show that for linearly binomial networks the effective parameters are recovered from the steady states and there are only $s$ of them, where $s$ is the number of species variables, which is a number typically much smaller than the number of parameters. We prove
 that multistationarity 
is guaranteed (by degree theory) when the critical function changes sign.
In Section~\ref{sec:messi}, we consider so-called MESSI networks (which describe Modifications of type Enzyme-Substrate or Swap with Intermediates)~\cite{messi}, which include many biological signaling networks.
We give sufficient conditions for 
such networks to be linearly binomial, which generalize Example~\ref{ex:CFMW611} (continued along the paper) and Example~\ref{cascade}.
 In Section~\ref{sec:triangular}, we consider networks whose
 steady-state equations can be replaced by equivalent triangular-form equations.
We give sufficient conditions for a degenerate steady state of 
such a network to break into multiple steady states and
we show that triangular-form equations exist under general conditions.   
We end with a Discussion in Section~\ref{sec:discussion}. Finally, we include,  in three Appendices, proofs of those results which require further background.
% BACKGROUND	
\section{Background} \label{sec:background}
In this section, we introduce reaction networks and their mass-action kinetics systems (Section~\ref{sec:networks}), 
and then recall a useful result pertaining to Newton polytopes (Section~\ref{sec:NP}).

\subsection{Reaction networks} \label{sec:networks}
Here we largely follow the notation of Conradi, Feliu, Mincheva, and Wiuf~\cite{CFMW}.
A {\em reaction network} $G$ consists of a set of $s$ species $\{X_1, X_2, \ldots, X_s\}$ and a set of $m$ reactions:
\[
%R_j: 
\sum_{i=1}^s\alpha_{ij}X_i\rightarrow \sum_{i=1}^s\beta_{ij}X_i, \quad \quad
	j=1,2, \ldots, m~,
\]
where $\alpha_{ij}$ and $\beta_{ij}$ are non-negative integers. The {\em stoichiometric matrix} of 
$G$, 
denoted by $N$,  is the $s\times m$ matrix with
$(i, j)$-entry equal to $\beta_{ij}-\alpha_{ij}$.
Let $im(N)^{\perp}$ denote the orthogonal complement of the image of the stoichiometric matrix $N$, and let $d=s-{\rm rank}(N)$. 
A {\em conservation-law matrix} of $G$, denoted by $W$, is any row-reduced $d\times s$-matrix whose rows form a basis of $im(N)^{\perp}$.
If all entries of $W$ are nonnegative and 
every column of $W$ contains at least one nonzero entry, that is, 
every species occurs with a positive coefficient in at least one conservation law, then $G$ is {\em conservative}.

The concentrations of the species $X_1,X_2, \ldots, X_s$ are denoted by $x_1, x_2, \ldots, x_s$, respectively. 
The evolution of the concentrations with respect to time is given by a system of ordinary differential equations:
\begin{equation}\label{sys}
\dot{x}~=~f(x)~:=~N\cdot v(x)~, 
\end{equation}
where $x=(x_1, x_2, \ldots, x_s)$ and
$v:{\mathbb R}_{\geq 0}^s\rightarrow {\mathbb R}_{\geq 0}^{m}$ is a 
{\em reaction rate function}.  This function, in the case of {\em mass-action kinetics}, 
 is given by:
\[
v_j(x)=\kappa_j \, x_1^{\alpha_{1j}} 
		x_2^{\alpha_{2j}} 
		\cdots x_s^{\alpha_{sj}},\;\;\;j=1,2,\ldots, m,
\]
where $\kappa_j \in \mathbb R_{>0}$ is called a {\em reaction rate constant}. 
 We can consider the reaction rate constants as parameters $\kappa=(\kappa_1, \dots, \kappa_m)$ and 
 view the polynomials $f_{\kappa,i} \in \mathbb Q[\kappa,x]$, for $i=1, \dots, s$.
For ease of notation we will sometimes simply write $f_i$ in place of $f_{\kappa,i}$, for $i=1, \dots, s$.

Our system~\eqref{sys} satisfies $W \dot x = Wf_\kappa(x)=0$,  and both the positive orthant $\mathbb R_{>0}^s$ and its closure $\mathbb R_{\ge 0}$ are forward-invariant for the dynamics. Thus, 
a trajectory $x(t)$ beginning at a nonnegative vector $x(0)=x^0 \in
\mathbb{R}^s_{> 0}$ remains, for all positive time,
 in the following {\em stoichiometric compatibility class} with respect to the {\em total-constant vector} $c:= W x^0 \in {\mathbb R}_{>0}^d$:
\[
\scc_c~:=~ \{x\in {\mathbb R}_{\geq 0}^s \mid Wx=c\}~,
\]
that is, $\scc_c$ is also forward-invariant with
respect to the dynamics~\eqref{sys}.
A {\em steady state} of~\eqref{sys} is a nonnegative concentration vector 
$x^* \in \mathbb{R}_{\geq 0}^s$ at which the ODEs~\eqref{sys}  vanish: $f_\kappa (x^*) = f(x^*) =0$.  
%A steady state $x^*$ is {\em nondegenerate} if ${\rm Im}\left( df_{\kappa} (x^*)|_{S} \right) = \St$. 
%(Here, $df_{\kappa}(x^*)$ is the Jacobian matrix of $f_{\kappa}$ at $x^*$.)  
%A nondegenerate steady state is 
%{\em hyperbolic} if each of the $\sigma:= \dim (\St)$ nonzero eigenvalues of $df_{\kappa}(x^*)$ has nonzero real part and is 
%{\em exponentially stable} if each of the $\sigma:= \dim(\St)$ nonzero eigenvalues of $df_{\kappa}(x^*)$ has negative real part. 
We distinguish between {\em positive steady states} $x ^* \in \mathbb{R}^s_{> 0}$ and {\em boundary steady states} 
$x^*\in {\mathbb R}_{\geq 0}^s\backslash {\mathbb R}_{>0}^s$.

In order to analyze steady states within a stoichiometric compatibility class, we 
use the conservation laws 
in place of linearly dependent steady-state equations, as follows.
Recall that the conservation-law matrix $W$ is row-reduced.  
Let $I = \{i_1, \dots, i_d\}$ be the indices of the first nonzero coordinate of the rows of $W$,
 and assume that  $i_1<i_2<\cdots<i_d$.
Define %a ${\mathcal C}^{1}$-function 
the function $\F: {\mathbb R}_{\geq 0}^s\rightarrow {\mathbb R}^s$ by 
\begin{equation}\label{consys}
f_{c,\kappa,i} =\F(x)_i=
\begin{cases}
f_{\kappa,i}(x)&~\text{if}~i\not\in I,\\
(Wx-c)_k &~\text{if}~i~=~i_k\in I .
\end{cases}
\end{equation}
This particular choice is needed for the validity of Theorem~\ref{thm:c-general} below.
We refer to system~\eqref{consys}
as the system (\ref{sys}) {\em augmented by conservation laws}. 
For a rate-constant vector 
$\kappa \in \mathbb{R}^{m}_{>0}$ and a total-constant vector $c  \in \mathbb{R}^{d}_{>0}$, 
we say $x^*\in {\mathbb R}_{\geq 0}^s$ is a {\em steady state of the network for $\kappa$ and $c$} if it is a root of the augmented system $\F(x^*)=0$.
Such a steady state $x^*$ is {\em nondegenerate} if 
 the Jacobian matrix of $\F$ at $x^*$ has full rank (namely, equal to $s$).

A {\em multistationary} network admits two or more positive steady states for some rate-constant 
vector $\kappa$ and total-constant vector $c$.  
 Non-multistationary networks are {\em monostationary}.

\begin{example}[Phosphorylation/dephosphorylation of two substrates]\label{ex:CFMW611}
Consider the following network, which is from \cite[\S6.1 in Supplementary Information]{CFMW}:
\begin{center}
\begin{tabular}{cc}
\ce{A + K
<=>[\kappa_1][\kappa_2]
AK
->[\kappa_3]
A_p + K
},
&
\ce{
A_p + F
<=>[\kappa_4][\kappa_5]
A_pF
->[\kappa_6]
A  + F
},\\
\ce{B + K
<=>[\kappa_{7}][\kappa_{8}]
BK
->[\kappa_9]
B_p + K
},
&
\ce{
B_p + F
<=>[\kappa_{10}][\kappa_{11}]
B_pF
->[\kappa_{12}]
B + F
}.\\
\end{tabular}
\end{center}
In this network, phosphorylation and dephosphorylation of two substrates $A$ and $B$
are catalyzed by a kinase $K$ and a phosphatase $F$, respectively.

Following~\cite{CFMW}, let
\begin{center}
\begin{tabular}{lllll}
$X_1$=\ce{K},& $X_3$=\ce{A},& $X_5$=\ce{B},&$X_7$=\ce{AK}, & $X_9$=\ce{A_pF},\\
$X_2$=\ce{F},& $X_4$=\ce{A_p},& $X_6$=\ce{B_p},                             & $X_8$=\ce{BK},& $X_{10}$=\ce{B_{p}F}.\\
\end{tabular}
\end{center}
The system evolves according to the ODEs $\dot x = f_\kappa(x)$, where the function $f_\kappa=(f_1, \dots, f_{10})$ arising from mass-action kinetics is as follows:
\begin{center}
\begin{tabular}{ll}
$f_1 =  -\kappa_1x_1x_3 + \kappa_2 x_7 + \kappa_3 x_7 - \kappa_7 x_1x_5 + \kappa_8x_8 + \kappa_9x_8$,  &\\
$f_2 =  -\kappa_4x_2x_4 + \kappa_5x_9 + \kappa_6x_9 - \kappa_{10}x_2x_ 6 + \kappa_{11}x_{10} +\kappa_{12}x_{10}$,  &\\
$f_3 =  -\kappa_1x_1x_3 + \kappa_2 x_7 +  \kappa_6x_9$, &
 $f_4 =   -\kappa_4x_2x_4 + \kappa_3x_7 + \kappa_5x_9$, \\
$f_5 =  - \kappa_7 x_1x_5 + \kappa_8x_8 + \kappa_{12}x_{10}$, &
$f_6 =  -\kappa_{10}x_2x_6+\kappa_9x_8+\kappa_{11}x_{10}$, \\
$f_7 =   \kappa_1x_1x_3-\kappa_2x_7-\kappa_3x_7$, &
$f_8 =   \kappa_7x_1x_5-\kappa_8x_8-\kappa_9x_8$, \\
$f_9 =   \kappa_4x_2x_4-\kappa_5x_9-\kappa_6x_9$, &
$f_{10}=   \kappa_{10}x_2x_6-\kappa_{11}x_{10}-\kappa_{12}x_{10}$. 
\end{tabular}
\end{center}
Letting $c_1, c_2, c_3,c_4$ denote the 
total amounts of $K$, $F$, $A$, and $B$, respectively, then
the conservation laws are:
\[
 x_1 + x_7 + x_8 ~=~ c_1~, \quad 
  x_2  + x_9 + x_{10} ~=~ c_2~, \quad 
 x_3 + x_4 + x_7 + x_9 ~=~ c_3~, \quad 
x_5+ x_6 + x_8 + x_{10} ~=~ c_4~.
\]
Every $x_i$ appears in at least one conservation law and all have nonnegative coefficients, so the network is conservative. 
The resulting conservation-law matrix $W$ is:
\begin{equation*}
\left(
\begin{array}{cccccccccc}
1 & 0 & 0& 0& 0 & 0 & 1& 1& 0& 0 \\
0 & 1 & 0& 0& 0 & 0 & 0& 0& 1& 1\\
0 & 0 & 1& 1& 0 & 0 & 1& 0& 1& 0 \\
0 & 0 & 0& 0& 1 & 1 & 0 &1 & 0 & 1
\end{array}
\right),
\end{equation*}
which is already in row-reduced form, and 
the indices of the first nonzero coordinate of the rows 
are $1, 2, 3$, and $5$. 
So, by~(\ref{consys}), the function $\F=\F(x)$ is as follows:
\begin{center}
\begin{tabular}{ll}
$f_{c, \kappa,1} =  x_1 + x_7 + x_8 - c_1$,  &
$f_{c, \kappa,2} =  x_2  + x_9 + x_{10} - c_2$,  \\
$f_{c, \kappa,3} =   x_3 + x_4 + x_7 + x_9 - c_3$, &
 $f_{c, \kappa,4} =   -\kappa_4x_2x_4 + \kappa_3x_7 + \kappa_5x_9$, \\
$f_{c, \kappa,5} =   x_5+ x_6 + x_8 + x_{10} - c_4$, &
$f_{c, \kappa,6}=  -\kappa_{10}x_2x_6+\kappa_9x_8+\kappa_{11}x_{10}$, \\
$f_{c, \kappa,7} =   \kappa_1x_1x_3-\kappa_2x_7-\kappa_3x_7$, &
$f_{c, \kappa,8} =   \kappa_7x_1x_5-\kappa_8x_8-\kappa_9x_8$, \\
$f_{c, \kappa,9} =   \kappa_4x_2x_4-\kappa_5x_9-\kappa_6x_9$, &
$f_{c, \kappa,10}=   \kappa_{10}x_2x_6-\kappa_{11}x_{10}-\kappa_{12}x_{10}$. 
\end{tabular}
\end{center}
\end{example}

\subsection{Newton polytopes} \label{sec:NP}
Consider a 
real, multivariate polynomial
\begin{align} \label{eq:poly}
	f ~=~  b_1 x^{\sigma_1} +  b_2 x^{\sigma_2} + \dots + b_{\ell} x^{\sigma_{\ell}} 
		~\in~ \mathbb{R}[x_1,x_2,\dots, x_s]~,
\end{align}
where  $\sigma_i \in \mathbb{Z}^s$ are distinct,  $x=(x_1, x_2, \ldots, x_s)$, $x^\sigma =\prod_{j=1}^s x_j^{\sigma_{ij}}$,  
 and 
we have $ b_i \neq 0$ for all $i$.
The {\em Newton polytope} of $f$ is the convex hull of its exponent vectors:
\begin{align*}
	{\rm NP}(f) ~:=~ {\rm conv} \{\sigma_1,~ \sigma_2,~ \dots~ ,~ \sigma_{\ell}\} ~\subseteq~ \mathbb{R}^s~.
\end{align*}

We will use the following well-known lemma.
\begin{lemma} \label{lem:NP}
For a
real, multivariate polynomial $f$ as in~\eqref{eq:poly},
if $\sigma_i$ is a vertex of NP$(f)$,
then  there exists $x^* \in \mathbb{R}^s_{>0}$ such that $f(x^*)$ and $ b_{i}$ have the same sign.
\end{lemma}
Indeed, as $\sigma_i$ is a vertex of NP$(f)$, there exists an integer vector $\eta\in \mathbb{Z}^s$ such that the inner product $\eta \cdot x$ is maximized over NP$(f)$ only at $\sigma_i$.
Consider the univariate polynomial
$g(\lambda) \, = f(\lambda^{\eta_1}, \dots, \lambda^{\eta_s})$.
The degree of $g$ equals $\eta \cdot \sigma_i$ and thus it is clear that when
$\lambda$ tends to $+ \infty$, the sign of $g$ agrees with the sign of its leading coefficient
$b_i$. This result is not true for non-vertex exponents, as the univariate polynomial
$\lambda^2 - 2 \lambda + 1$ shows.

% SECTION: degree theory
\section{Establishing multistationarity using degree theory} \label{sec:degree}
In this section we show that 
determining whether certain networks are multistationary 
is equivalent to checking whether their ``critical functions'' change sign
(Theorem~\ref{thm:c-general}). 
Such critical functions arise from steady-state parametrizations (Definition~\ref{def:parametrization-for-h}), 
which generalize those considered by Conradi, Feliu, Mincheva, and Wiuf~\cite{CFMW}.  
Another way we generalize the approach in~\cite{CFMW}
is by working in terms of certain ``effective parameters'' rather than the original rate constants.
Accordingly, our statements extend results in~\cite{CFMW} and our proofs rely heavily on their arguments (which use degree theory).
We show in Examples~\ref{erk3} and~\ref{CFMW612},
that our critical functions can be simpler and easier to analyze than those from~\cite{CFMW}.

\subsection{Parametrizations and critical functions} \label{sec:param}

We will consider simplified versions of $\F$ obtained by linear operations.  
To motivate these ``simplified versions'', consider from Example~\ref{ex:CFMW611} one coordinate of $\F(x)$:
\[
f_7 \quad = \quad   \kappa_1\, x_1x_3-(\kappa_2 + \kappa_3)\, x_7~.
\]
We will replace $f_7$ by 
$\frac{1}{\kappa_1} f_7 \, =\, x_1x_3-\left( \frac{\kappa_2 + \kappa_3}{\kappa_1} \right) x_7$,  and then 
view $a_1=\frac{\kappa_2 + \kappa_3}{\kappa_1}$
as a new ``effective parameter''. So we will consider the rational function $a_1=a_1(\kappa)$ 
instead of the three original parameters $\kappa_1, \kappa_2, \kappa_3$ \footnote{In this example, 
$\frac{\kappa_2 + \kappa_3}{\kappa_1}$ is what is  referred to as a $K_m$-value in structural biochemistry.}.  We must require that any resulting 
reparametrization, which replaces the parameters $\kappa_i$, for $ i =1, \dots,m,$ by  new  parameters $a_j$, for $j=1, \dots, \bar m,$ which are rational functions of them,
is surjective (from $\mathbb R_{>0}^m$ to $\mathbb R_{>0}^{\bar m}$).  
 This surjectivity implies that $\bar m \leq m$, and in fact 
the inequality ${\bar m}<m$ holds
in Examples \ref{CFMW612}, \ref{erk}, and \ref{cascade}.  Indeed,
one of our motivations for introducing effective parameters is to reduce the 
number of parameters and thereby simplify the system.

 %In other words, a certain map 
%$\mathbb{R}^m_{>0} \to \mathbb{R}^{\widetilde m}_{>0}$ of the form 
%$(\kappa_i) \mapsto (a_i)$ must be surjective.  

Next we present an important example to further motivate and clarify the notions of {\em effective parameters}
and {\em effective steady-state function} 
(which will be defined in Definition~\ref{def:effective}).
 The network below in Example \ref{erk} underlies ERK regulation by dual-site phosphorylation by the kinase MEK 
(denoted by $E$) and dephosphorylation by the phosphatase MKP3 ($F$). 
Rubinstein, Mattingly, Berezhkovskii, and Shvartsman 
showed that this network is multistationary, and found witnesses to multistationarity by sampling parameters~\cite{long-term}.  
In this paper we take a more systematic approach to deciding multistationary and to finding a witness to multistationarity,
via a simplified system $\augmentH(x)$ (see Definition \ref{def:effective} and Theorem \ref{thm:c-general}).  

\begin{example}[ERK network]\label{erk}
Consider the following network from~\cite{long-term}:
\begin{center}
\begin{tabular}{cc}
\ce{S_{00} + E
<=>[\kappa_1][\kappa_2]
S_{00}E
->[\kappa_3]
S_{01}E
->[k_{cat}]
S_{11} + E
},
&
\ce{
S_{11} + F
<=>[l_1][l_2]
S_{11}F
->[l_3]
S_{10}F
->[l_{cat}]
S_{00} + F},\\
\ce{S_{01}E
<=>[k_{off}][k_{on}]
S_{01} + E
},
&
\ce{
S_{10} F
<=>[l_{off}][l_{on}]
S_{10} + F
},\\
\ce{S_{10} + E
<=>[m_2][m_1]
S_{10}E
->[m_3]
S_{11} + E
},
&
\ce{
S_{01} + F
<=>[n_1][n_2]
S_{01}F
->[n_3]
S_{00} + F}.
\end{tabular}
\end{center}
%\begin{align}
%\label{eq:erk-1}
%S_{00}+E &\lra S_{00}E \to S_{01}E 
%\xrightarrow{k_{\rm cat}}
 %S_{11}+E \quad 
%& S_{11}+F & \lra S_{11}F \to S_{10}F \xrightarrow{l_{\rm cat}} S_{00}+F                         
  %                          \\
%\label{eq:erk-2}
%S_{01}E  & \overset{k_{\rm off}}{\rightleftarrows}  S_{01}+E                   \quad                                                                       
%& S_{10}F &  \overset{l_{\rm off}}{\rightleftarrows}  S_{10}+F  \\
%\label{eq:erk-3}
%S_{10} + E &\lra S_{10}E {\to}  S_{11}+E   
%& S_{01} + F & \lra S_{01}F  \to S_{00}+F~.
%\end{align}

The rate constants of the network are as follows:
\[\kappa=(\kappa_1, \kappa_2, \kappa_3, k_{cat}, k_{on}, k_{off}, 
m_1, m_2, m_3, l_1, l_2, l_3, l_{cat}, l_{on}, l_{off}, n_1, n_2, n_3)\in {\mathbb R}_{>0}^{18}~.\]
Also, we have $s=12$ species: 
\begin{center}
\begin{tabular}{llllll}
$X_{1}=\ce{S_{00}}$, & $X_3$=\ce{F},& $X_5$=\ce{S_{10}F},&$X_{7}$=
\ce{S_{01}E},  & $X_9$=\ce{S_{01}}, & $X_{11}$=\ce{S_{00}E},\\
$X_2$=\ce{E},& $X_4$=\ce{S_{11}F},& $X_6$=\ce{S_{01}F},                            
 & $X_8$=\ce{S_{10}E},& $X_{10}$=\ce{S_{10}}, &$X_{12}=\ce{S_{11}}$.\\
\end{tabular}
\end{center}

%The function $f(x)$ is 
%\begin{center}
%\begin{tabular}{ll}
%$f(x)_1 = \kappa_1x_7x_{11}-\kappa_2x_1-\kappa_3x_1$,  & \\
%$f(x)_2 =  k_{on}x_7x_9+{\kappa}_3x_1-k_{cat}x_2-k_{off}x_2$,  &\\
%$f(x)_3 =  m_2x_7x_{10}-m_1x_3-m_3x_3$, &\\
% $f(x)_4 = l_1x_8x_{12}-l_2x_4-l_3x_4 $, & \\
%$f(x)_5 = l_{on}x_8x_{10}+l_3x_4-l_{cat}x_5-l_{off}x_5 $, & \\
%$f(x)_6 =  n_1x_8x_9-n_2x_6-n_3x_6$, &\\
%$f(x)_7 =  -n_1x_8x_{9}-l_{on}x_8x_{10}-l_1x_8x_{12}+l_2x_4+l_{cat}x_5+l_{off}x_5+n_2x_6+n_3x_6$, & \\
%$f(x)_8 =  -{\kappa}_1x_7x_{11} - k_{on}x_7x_9 - m_2x_7x_{10} + {\kappa}_2x_1 + k_{cat}x_2 + k_{off}x_2 + m_1x_3 + m_3x_3$, &\\
%$f(x)_9 =  -k_{on}x_7x_9 - n_1x_8x_9+k_{off}x_2 + n_2x_6$, & \\
%$f(x)_{10} = -l_{on}x_8x_{10}-m_2x_7x_{10} + m_1x_3 + l_{off}x_5$,  &\\
%$f(x)_{11} = -\kappa_1x_7x_{11}+\kappa_2x_1 + l_{cat}x_5 + n_3x_6$, & \\
%$f(x)_{12} =-l_1x_8x_{12}+k_{cat}x_2 + m_3x_3 + l_2x_4$.
%\end{tabular}
%\end{center}

From the $d=3$ conservation laws, 
which 
arise from the total amounts of substrate $S$, kinase $E$, and phosphatase $F$, respectively:
\begin{align}
\notag
x_{1}+x_4+x_5+x_6+x_7+x_8 +x_9+x_{10}+x_{11}+x_{12} ~&=~ c_1~,\\
\label{eq:cons-erk}
x_2 +x_7+x_8+ x_{11}~&=~ c_2 ~,\\
 \notag
x_3+ x_4+x_5+x_6 ~&=~ c_3~,
\end{align}
we obtain $\F(x)$: %{\color{red} Should equations 1, 2, 3 be conservation equations?}
\begin{center}
\begin{tabular}{ll}
$f_{c, \kappa,1} =x_{1}+x_4+x_5+x_6+x_7+x_8 +x_9+x_{10}+x_{11}+x_{12} -c_1$, & 
$f_{c, \kappa,2} =  x_2 +x_7+x_8+x_{11}-c_2 $, \\
$f_{c, \kappa,3} =    x_3+ x_4+x_5+x_6 -c_3$,&
 $f_{c, \kappa,4}=  l_1x_3x_{12}-l_2x_4-l_3x_4 $, \\
$f_{c, \kappa,5} = l_{on}x_3x_{10}+l_3x_4-l_{cat}x_5-l_{off}x_5 $, & 
$f_{c, \kappa,6} =  n_1x_3x_9-n_2x_6-n_3x_6$,\\
$f_{c, \kappa,7} =   k_{on}x_2x_9+{\kappa}_3x_{11}-k_{cat}x_7-k_{off}x_7$, & 
$f_{c, \kappa,8} = m_2x_2x_{10}-m_1x_8-m_3x_8$,
\end{tabular}
\end{center}
\begin{tabular}{l}
$f_{c, \kappa,9} =  -k_{on}x_2x_9 - n_1x_3x_9+k_{off}x_7 + n_2x_6$, \\
$f_{c, \kappa,10}= -l_{on}x_3x_{10}-m_2x_2x_{10} + m_1x_8 + l_{off}x_5$,  \\
$f_{c, \kappa,11} =   \kappa_1x_2x_{1}-\kappa_2x_{11}-\kappa_3x_{11}$, \\
 $f_{c, \kappa,12} =-l_1x_3x_{12}+k_{cat}x_7 + m_3x_8 + l_2x_4$.
\end{tabular}

\bigskip

We introduce the following choice of $13$ {\em effective parameters}:
\noindent
%\begin{center}
\begin{equation}\label{eq:baraerk}
\begin{tabular}{lllllll}
$\bar a_1=\frac{l_{cat}}{k_{cat}}$,& $\bar a_2=\frac{m_{3}}{l_{cat}}$,& $\bar a_3=\frac{l_3}{l_{cat}}$,
&$\bar a_4=\frac{n_3}{k_{cat}}$, & $\bar a_5=\frac{\kappa_3}{k_{cat}}$, & $\bar a_6=\frac{m_3}{l_{on}}$, & $\bar a_7=\frac{l_{off}}{l_{on}}$,
 \\
$\bar a_8=\frac{n_{3}}{k_{on}}$,& $\bar a_9=\frac{k_{off}}{k_{on}}$,&$\bar a_{10}=\frac{\kappa_1}
{\kappa_2+\kappa_3}$, & $\bar a_{11}=\frac{m_2}{m_1+m_3}$, & $\bar a_{12}=\frac{l_1}{l_2+l_3}$, &
$\bar a_{13}=\frac{n_1}{n_2+n_3}$. & \\
\end{tabular}
\end{equation}
%\end{center}

\medskip

Note that the resulting map 
$\bar a: \mathbb{R}^{18}_{>0} \to \mathbb{R}^{13}_{>0}$ given by 
$ \kappa\mapsto \bar a(\kappa)$ is surjective. 
\end{example}

%\smallskip
%: SETUP
\begin{definition} \label{def:effective}
Let $G$ be a network with $m$ reactions and $s$ species and let
$\dot x = f_{\kappa}(x)$ denote the resulting mass-action system.  Denote by $W$  a
row-reduced conservation-law matrix and by $I$ the set of indices of the first nonzero
coordinates of its rows, as in \S\ref{sec:networks}. % ~\eqref{consys}
Enumerate the complement of $I$:
	\begin{equation}\label{reindex}
	[s] \setminus I ~=~ \{j_1 < j_2< \dots< j_{s-d}\}~.
	\end{equation} 
% Effective parameters
We say that $\bar a_1(\kappa), \bar a_2(\kappa), \dots, \bar a_{\bar m}(\kappa) \in \mathbb{Q}(\kappa)$ form a set of 
 {\em effective parameters}  
 for $G$ 
if the following hold: 
\begin{enumerate}[(i)]
	\item  $\bar a_i(\kappa^*)$ is defined for all $\kappa^* \in \mathbb{R}^m_{>0}$ and, moreover,  
	$\bar a_i(\kappa^*)>0$
	for every $i=1,2,\dots,\bar m$,
%for every $\kappa \in \mathbb{R}^m_{>0}$ and every $i=1,2,\dots,\widetilde m$ 
%(and hence the denominators of the $a_i$'s are nonzero for every $\kappa \in \mathbb{R}^m_{>0}$),  
	\item the following {\em reparametrization map} is surjective:
\begin{align}\label{eq:ep}
	\bar a ~:~\mathbb{R}^m_{>0} & ~\to~ \mathbb{R}^{\bar m}_{>0} \\ \notag
				\kappa & ~\mapsto~ (\bar a_1(\kappa), \bar a_2(\kappa), \dots, \bar a_{\bar m} (\kappa))~,
\end{align} %and
\item there exists an $(s-d)\times (s-d)$ matrix
	$M(\kappa)$ with entries in 
	$\mathbb{Q}(\kappa):=\mathbb{Q}(\kappa_1,\kappa_2,\dots, \kappa_m)$ such that:
	\begin{enumerate}[(a)]
	\item for all $\kappa^* \in \mathbb{R}^m_{>0}$, the matrix 
$M(\kappa^*)$ is defined 	and, moreover, 
	$\det M(\kappa^*)>0$, and 
	\item letting $(\bar h_{j_{\ell} })$ denote the  functions %polynomials 
	obtained from $(f_{j_\ell})$ via the linear operations defined by $M(\kappa)$, as follows:
\begin{equation}\label{eq:linearchange}
(\bar h_{j_1}, \bar h_{j_2}, \dots, \bar h_{j_{s-d}})^{\top}
\quad := \quad M(\kappa) ~  (f_{j_1}, f_{j_2}, \dots, f_{j_{s-d}})^{\top}~,
\end{equation}
 every nonconstant coefficient in every  $\bar h_{j_{\ell}}$ is equal to   
  a rational-number multiple of some
 % Reason: we want negative-multiples and integer-multiples in some examples.
 $\bar a_i(\kappa)$. In particular, 
$\bar h_{j_{\ell}}$ in $\mathbb{Q}(\kappa)[x]$.
	  \end{enumerate}
\end{enumerate}

Given such an effective parametrization, we consider
for $\ell=1,2,\dots, s-d,$  polynomials
$h_{j_{\ell}}=h_{j_{\ell}} (a,x) \in \mathbb{Q}[a_1,a_2,\dots, a_{\bar m}][x]$ (here, the $a_i$'s are indeterminates)  such that:
\begin{align} \label{eq:effective-reln}
	\bar h_{j_{\ell}} ~=~ h_{j_{\ell}}|_{a_1 = \bar a_1(\kappa), ~\dots~,~ a_{\bar m} = \bar a_{\bar m}(\kappa)}~.
\end{align}
 For $i=1,2, \dots, s$ and any choice of $c \in \mathbb{R}^{d}_{>0}$
and $a \in \mathbb{R}^{\bar m}_{>0}$,
 set
\begin{equation}\label{consys-h}
h_{c,a,i} = \augmentH(x)_i~:=~
\begin{cases}
 h_i(a,x)
&~\text{if}~i  \notin I \\
(Wx-c)_k &~\text{if}~i=i_k\in  I.~
\end{cases}
\end{equation}
We say that the function  $\augmentH: \mathbb{R}^s_{>0} \to \mathbb{R}^s $ is an {\em effective steady-state function} of $G$.
\end{definition}

\begin{remark} \label{rmk:effective}
The choice of effective parameters is not unique. 
In fact, we can simply choose the effective parameters to be the $\kappa_i$'s (that is, $\bar a_i(\kappa^*):=\kappa^*_i$ for all $i$), 
but in our examples we will instead choose the nonconstant coefficients 
(or $\mathbb{Q}$-multiples of them)
 of the $\bar h_{j_{\ell}}$'s.  
 Indeed, there are usually fewer such coefficients than $\kappa_i$'s (see e.g.\ Examples \ref{CFMW612}, \ref{erk}, and \ref{cascade}). %, and \ref{cc}). 
 However, we do not have an algorithm for determining a best choice of effective parameters. 
\end{remark}

\begin{definition}\label{def:lbin}
%We say that a given 
A network is {\it linearly binomial} if there exist binomials $\bar h_{j_1}, \bar h_{j_2}, \dots, \bar h_{j_{s-d}}$ 
and a matrix $M(\kappa)$ as in Definition~\ref{def:effective}, such that equality~\eqref{eq:linearchange} holds.
\end{definition} 
The networks in Examples \ref{CFMW612} and \ref{cascade} are linearly binomial. We will abstract the features that imply this property in Theorem~\ref{thm:binom}.

\begin{remark} \label{rmk:binomial-network-as-in-SF}
Every linearly binomial network is a ``binomial network'' as defined by Sadeghimanesh and Feliu~\cite{SF}. % Def 2.1
\end{remark}

% DEFINITION - parametrization
\begin{definition} \label{def:parametrization-for-h}
Let $G$ be a network with $m$ reactions, $s$ species, and
row-reduced conservation-law matrix $W$.  Let $\F$ arise from $G$ and $W$ as in~\eqref{consys}.
Suppose that $\augmentH$ is an effective steady-state function of $G$, as in~\eqref{consys-h},
arising from a matrix $M(\kappa)$, as in~\eqref{eq:linearchange}, 
a reparametrization map $\bar a$, as in~\eqref{eq:ep}, 
and polynomials $h_{j_{\ell}}$'s as in~\eqref{eq:effective-reln} as in Definition~\ref{def:effective}.

We say that the positive steady states of $G$ {\em admit a positive parametrization with respect to $\augmentH$} 
%and the effective parameters $a$}
 if there exists a function:
	\begin{align} \label{eq:effective} 
	\phi :~ \mathbb{R}^{\hat m}_{>0} \times  \mathbb{R}^{\hat s}_{>0} 
		&~\rightarrow~ 
		 \mathbb{R}_{>0}^{\bar m}\times  \mathbb{R}_{>0}^{s}~\\ \notag
		(\hat a, \hat x) & ~\mapsto~ \phi(\hat a, \hat x)~,
	\end{align}
for some $\hat m \leq \bar m$ and $\hat s \leq s$, 
such that:
	\begin{enumerate}[(i)]
	\item $\phi(\hat a, \hat x)$ extends the vector $(\hat a, \hat x)$. More precisely, 
	there exists a natural projection $\pi: \mathbb{R}^{\bar m}_{>0}\times
	  \mathbb{R}_{>0}^{s}\to \mathbb{R}^{\hat m}_{>0} \times \mathbb{R}^{\hat s}_{>0} $ 
	  such that $\pi \circ \phi$ is equal to the identity map.
	% such that $\pi(\phi(\hat a^*, \hat x^*))=(\hat a^*, \hat x^*)$ for 
	  %all $(\hat a^*, \hat x^*) \in  \mathbb{R}^{\bar m'}_{>0} \times  \mathbb{R}^{s'}_{>0} $), 
	\item Consider any $(a,x) \in \mathbb R^{\bar m}_{>0} \times \mathbb  R_{>0}^s$.  Then, the equality 
		$h_i (a,x) = 0$ holds for every  $i \notin I$ 
		if and only if 
		 there exists
	 $(\hat a^*, \hat x^*) \in  \mathbb{R}^{\hat m}_{>0} \times  \mathbb{R}^{\hat s}_{>0} $ 
 such that $(a, x)=\phi(\hat a^*,\hat x^*)$.
 \end{enumerate}
 \end{definition}

Note that given any $(a,x) \in \mathbb R^{\bar m}_{>0} \times \mathbb  R_{>0}^s$, with $(a,x)=\phi(\hat a, \hat x)$, if we set
\[ c \, = \, W \, x , \]
then $h_{c,a}(x) =0$, where $h_{c,a}$ is the effective steady state function in Definition~\ref{def:effective}.

Moreover,  we can summarize the information in Definition~\ref{def:parametrization-for-h} by asking that the diagram below commutes, where the maps $\mu(\kappa,x):=(\bar a(\kappa),x)$ and $\phi$ are surjective:
\begin{center}
\begin{tikzpicture}[node distance=2cm]
   \node (C) {$\left\{(\kappa^*,x^*) \in  \mathbb{R}_{>0}^{m}\times 
\mathbb{R}_{>0}^{s} ~:~
                 f_i|_{ \kappa=\kappa^*,~x=x^*} = 0 {\rm ~for~all~} i \notin
I\right\}$};
   \node (P) [below of=C] {$\mathbb{R}^{\hat m}_{>0} \times 
\mathbb{R}^{\hat s}_{>0} $ \quad \quad \quad};
   \node (Ai) [right of=P] {\quad \quad \quad $\mathbb{R}^{\bar m}_{>0}
\times  \mathbb{R}^{s}_{>0}  $};
   \draw[->] (C) to node {$\quad \quad$ $\mu$} (Ai);
   \draw[->] (C) to node [swap] {\hskip -1.2cm$\pi \circ \mu $~~~~} (P);
   \draw[<-] (P) to node [above] {\vspace{2cm}  $\phi$} (Ai);
\end{tikzpicture}
\end{center}

\begin{example}[ERK network, continued]\label{erk2}
Let
%{\footnotesize
%\[M(\kappa)=
%\left(\begin{array}{cccccccccccc}
%\frac{1}{\kappa_2+\kappa_3}&  0& 0& 0  & 0&0 & 0&0 & 0\\
% 0&  \frac{1}{k_{cat}}& 0& 0  & 0&  \frac{1}{k_{cat}} & \frac{1}{k_{cat}} &0 & 0\\
 %0 & 0 &  \frac{1}{m_1+m_3}& 0 & 0 & 0 & 0&0&0\\
 %0 &  0 & 0     &  \frac{1}{l_2+l_3} & 0 & 0& 0&0&0\\
 %0 & 0 & \frac{1}{l_{cat}}   & 0   & \frac{1}{l_{cat}} & 0& 0& \frac{1}{l_{cat}}&0\\
 % 0     &  0 & 0   & 0   & 0       &  \frac{1}{n_2+n_3} & 0& 0 &0 \\
   %0&  0& 0& 0  & 0& \frac{1}{k_{on}} &\frac{1}{k_{on}} &0 & 0\\
    %0&  0& \frac{1}{l_{on}}& 0  & 0& 0 &0 &\frac{1}{l_{on}} & 0\\
     %0&  0& \frac{1}{k_{cat}}& \frac{1}{k_{cat}} & \frac{1}{k_{cat}}& 0 &0 &\frac{1}{k_{cat}} & \frac{1}{k_{cat}}
%\end{array}
%\right)~, 
%\]
%}
%and let
%\begin{equation*}
%(h_{1}, h_{2}, h_{3}, h_{4}, h_{5}, h_{6}, h_{9}, h_{10}, h_{12})^T
%\quad := \quad M(\kappa) ~  (f_{1}, f_{2},  f_{3}, f_{4}, f_{5}, f_{6}, f_{9}, f_{10}, f_{12})^T~,
%\end{equation*}
{\footnotesize
\[M(\kappa)=
\left(\begin{array}{cccccccccccc}
\frac{1}{l_2+l_3} &  0& 0& 0  & 0&0 & 0&0 & 0\\
0 & \frac{1}{l_{cat}}  & 0 & 0   & \frac{1}{l_{cat}} & 0& \frac{1}{l_{cat}}&0&0\\
 0 & 0 & \frac{1}{n_2+n_3}& 0 & 0 & 0 & 0&0&0\\
 0 &  0 & \frac{1}{k_{cat}}    &  \frac{1}{k_{cat}}& 0 & \frac{1}{k_{cat}}& 0&0&0\\
  0 &  0 & 0     &  0& \frac{1}{m_1+m_3} & 0& 0&0&0\\
  0     &  0 &  \frac{1}{k_{on}}   & 0   & 0       &   \frac{1}{k_{on}}  & 0& 0 &0 \\
   0&  0& 0& 0  & \frac{1}{l_{on}}& 0 &\frac{1}{l_{on}}&0 & 0\\
    0&  0& 0& 0  & 0& 0 &0 &\frac{1}{\kappa_2+\kappa_3}& 0\\
     \frac{1}{k_{cat}}&  \frac{1}{k_{cat}}& 0& 0 & \frac{1}{k_{cat}}& 0 &\frac{1}{k_{cat}}&0 & \frac{1}{k_{cat}}
\end{array}
\right)~, 
\]
}
It is straightforward to check that 
 ${\rm det} M(\kappa)>0$ for all $\kappa\in \mathbb{R}^{18}_{>0}$. 
From~the effective parameters \eqref{eq:baraerk} and equations \eqref{eq:linearchange}--\eqref{consys-h}, the resulting system $\augmentH(x)$ is:
%Notice that for any fixed $c\in {\mathbb R}_{>0}^{3}$,  there exists $\kappa\in {\mathbb R}_{>0}^{18}$ such that $\F(x)=0$ has two or more positive solutions if and only if there exists $a\in {\mathbb R}_{>0}^{13}$ such that $h_c(x)=0$ has two or more positive solutions, where $h_c(x)$ is defined by 

\bigskip
\noindent
\textcolor{black}{
\begin{tabular}{ll}
$h_{c, a,1} = x_{1}+x_4+x_5+x_6+x_7+x_8 +x_9+x_{10}+x_{11}+x_{12} -c_1$,  & 
$h_{c, a,2} = x_2 +x_7+x_8+ x_{11} -c_2$,  \\
$h_{c, a,3} =x_3+ x_4+x_5+x_6 -c_3$, & 
$h_{c, a,4}  =  \underline{a_{12}}x_3x_{12}-x_4$,  \\
$h_{c, a,5}  =  \underline{a_3}x_4 - x_5 -a_2x_8 $, & 
$h_{c, a,6}  =  \underline{a_{13}}x_3x_{9}-x_6$,  \\
$h_{c, a,7}  =  \underline{a_5}x_{11} -a_4x_6- x_7 $,& 
$h_{c, a,8}  =  \underline{a_{11}}x_2x_{10}-x_8$,\\
$h_{c, a,9} =  \underline{a_9}x_7-x_2x_9 - a_8x_6 $,& 
$h_{c, a,10}  =  \underline{a_7}x_5-x_3x_{10} - a_6x_8$,\\
\end{tabular}
\begin{tabular}{l}
$h_{c, a,11}  =  \underline{a_{10}}x_1x_{2}-x_{11}$,\\
$h_{c, a,12}  = x_7 -\underline{a_1}x_5$.
\end{tabular}
}

\bigskip
\noindent 
Let $\hat a = (a_2, a_4, a_6, a_8)$ and $\hat x = x$.
By solving the 
non-conservation-law 
equations $\augmentH(x)_i=0$ (for $i=4,\ldots,12$), 
for the underlined unknowns $a_1, a_3, a_5, a_7, a_9, a_{10}, a_{11}, a_{12}, a_{13}$, we obtain
 the positive parametrization $\phi: \mathbb{R}^{16}_{>0} \rightarrow  
  \mathbb{R}_{>0}^{13}\times  \mathbb{R}_{>0}^{12}$ with respect to the effective steady-state function $\augmentH$, 
where %for any point $(a_2, a_4, a_6, a_8, x)\in \mathbb{R}^{16}_{>0}$, 
%$u = (a_2, a_4, a_6, a_8, x)$ and 
$\phi(\hat a, \hat x)$ $\left(=\phi(a_2, a_4, a_6, a_8;x)\right)$ 
is defined as 
\textcolor{black}{
\[\left(\frac{x_7}{x_5}, a_2, \frac{a_2x_8+x_5}{x_4}, a_4, \frac{a_4x_6+x_7}{x_{11}}, a_6,
\frac{a_6x_8+x_3x_{10}}{x_5}, a_8, \frac{a_8x_6+x_2x_9}{x_7}, \frac{x_{11}}{x_1x_{2}}, \frac{x_8}{x_2x_{10}}, \frac{x_4}{x_3x_{12}},\frac{x_6}{x_3x_9}; ~x\right).
\]
}
We will use this information in Example~\ref{erk3} below.
\end{example}

We need one more definition in order to state Theorem~\ref{thm:c-general}.

\begin{definition} \label{def:Critical}
Under the notation and hypotheses of Definition~\ref{def:parametrization-for-h}, assume  
that the steady states of $G$  admit a positive parametrization with respect to $\augmentH$.
For such a positive parametrization $\phi$, the {\em  critical function} 
$C: \mathbb{R}^{\hat m}_{>0} \times  \mathbb{R}^{\hat s}_{>0}  \rightarrow   \mathbb{R}$ is given by:
\begin{equation*}
 C(\hat a, \hat x) \quad  = \quad \left(\det {\rm Jac}(\augmentH)\right)|_{(a, x)=\phi(\hat a, \hat x)}~,
\end{equation*}
where ${\rm Jac}(\augmentH)$ denotes the Jacobian matrix of $\augmentH$ with respect to $x$. 
\end{definition}

\begin{remark} \label{rmk:degenerate-C-tilde}
Assume $(a^*,x^*)=\phi(\hat a^*, \hat x^*)$. 
It follows that {\em  $C(\hat a^*, \hat x^*)=0$ if and only if $x^*$ is a degenerate steady state of any mass-action
 system defined by network $G$ for the total-constant vector $c^*=Wx^*$ and a choice of rate-constant vector $\kappa^*$ for which $a^*= \bar a(\kappa^*)$}. 
 Here, $\bar a(\kappa^*)$ refers to the map~\eqref{eq:ep}.
\end{remark}

\begin{remark}[Comparison with~\cite{CFMW}] \label{rmk:CFMW}
The parametrizations considered by Conradi, Feliu, Mincheva, and Wiuf 
did not allow the reaction rates $\kappa_i$ to depend on the $x_i$'s~\cite{CFMW}.
Specifically, their parametrizations
have the form
  $(\kappa,\hat x) \mapsto (\kappa,(\hat x,\Phi(\kappa, \hat x))$.  The resulting critical functions
  are denoted by ``$a(\hat x)$" in their work (notice that ``$a(\hat x)$" also depends on $\kappa$, and the ``$a$" is a different notion from our $a$ in $\augmentH$); see Examples~\ref{erk3} and \ref{CFMW612}.  
  Additionally, the critical functions in~\cite{CFMW} 
 simply arise from the case when the matrix $M(\kappa)$
in \eqref{eq:linearchange}
 is the identity matrix (and the 
 ``reparametrization'' map $\bar a$,
 in \eqref{eq:ep}, is the identity map).
 In summary, 
we allow for more general positive parametrizations and critical functions
than in~\cite{CFMW}, and hence Theorem~\ref{thm:c-general} below generalizes~\cite[Theorem 1]{CFMW}; our proof is just a translation of their arguments to our setting.  
\end{remark}

\begin{remark}[Existence of steady-state parametrizations] \label{rmk:when-parametrizations}
Steady-state parametrizations exist for many biological signaling
 networks~\cite{feliu-wiuf-ptm,messi,TG}.
 They can be computed %by hand, 
 by following one of the procedures in the
 references above, or, as suggested in~\cite{CFMW,perspective}, 
 by using computer-algebra software
to solve the steady-state equations for all but $d$ variables (see, e.g., 
Example~\ref{CFMW612}), eliminating if possible all
  intermediates.   (Here $d$ is the number of conservation laws.)
\end{remark}

A network is {\em dissipative} if for all choices of rate constants and stoichiometric compatibility classes $\scc_c$, there exists a compact subset of $\scc_c$
which every trajectory beginning in 
$\scc_c$
eventually enters.   Every conservative network is dissipative~\cite[pg.\ 6]{CFMW}.  
%---------------------------------
% THEOREM: parametrization & degree theory result
%---------------------------------
\begin{theorem} \label{thm:c-general}
Under the notation and hypotheses of 
Definitions~\ref{def:parametrization-for-h} and~\ref{def:Critical},
 assume also that $G$ is a dissipative network without boundary steady states in any compatibility class.
%Then: %if $\widetilde C$ changes sign,then $G$ is multistationary.  
\begin{enumerate}[(A)]
% MULTI
\item  {\bf Multistationarity.}
 $G$ is multistationary 
if there exists 
$(\hat a^*, \hat x^*) \in 
\mathbb{R}^{\hat m}_{>0} \times  \mathbb{R}^{\hat s}_{>0} $ 
%$u^* \in \mathbb{R}^T_{>0}$
such that 
\[
{\rm sign}(C(\hat a^*, \hat x^*)) = (-1)^{\mathrm{rank}(N)+1}~,
\]
where 
 $N$ denotes the stoichiometric matrix of $G$.
% WITNESS
\item  {\bf Witness to multistationarity.}
 Every 
 $(\hat a^*, \hat x^*) \in 
\mathbb{R}^{\hat m}_{>0} \times  \mathbb{R}^{\hat s}_{>0} $
 %$u^* \in \mathbb{R}^T_{>0}$ 
 with 
	${\rm sign} (C(\hat a^*, \hat x^*))=
	(-1)^{\mathrm{rank}(N)+1}$ 
		yields 
 a witness
to multistationarity $(\kappa^*, c^*)$ as follows.  
Let $(a^*,x^*)=\phi(\hat a^*, \hat x^*)$.  
Let $c^* = W x^*$
(so, $c^*$ is the total-constant vector defined by $x^*$,
where $W$ is the conservation-law matrix), and
let $\kappa^* \in \mathbb{R}_{>0}^m$ be such that 
$\bar a (\kappa^*)= a^*$. %under reparametrization map~\eqref{eq:ep}.
%MONO
\item {\bf Monostationarity.}
$G$ is monostationary 
if ${\rm sign}(C(\hat a, \hat x)) = (-1)^{\mathrm{rank}(N)}$ for all 
$(\hat a^*, \hat x^*) \in \mathbb{R}^{\hat m}_{>0} \times  \mathbb{R}^{\hat s}_{>0} $. 
\end{enumerate} 
\end{theorem}

\begin{proof} 
We begin with (A) and (B).  
Assume that 
 $(\hat a^*, \hat x^*) \in \mathbb{R}^{\hat m}_{>0} \times  \mathbb{R}^{\hat s}_{>0} $
satisfies
	${\rm sign} (C(\hat a^*, \hat x^*))=
	(-1)^{\mathrm{rank}(N)+1}$.
That is, 
	${\rm sign}  \left(\det {\rm Jac}(\augmentH)\right)|_{(a, x)=\phi(\hat a^*, \hat x^*)} =
	(-1)^{\mathrm{rank}(N)+1}$.

By \cite[Theorem 1]{CFMW}, $G$ is multistationary if there exist 
$\kappa^* \in \mathbb{R}^m_{>0}$ and $x^* \in \mathbb{R}^s_{>0}$ such that 
	\begin{enumerate}[(i)]
	\item 
$f_i|_{ \kappa=\kappa^*,~x=x^*} = 0$ for all $i \notin I$, and 
	\item 
	 ${\rm sign}  \left(\det {\rm Jac}(\F)\right)|_{\kappa=\kappa^*,x=x^*} =
	(-1)^{\mathrm{rank}(N)+1}$.  
	\end{enumerate}
Moreover, in this case, \cite[Theorem 1]{CFMW} gives the following witness to multistationarity: $c^*=Wx^*$ and $\kappa^*$.  
To use this result, let $(a^*,x^*)=\phi(\hat a^*, \hat x^*)$, and then define 
 $c^* = W x^*$
and pick $\kappa^* \in \mathbb{R}_{>0}^m$ such that 
$\bar a (\kappa^*)= a^*$.  Then we only need to show that conditions (i) and (ii) above hold.
	
To see (i), recall that the positive-determinant matrix $M(\kappa)$ transforms the 
$f_i$'s, for $i \notin I$, to the $\bar h_i$'s, as in~\eqref{eq:linearchange}.  So, we need to show 
that $\bar h_i|_{x=x^*,\kappa=\kappa^*}=0$ if $i \notin I$.  Indeed,
\[
	\bar h_i|_{x=x^*,\kappa=\kappa^*}
	~=~
	h_i|_{x=x^*,a=\bar a(\kappa^*)}
	~=~
	h_i|_{(a,x)=\phi(\hat a^*, \hat x^*)}
	~=~
	0~,		
\]
where the final equality comes from requirement (ii) for a positive parametrization (in Definition~\ref{def:parametrization-for-h}).

For (ii), we first note that~\eqref{eq:linearchange} implies the following:
\begin{align} \label{eq:tilde-M}	
(h_{c,a,1}, h_{c,a,2}, \dots, h_{c,a,s})^{\top}|_{a=\bar a}
\quad = \quad \widetilde M(\kappa) ~  (f_{c,\kappa,1}, f_{c,\kappa,2}, \dots, f_{c,\kappa,s})^{\top}~,
\end{align}
where $\widetilde M(\kappa)$
is the $(s \times s)$-matrix obtained from $M(\kappa)$, in \eqref{eq:linearchange},
 by inserting rows and columns corresponding to the indices $i_1, i_2, \dots, i_d$ such that row $i_k$ 
 and column $i_k$ are both the canonical basis vector $e_{i_k}$, for $k=1,2,\dots, d$. 
 Thus, $\det \widetilde M(\kappa) =\det M(\kappa)>0$ for any $\kappa \in \mathbb{R}^m_{>0}$, and so \eqref{eq:tilde-M} yields:
\begin{align*}
{\rm sign}  \left(\det {\rm Jac}(\F)\right)|_{\kappa=\kappa^*,x=x^*}
	~ &=~
{\rm sign}  \left(\det {\rm Jac}(\augmentH)\right)|_{(a,x)=(\bar a(\kappa^*),x^*)=\phi(\hat a^*, \hat x^*)} \\
	~ &=~
	{\rm sign} (C(\hat a^*, \hat x^*)) \\
	~ &=~
	(-1)^{\mathrm{rank}(N)+1}~,
\end{align*}
where the final equality is by hypothesis. Thus, (ii) holds.

For (C), assume that ${\rm sign}(C(\hat a, \hat x)) = (-1)^{\mathrm{rank}(N)}$ for all 
$(\hat a^*, \hat x^*) \in \mathbb{R}^{\hat m}_{>0} \times  \mathbb{R}^{\hat s}_{>0} $. 
Suppose that $\kappa^* \in \mathbb{R}^m_{>0}$ and $x^* \in \mathbb{R}^s_{>0}$ 
are such that $f_i|_{\kappa=\kappa^*,x=x^*}=0$ for all $i \notin I$.  By~\cite[Theorem 1]{CFMW},
we only need to show the following:
\begin{align} \label{eq:mono-sign}
{\rm sign}  \left(\det {\rm Jac}(\F)\right)|_{\kappa=\kappa^*,x=x^*} =
	(-1)^{\mathrm{rank}(N)}~.
\end{align}
To this end, let $(\hat a^*, \hat x^*)=\pi \circ \mu (\kappa^*, x^*)$, so that 
(using the commutative diagram following Definition~\ref{def:parametrization-for-h}) we have 
$\phi(\hat a^*, \hat x^*) = (\bar a(\kappa^*), x^*)$.  We now verify equality~\eqref{eq:mono-sign}:
\begin{align*}
{\rm sign}  \left(\det {\rm Jac}(\F)\right)|_{\kappa=\kappa^*,x=x^*} 
	~&=~
	{\rm sign}  \left(\det {\rm Jac}(\augmentH)\right)|_{(a,x)=(\bar a(\kappa^*),x^*)=\phi(\hat a^*, \hat x^*)} \\
	~ &=~
	{\rm sign} (C(\hat a^*, \hat x^*)) \\
	~ &=~
	(-1)^{\mathrm{rank}(N)}.
\end{align*}
\end{proof}

%\begin{remark}\label{rmk:conservative}
%{\bf
%In the case of conservative networks (with non-degenerate Jacobian), the existence of at least one positive steady state is assured. 
%}
%\end{remark}

\begin{remark}\label{rmk:vertex}
Theorem~\ref{thm:c-general}(B) suggests a procedure for finding a witness to multistationarity, 
which relies on picking some $(\hat a^*, \hat x^*)$ with 
	${\rm sign} (C(\hat a^*, \hat x^*))=
	(-1)^{\mathrm{rank}(N)+1}$. 
There is no general method for picking such a 
vector 
$(\hat a^*, \hat x^*)$, but 
we can sometimes accomplish this via 
the Newton polytope of $C(\hat a, \hat x)$,
namely, when one of the vertices corresponds to a coefficient of $C(\hat a, \hat x)$ 
with the desired sign $(-1)^{\mathrm{rank}(N)+1}$
(recall Lemma~\ref{lem:NP}).  Indeed, we will
see in Section~\ref{sec:binomial} that 
this approach always succeeds for linearly binomial networks.
\end{remark}

\begin{remark} \label{rmk:stability}
Theorem~\ref{thm:c-general} can guarantee multiple steady states, but does not say anything about whether they are stable
\cite{CFMW}.
Also in that theorem, ``monostationarity''
in part (C) can be strengthened to state the existence of a unique positive steady state in each compatibility class~\cite[Theorem 1]{CFMW}.
\end{remark}

\begin{remark} \label{rmk:M-notation}
We use $M(\kappa)$ to denote, as in~\eqref{eq:linearchange}, the matrix transforming 
the $f_i$'s to the $\bar h_i$'s, 
while the authors of \cite{CFMW} use ``$M(x)$'' to denote the matrix ${\rm Jac}(\F)$, after substituting the parametrization.
\end{remark}

\begin{remark} \label{rmk:choose-M}
Our strategy for constructing a useful matrix $M(\kappa)$ as 
in~\eqref{eq:linearchange}, is to perform linear operations on the $f_i(x)$'s  to obtain as many binomials as possible (see, e.g., Example~\ref{erk}).
Additionally, we will show in Theorem~\ref{thm:binom} that
 for many biological signaling networks,
 a suitable $M(\kappa)$ exists so that in the resulting system $\augmentH$, every non-conservation-law equation is a binomial.
\end{remark}

\begin{remark} \label{rmk:hypotheses}
Theorem~\ref{thm:c-general} relies on degree theory, 
which is why two key hypotheses are required: being dissipative and 
having no boundary steady states.  A discussion on how to verify 
these two hypotheses is in~\cite[pp.\ 11-12]{CFMW}.  In the examples below,
the networks are conservative and thus dissipative, and checking that there are no 
boundary steady states can be done using, for instance, results from \cite{SS2010} or  \cite[Theorem~3.13]{messi}.
\end{remark}

\subsection{Examples} \label{sec:examples}
We present two examples.
In the first one, 
we will  see that the critical function arising from  
$\augmentH$ in Example \ref{erk2} is simpler than 
that  
allowed by Conradi, Feliu, Mincheva, and Wiuf~\cite{CFMW}.
We will also see in this example how to
use Theorem~\ref{thm:c-general} to obtain a witness to multistationarity.

\begin{example}[ERK network, continued]\label{erk3}
%*** Wording is missing here, recalled from Examples~\ref{erk} and~\ref{erk2} ***
%By Definition \ref{def:pp}, we see that the network $G$ admits the positive parametrization $\phi$.  
The critical function $C(\hat a, \hat x)$ derived by the positive parametrization in Example \ref{erk2} (recall 
that $(\hat a, \hat x)=(a_2, a_4, a_6, a_8, x)$) is a rational function, where
the denominator is the 
monomial \textcolor{black}{$x_{1}x_2x_3x_5x_9x_{10}x_{12}$}, 
and the numerator has total degree $11$ and $704$ terms. 

This network is conservative (see~\eqref{eq:cons-erk}) and hence dissipative. 
It is  straightforward to check (for instance, using criteria in \cite{SS2010}) that it has no boundary steady states. In fact, the ERK network is a MESSI network (see the definition in Section~\ref{sec:messi}) and the absence of boundary steady states is a direct and easy consequence of Theorem~3.13 in~\cite{messi}.
So we can apply Theorem~\ref{thm:c-general} to find a witness to  multistationarity. Accordingly, 
we compute ${\rm rank}(N)=s-d=12-3=9$, so the sign of interest is $(-1)^{{\rm rank}(N)+1}=1$. 
Thus, we must find $(\hat a^*, \hat x^*)$ for which 
 $C(\hat a^*, \hat x^*) >0$.    
 We find the following such point by
using a vertex of the Newton polytope of the numerator of $C$
that corresponds to a monomial with positive coefficient (recall Remark~\ref{rmk:vertex}), for instance:
\textcolor{black}{
\[(\hat a^*, \hat x^*) ~=~(a^*_2, a^*_4, a^*_6, a^*_8, x^*)~=~\left(\frac{1}{10}, 10, 10, \frac{1}{10}, \frac{1}{10}, 10, \frac{1}{10}, 10, 10, 10, \frac{1}{10}, 10, 10, 10, 10, \frac{1}{10}\right)~.\]
}
Letting $(a^*,x^*)=\phi(\hat a^*, \hat x^*)$, we have
\[a^*=\left( \frac{1}{100},  \frac{1}{10}, \frac{11}{10}, 10, \frac{1001}{100}, 10, \frac{101}{10}, \frac{1}{10}, 1010, 10, \frac{1}{10}, 1000, 10 \right)~,\]
and
\[c^* ~:=~Wx^*~=~ \left(\frac{703}{10},~~\frac{301}{10},~~\frac{301}{10}\right).\]
Thus, by Theorem~\ref{thm:c-general},
we obtain a witness to multistationarity via any choice of $\kappa^*$ for which $\bar a(\kappa^*) = a^*$, where $\bar a$ is defined in equation \eqref{eq:baraerk}.  
One such $\kappa^*$ is as follows:  
%\textcolor{red}{some rate constants are relatively small...would that be an issue to us?}
\begin{align*}
\kappa^*~&=~
%	(\kappa_1, \kappa_2, \kappa_3, k_{cat}, k_{on}, k_{off}, m_1, m_2, m_3, l_1, l_2, l_3, l_{cat}, l_{on}, l_{off}, n_1, n_2, n_3) \\
%~&=~ 
\left(
%\kappa_{1}^*=
\frac{1001}{10}, \;
% \kappa_{2}^*=
 \frac{1001}{100} , \;
% \kappa_3^*=
 \frac{1001}{100}, \; 
% k_{cat}^*=
1, \;
%  k_{on}^*=
 100, \; 
% k_{off}^*=
101000, \; 
%  m_1^*=
 \frac{1}{1000},\;
 % m_2^*=
 \frac{1}{10000},\;
%  m_3^*=
 \frac{1}{1000}, \; 
 % l_1^*=
 11,\;
 % l_2^*=
 \frac{11}{1000},\;
%  l_3^*=
 \frac{11}{1000}, \;  \right. \\
   & \quad \quad \quad \left.
% l_{cat}^*=
\frac{1}{100},\; 
% l_{on}^*=
\frac{1}{10000}, \;
%  l_{off}^*=
 \frac{101}{100000}, \; 
%  n_1^*=
 200,\;
%  n_2^*= 
 10{\bf,}\;
% n_3^*=
10
\right)~.
\end{align*}
To confirm this witness, using $a^*$ and $c^*$ above, we approximately solve the polynomial system 
$h_{c,a}|_{(a,c)=(a^*, c^*)}=0$ and find seven real solutions for $x$.
Three of these are positive steady states:
\textcolor{black}{
 \begin{align*}
x^{(1)}~&\approx~
(1.18,~2.25,~0.46,~1.54,~1.56,~26.55,~0.016,~1.31,~5.81,~5.81,~26.52,~0.0034)~,
\\
x^{(2)}~&=~(0.1,~ 10,~ 0.1,~  10,~ 10,~ 10,~ 0.1,~ 10,~ 10,~ 10,~ 10,~ 0.1)~, \quad \quad {\rm and}
\\
x^{(3)}~&\approx~( 0.0027,~ 14.56,~ 0.0037,~ 14.87 ,~ 14.85 ,~ 0.38 ,~ 0.15,~ 15.00,~ 10.30,~   10.30,~ 0.39,~ 4.06)~.
\end{align*}
}
With an eye toward comparing our approach to that of~\cite{CFMW}, 
we now obtain another critical function, using the procedures in \cite{CFMW}.
Recall from Remark \ref{rmk:CFMW} that the parametrization in \cite{CFMW} has the form
  $(\kappa,\hat x) \mapsto (\kappa,(\hat x,\Phi(\kappa, \hat x))$.  Applying their method here,
we first solve for 
$9$ of the variables $x_i$ (namely, those with $i=4, \ldots,12$)
in the 9 non-conservation-law equations $f_{c,\kappa,i}=0$, for $i=4,\ldots,12$.  This yields a steady-state parametrization:
 \[
 \psi: \mathbb{R}^{21}_{>0} \rightarrow   \mathbb{R}_{>0}^{18}\times  \mathbb{R}_{>0}^{12}~,\] 
 of the form $\psi( \kappa, x_1, x_2, x_{3}) = (\kappa, x)$.  
Substituting this parametrization
into $\det {\rm Jac}({\F})$ gives
 a critical function which we denote by
$\widetilde{C}(v)$ (here we are writing $v=(\kappa, x_1, x_2, x_3)$).
%Such a function was denoted by ``$a(\hat{x})$" in \cite{CFMW}.  
The denominator of $\widetilde{C}(v)$ is a polynomial in $\mathbb{Q}[v]$
that is positive 
for every $\kappa \in \mathbb{R}^{18}_{>0}$ and every 
$(x_1, x_2, x_{3}) \in \mathbb{R}^3_{>0}$, and has total degree $33$ and $1750$ terms.
The numerator, also a polynomial 
in $\mathbb{Q}[v]$, has total degree $47$ and $246232$ terms. 
%In comparison, the denominator of our $C(x)$ is simply $x_5x_7x_8x_9x_{10}x_{11}x_{12}$
%and the numerator of $C(x)$ has total degree $11$ and $704$ terms. 

As seen in Table~\ref{tab:compare},
which compares the critical functions $C$ and $\widetilde{C}$, our critical function $C$ is simpler 
(fewer variables, with both numerator and denominator having lower degree and fewer terms).  
Moreover, the fact that the denominator of $C$ is a monomial makes it apparent that this denominator 
is positive on the positive orthant.  In contrast, the denominator of $\widetilde{C}$ has 1750 terms (which nonetheless are all positive).

% TABLE: COMPARISON
\begin{table}%[ht]
\renewcommand{\arraystretch}{1.3}
\centering
\begin{tabular}{@{}llll@{}}
\toprule
~ & $C$ & $\widetilde{C}$ \\
\hline
Numerator & & \\
\quad total degree & 11 & 47 \\
\quad number of terms & 704 & 246,232 \\
Denominator & & \\
\quad total degree & 7 & 33 \\
\quad number of terms & 1 & 1750 \\
\bottomrule
\end{tabular}
\vspace{.5cm}
\caption{Comparison for Example~\ref{erk} of the critical functions $C$ and $\widetilde{C}$.  
We see that $C$ is simpler and hence easier to analyze. In particular, its denominator is a monomial and thus is easily seen to be positive on the positive orthant.}
\label{tab:compare}
\end{table}
%{\color{blue} Note that the ERK network is MESSI \cite{messi}, but not binomial.  We are the first to give a systematic way to obtain a witness for this network.}
\end{example}

%Phosphorylation of two substrates
\begin{example}[Phosphorylation of two substrates, continued]
\label{CFMW612}
Revisiting the network in Example~\ref{ex:CFMW611},
we exhibit two critical functions $C$ and $\widetilde{C}$:
 one that 
 arises from a simplification $\augmentH$, and then the one that
Conradi, Feliu, Mincheva, and Wiuf presented~\cite{CFMW}. We will see again that $C$ is much simpler than $\widetilde{C}$. 
%Here we present two different ways to deduce two different critical functions for the network $G$ in the Example \ref{ex:CFMW611}, see the functions 
%$C_{G, \psi}$ and $C_{G, \phi}$ below.

\smallskip

{\bf The first critical function.}
Consider the following upper-triangular matrix:
{\footnotesize
\[M(\kappa)=
\left(\begin{array}{cccccccccc}
\frac{1}{\kappa_3}&  0& 0& 0  & \frac{1}{\kappa_3}&0 \\
 0&   \frac{1}{\kappa_9}& 0& 0  & 0& \frac{1}{\kappa_9}\\
 0 & 0 &  \frac{1}{\kappa_1}& 0 & 0 & 0\\
 0 &  0 & 0     &  \frac{1}{\kappa_7} & 0 & 0\\
 0 & 0 & 0   & 0   &  \frac{1}{\kappa_4}& 0\\
  0     &  0 & 0   & 0   & 0       &  \frac{1}{\kappa_{10}}
\end{array}
\right).
\]
}
Note that ${\rm det} M(\kappa)>0$ for all $\kappa\in \mathbb{R}^{10}_{>0}$. 
Following~\eqref{eq:linearchange}--\eqref{consys-h}, the resulting system is
%{\color{blue} Cite Proposition~\ref{prop:multiply-by-M} here?}
\begin{center}
\begin{tabular}{ll}
$h_{c, a,1}  =x_1 + x_7 + x_8 - c_1$, &
$h_{c, a,2}  =  x_2  + x_9 + x_{10} - c_2$,\\
$h_{c, a,3}=   x_3 + x_4 + x_7 + x_9 - c_3$,&
$h_{c, a,4}=  x_7 -a_1x_9$,\\
$h_{c, a,5} =  x_5+ x_6 + x_8 + x_{10} - c_4$,&
$h_{c, a,6} =   x_8-a_2x_{10}$,\\
$h_{c, a,7} =  x_1x_3-a_3x_7$,&
$h_{c, a,8} =   x_1x_5-a_4x_8$,\\
$h_{c, a,9}=   x_2x_4-a_5x_9$,&
$h_{c, a,10} =   x_2x_6-a_6x_{10}$, 
\end{tabular}
\end{center}
where the effective parameters are: 
\begin{center}
\begin{tabular}{llllll}
$\bar a_1=\frac{\kappa_6}{\kappa_3}$,& $\bar a_2=\frac{\kappa_{12}}{\kappa_9}$,& $\bar a_3=
\frac{\kappa_2+\kappa_3}{\kappa_1}$,&$\bar a_4=\frac{\kappa_8+\kappa_9}{\kappa_7}$, & $\bar a_5=\frac{\kappa_5+\kappa_6}{\kappa_4}$, & 
$\bar a_6=\frac{\kappa_{11}+\kappa_{12}}{\kappa_{10}}$~.
\end{tabular}
\end{center}
Notice that the  map 
$\mathbb{R}^{12}_{>0} \to \mathbb{R}^{6}_{>0}$ given by 
$ \kappa\mapsto \bar a(\kappa)$ is surjective. 
%Notice that for any fixed $c\in {\mathbb R}_{>0}^{4}$,  there exists $\kappa\in {\mathbb R}_{>0}^{12}$ 
%such that $\F(x)=0$ has two or more positive solutions if and only if there exists $a\in {\mathbb R}_{>0}^{6}$ such that $h_c(x)=0$ has two or more positive solutions. 
The Jacobian matrix of $\augmentH$ is %with respect to $x$ is
{\footnotesize
\[{\rm Jac}(\augmentH)=
\left(\begin{array}{cccccccccc}
1 & 0 & 0& 0& 0 & 0&1 & 1 & 0& 0\\
0 & 1 & 0& 0& 0& 0 & 0& 0 & 1& 1  \\
0 & 0 & 1 & 1 & 0 & 0 & 1 &0 & 1 &0\\
0& 0 & 0& 0& 0& 0& 1& 0  & -a_1&0 \\
0& 0& 0& 0& 1& 1& 0 & 1 &0 &1\\
0& 0 & 0& 0& 0& 0& 0& 1  & 0&-a_2 \\
x_3 & 0 & x_1 & 0 & 0 & 0 & -a_3& 0 & 0 & 0\\
x_5 & 0 & 0    & 0 & x_1 & 0 & 0     & -a_4 & 0 & 0\\
0     & x_4 & 0    & x_2 & 0 & 0 & 0   & 0   & -a_5 & 0\\
0     & x_6& 0    &  0     & 0 & x_2 & 0   & 0   & 0       & -a_6 
\end{array}
\right).
\]
}By solving the equations 
$h_{c, a,4}=h_{c, a,6}=h_{c, a,7}=h_{c, a,8}=h_{c, a,9}=h_{c, a,10}=0$ 
in the unknowns $a_1, a_2, a_3, a_4, a_5, a_6$, we obtain
\begin{center}
\begin{tabular}{llllll}
$a_1=\frac{x_7}{x_9}$, &
$a_2=\frac{x_8}{x_{10}}$, & 
$a_3=\frac{x_1x_3}{x_7}$, & 
$a_4=\frac{x_1x_5}{x_8}$, & 
$a_5=\frac{x_2x_4}{x_9}$,&
$a_6=\frac{x_2x_6}{x_{10}}$.
\end{tabular}
\end{center}
This yields a steady-state parametrization
%By the parametric solutions above, consider a function 
$\phi: \mathbb{R}^{10}_{>0} \rightarrow   \mathbb{R}_{>0}^{6}\times  \mathbb{R}_{>0}^{10}$, 
%for any point $x\in \mathbb{R}^{10}_{>0}$, 
where 
%$u=x$
$\hat a=\emptyset$ and $\hat x=x$, so we denote  
 $\phi(\hat a, \hat x)$ by $\phi(x)=(a_1, \dots, a_6,x)$, which is defined as follows:  
% PARAMETRIZATION
\begin{align} \label{eq:CFMW-param-2}
% \phi(x) ~:=~
\left(a_1=\frac{x_7}{x_9},~a_2=\frac{x_8}{x_{10}},~a_3=\frac{x_1x_3}{x_7} ,
~a_4=\frac{x_1x_5}{x_8} ,~a_5=\frac{x_2x_4}{x_9} , ~a_6=\frac{x_2x_6}{x_{10}}; ~ x \right)~.
\end{align}
We substitute the parametrization $\phi(x)$ into $\det {\rm Jac}(\augmentH)$ to 
obtain $C(x)=-\frac{x_1x_2}{x_9x_{10}}B(x)$, where $B(x)$ is the following polynomial:
\begin{center}
{$x_1x_{10}x_2x_3+x_1x_{10}x_2x_4+x_1x_{10}x_2x_7+x_1x_{10}x_2x_9+x_1x_{10}x_3x_6
+x_1x_{10}x_4x_6+x_1x_{10}x_4x_9+x_1x_{10}x_6x_7+x_1x_{10}x_6x_9+x_1x_2x_3x_5
+x_1x_2x_3x_6+x_1x_2x_3x_8+x_1x_2x_4x_5+x_1x_2x_4x_6+x_1x_2x_4x_8+
x_1x_2x_5x_7+x_1x_2x_5x_9+x_1x_2x_6x_7+x_1x_2x_6x_9+x_1x_2x_7x_8+x_1x_2x_8x_9+
x_1x_4x_5x_9+x_1x_4x_6x_9+x_1x_4x_8x_9+x_{10}x_2x_3x_7+
x_{10}x_3x_6x_7\underline{-x_{10}x_4x_5x_7}+x_2x_3x_5x_7+x_2x_3x_5x_8+x_2x_3x_6x_7+x_2x_3x_7x_8+
x_2x_4x_5x_8+x_2x_5x_7x_8+x_2x_5x_8x_9\underline{-x_3x_6x_8x_9}+x_4x_5x_8x_9
$.
}
\end{center}
We see that 
$B(x)$ is square-free and homogeneous, so every monomial of $B(x)$ defines a vertex of NP$(B(x))$.  So, as $B(x)$ has both positive and negative terms (underlined), 
we conclude from Lemma~\ref{lem:NP} that $B(x)$, and thus $C(x)$ too, changes sign.
So, by Theorem~\ref{thm:c-general} (this network is conservative and hence dissipative, and 
has no boundary steady states~\cite{CFMW,messi}), the network is multistationary. 
We will see that such properties of $B(x)$ -- having square-free monomials and homogeneous -- 
come from the fact that every non-conservation-law $h_{c,a,i}$
is a binomial (Theorem~\ref{thm:NP}).

\smallskip

{\bf The second critical function (from \cite{CFMW}).}
Using the equations for $\F$ given in Example~\ref{ex:CFMW611}, the resulting Jacobian matrix is as follows: %with respect to $x$ is
{\footnotesize
\[{\rm Jac} ({\F})=
\left(\begin{array}{cccccccccc}
1 & 0 & 0& 0& 0 & 0& 1 & 1 & 0& 0\\
0 & 1 & 0& 0& 0& 0  &  0& 0& 1 & 1 \\
0 & 0 & 1 & 1 & 0 & 0 & 1 &0 & 1& 0\\
0 & -\kappa_4x_4 & 0 & -\kappa_4x_2 & 0&  0 & \kappa_3 &0&\kappa_5&0\\
0 & 0 & 0 & 0 & 1 & 1 & 0 & 1 & 0 & 1\\
0 & -\kappa_{10}x_6 & 0 &0 & 0 & -\kappa_{10}x_2 & 0  &\kappa_9& 0 &\kappa_{11}\\
\kappa_1x_3 & 0 & \kappa_1x_1 & 0 & 0 & 0 &  -\kappa_2-\kappa_3 & 0& 0& 0\\
\kappa_7x_5 & 0 & 0 & 0 &  \kappa_7x_1 & 0 & 0  & -\kappa_8-\kappa_9 & 0 & 0\\
0 & \kappa_4x_4 & 0 & \kappa_4x_2 & 0  & 0 & 0 & 0 &-\kappa_5-\kappa_6 & 0   \\
0 & \kappa_{10}x_6 & 0 & 0& 0 &  \kappa_{10}x_2 & 0 & 0 & 0& -\kappa_{11}-\kappa_{12}\\
\end{array}
\right).
\]
}%
In \cite{CFMW}, a steady state parametrization of the form
  $(\kappa,\hat x) \mapsto (\kappa,(\hat x,\Phi(\kappa, \hat x))$, where $\hat x=(x_1, x_2, x_3, x_5)$, is derived as follows.  
The equations $f_{c,\kappa,i}=0$ $(i=4, 6, 7, 8, 9, 10)$
are solved for the unknowns $x_i$ $(i=4, 6, 7, 8, 9, 10)$,
which yields a steady-state parametrization,
 \[
 \psi : \mathbb{R}^{16}_{>0} \longrightarrow   \mathbb{R}_{>0}^{12}\times  \mathbb{R}_{>0}^{10}~,\] 
 where 
 $\psi( \kappa; x_1, x_2, x_3, x_5)$ is defined as
{\footnotesize \[\left(\kappa; ~x_1, x_2, x_3, \frac{\kappa_1\kappa_3(\kappa_5+\kappa_6)x_1x_3}{\kappa_4\kappa_6(\kappa_2+\kappa_3)x_2}, 
x_5, \frac{(\kappa_{11}+\kappa_{12})\kappa_7\kappa_9x_1x_5}{(\kappa_8+\kappa_9)\kappa_{10}\kappa_{12}x_2}, 
\frac{\kappa_1x_1x_3}{\kappa_2+\kappa_3},
\frac{\kappa_7x_1x_5}{\kappa_8+\kappa_9},
\frac{\kappa_1\kappa_3x_1x_3}{(\kappa_2+\kappa_3)\kappa_6},
\frac{\kappa_1\kappa_7\kappa_9x_1x_5}{(\kappa_8+\kappa_9)\kappa_{12}}
\right)~.
\]
}
Substituting %positive parametrization 
$(\kappa, x)=\psi(\kappa, x_1, x_2, x_3, x_5)$ into $\det {\rm Jac} ({\F})$ yields 
 the critical function 
$\widetilde{C}(\kappa; x_1, x_2, x_3, x_5)$, which is exactly the function ``$a(\hat{x})$" in \cite[\S 6 in Supplementary Information]{CFMW}.  
This critical function
$\widetilde{C}$ 
is a rational function, and the denominator is $(\kappa_2+\kappa_3)\kappa_6(\kappa_8+\kappa_9)\kappa_{12}x_2$ and thus is positive for every 
choice of positive $(\kappa; x_1, x_2, x_3, x_5)\in \mathbb{R}^{16}_{>0}$.
The numerator is a polynomial 
in $\mathbb{Q}[\kappa; x_1,x_2,x_3,x_5]$ 
with total degree $15$ and $225$ terms. 
Conradi, Feliu, Mincheva, and Wiuf showed that $\widetilde{C}$ changes sign, and so again by 
Theorem~\ref{thm:c-general} 
(or \cite[Corollary 2]{CFMW}) the network is multistationary.

% THE SECOND CRITICAL FUNCTION

We compare the two critical functions $C$ and  $\widetilde{C}$ described above:
\begin{enumerate}
\item $C$ is simpler than $\widetilde{C}$. The numerator of $C$ has $36$ terms while that of $\widetilde{C}$ has $225$ terms. 
\item Unlike the numerator of $\widetilde{C}$, the numerator of $C$ (namely, $B(x)$)
is square-free and homogeneous. 
So, to conclude that the network is multistationary, 
we only had to examine the signs in $B(x)$, whereas in~\cite{CFMW} 
a more careful examination of the Newton polytope of the numerator of $\widetilde{C}$ had to be undertaken.
\item
The critical function $\widetilde{C}$ contains the rate-constant vector $\kappa$, 
so the authors of~\cite{CFMW}
used the structure of $\widetilde{C}$ to derive
necessary and sufficient
conditions on $\kappa$ for multistationarity -- in other words, they found the multistationary region of parameter space.  
In contrast, $C$ does not contain $\kappa$, 
so we can not readily use $C$ to find the multistationary region.
Nonetheless, we can use $C$ to find a witness to multistationarity (by Theorem~\ref{thm:c-general}).
\end{enumerate}
\end{example}

%-----------------------------
% SECTION: Discriminants
%-----------------------------
\subsection{Critical functions and discriminants} \label{sec:discriminant}

For readers familiar with real algebraic geometry,  we now explain the relationship between a critical function 
$C(\hat a, \hat x)$ and the mixed discriminant of the polynomial equation system $\augmentH(x)=0$ \cite{GKZ}. Let $I$ be the ideal in the polynomial ring
$\mathbb Q[c,a,x_1,x_1^{-1}, \dots, x_s, x_s^{-1},y_1, y_1^{-1}, \dots, y_s, y_s^{-1}]$ generated by $h_{c,a,1}(x), \dots, h_{c,a,s}(x)$ and the equations
expressing that $y=(y_1, \dots, y_s)$ lies in the kernel of the Jacobian matrix
${\rm Jac}(\augmentH)$ with respect to the $x$ variables. Note that the vanishing of these equations for given values $(c^*,a^*)$ of $(c,a)$ at $x^*, y^*$ imply that $\det \left( {\rm Jac}(h_{c^*,a^*})\right)(x^*)=0$ and so $x^*$ is a non-simple common root of $h_{c^*,a^*}$.
If the elimination ideal
	%\begin{align} \label{eq:discriminant}
$I ~\cap~ {\mathbb Q}[c,a]$
	%\end{align}
	has codimension one, we call 
	any generator of this ideal a 
	 {\em mixed discriminant} of the family $h_{c,a}(x)$, and we denote it by $D(c,a)$. 
	 So, 
 $D(c^*,a^*)=0$ whenever $h_{c^*,a^*}$ has a multiple root over $\mathbb{C}^*$ where all maximal minors of size $(s-1)$ of the Jacobian matrix are nonzero.  As a real positive root $x^*$ of $h_{c^*,a^*}(x)$ is a real positive root of $f_{k^*}(x)=0$ when $\bar a(k^*) = a^*$,  and the matrix $M(\kappa^*)$ is invertible,  such  a positive multiple root $x^*$ of $h_{c^*,a^*}$ is a degenerate steady state of the dynamical system $\dot{x}=f_{\kappa^*}(x)$, and in this case $D(c^*,a^*)=0$.

Assume that the codimension of the elimination ideal is one and
 consider the {\em discriminant locus}, i.e., the set of pairs of vectors of parameters $(c^*,a^*)$ where the discriminant vanishes ($D(c^*,a^*)=0$). Consider also
 the {\em critical locus}, i.e., the set of pairs $(\hat a, \hat x)$ where 
the critical function vanishes ($C(\hat a, \hat x)=0$).  We relate these
two conditions in Proposition~\ref{prop:discriminant} below.
%{\color{teal} Say ``branch points'' or ``branch set''?}
We begin with a simple example.

%: EXAMPLE
\begin{example} \label{ex:1-species-net}
Consider the following network~\cite{joshi2013complete}, where
we set two of the rate constants to 1 to simplify the analyses:
\begin{align}\label{eq:net-1-species}
\begin{tabular}{cc}
\ce{0
<=>[1][\kappa]
A
}, &
\ce{2A
->[1]
3A
}~.
\end{tabular}
\end{align}
%\begin{align}  
%0 
 % \arrowschem{1}{\kappa } A \quad \quad
  %2A
  %\stackrel{1}{\rightarrow} 3A
%\end{align}
The resulting mass-action kinetics ODE is: 
	\begin{align}\label{eq:ODE-1-species}
	\frac{dx_A}{dt}~=~  x_A^2 - \kappa x_A +1~,
	\end{align}
and the discriminant equals $\kappa^2-4$.  We obtain, from the ODE~\eqref{eq:ODE-1-species}, 
the steady-state parametrization $\phi: \mathbb{R}_{>0} \to \mathbb{R}^2_{>0}$ given by: 
\begin{align}  \label{eq:param-1-species}
 x_A ~\mapsto ~\left( x_A, ~ \kappa=\frac{x_A^2+1}{x_A}\right)~.
\end{align}
(Here, $\bar a$ is the identity map  $\bar a(\kappa)=\kappa$,  and $\hat a=\emptyset, \hat x=x_A$.)
This steady-state parametrization is depicted in Figure~\ref{fig:param}, 
along with the unique degenerate steady state (when $(x_A,\kappa)=(1,2)$) 
and the corresponding projections onto the critical locus ($x_A=1$)
and
 onto the discriminant locus ($\kappa =2$). The second point on the discriminant locus, $\kappa=-2$, is not shown.
\end{example}

\begin{center}
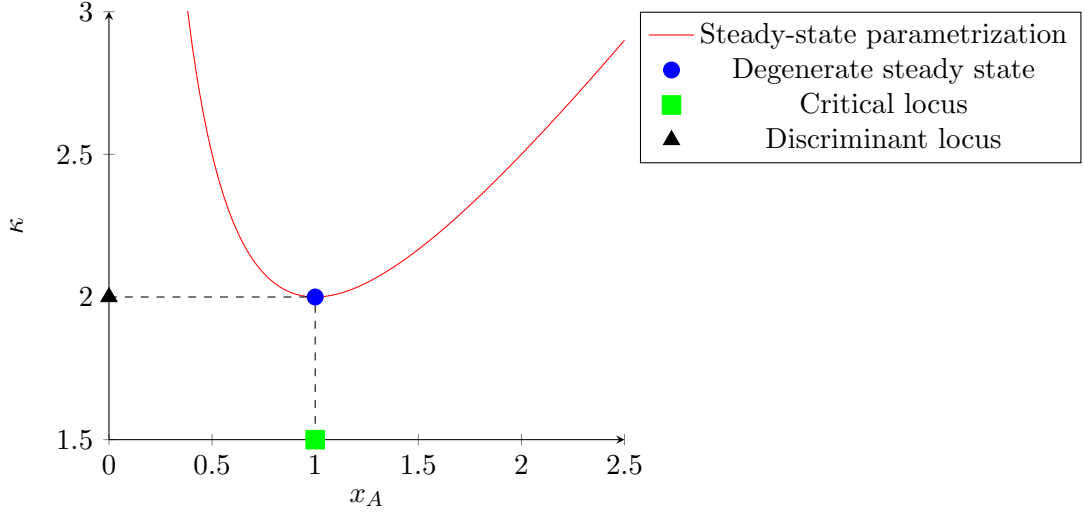
\begin{figure}[ht]
\begin{tikzpicture}
\begin{axis}[legend pos=outer north east,
    ymin=1.5, ymax=3,
    axis lines = left,
    xlabel = $x_A$,
    ylabel = {$\kappa$},
]
\addplot[
	domain=0:2.5, 
    samples=100, 
    color=red
    ]
{x+1/x};
\addlegendentry{Steady-state parametrization}
\addplot[
    only marks,
    mark=oplus*,
    mark size=3pt,
    color=blue
    ]
    coordinates {
    (1,2)    };
\addlegendentry{Degenerate steady state}
\addplot[
    only marks,
    mark=square*,
     mark size=3.5pt,
   color=green
    ]
    coordinates {
    (1,1.5)    };
\addlegendentry{Critical locus}
\addplot[
    only marks,
    mark size=3.5pt,
    mark=triangle*,
    color=black
       ]
    coordinates {
    (0,2)    };
\addlegendentry{Discriminant locus}
% DASHED LINES
\addplot [domain=0:1,samples=2, dashed] {2};
\addplot +[mark=none,dashed,color=black] coordinates {(1,1.5) (1, 2)};
\end{axis}
\end{tikzpicture}
\caption{The image of the
steady-state parametrization~\eqref{eq:param-1-species}
for network~\eqref{eq:net-1-species}. %$\kappa=(x_A^2+1)/x_A$, 
Also shown is the unique degenerate steady state
and its projections to the critical locus (where $C=0$)
 and discriminant locus (where $D=0$).
  \label{fig:param}}
\end{figure}
\end{center}

We see in Figure~\ref{fig:param} that the critical locus is a projection of the degenerate steady states, 
and the discriminant locus is a subset of %(and in general not equal to) 
a similar projection.  
The first observation -- viewing the critical locus as a projection of the degenerate steady states --
 gives an {\em intrinsic definition of the critical locus}, which does not depend on the specific parametrization.  
 As for the second observation, pertaining to the discriminant locus, we generalize this in 
Proposition~\ref{prop:discriminant}.

%-----------------------------
% PROPOSITION: Discriminants
%-----------------------------

\begin{proposition}[$C=0$ implies $D=0$] \label{prop:discriminant}
Consider
a network $G$ with conservation-law matrix~$W$ and 
simplified system $\augmentH(x)$ as in Definition \ref{def:effective}.
Assume $I \cap \mathbb Q [c,a]$ has codimension one, and
let $D(c,a)$ denote the mixed discriminant.
	Consider, as in Definition~\ref{def:parametrization-for-h}, a positive parametrization 
	$\phi: \mathbb{R}^{\hat{m}}_{>0}\times  \mathbb{R}_{>0}^{\hat s}\rightarrow  \mathbb{R}_{>0}^{\bar m}\times  \mathbb{R}_{>0}^{s}$, 
	denoted by $(\hat a, \hat x) \mapsto \phi(\hat a, \hat x)$, with respect to $\augmentH$.
	Let $C(\hat a, \hat x)$ denote the resulting critical function. 
	Then, for $(\hat a^*, \hat x^*) \in  \mathbb{R}^{\hat{m}}_{>0}\times  \mathbb{R}_{>0}^{\hat s}$, if $C(\hat a^*, \hat x^*)=0$ and any  minor of the Jacobian matrix ${\rm Jac}(h_{\hat c^*, \hat a^*})(\hat x^*)$ of size $(s-1)$ is nonzero, then 
	$D( c^*, a^*)=0$, where $(a^*, x^*)=\phi(\hat a^*, \hat x^*)$ and 
	$c^*=W x^*$.
	\end{proposition}
\begin{proof} 
If $C(\hat a^*, \hat x^*)=0$, 
then, by definition,
%Definition \ref{def:parametrization-for-h}, 
$ \left(\det {\rm Jac}(\augmentH)\right)|_{(a, x)=\phi(\hat a^*, \hat x^*)}=0$.
Let $(a^*, x^*)=\phi(\hat a^*, \hat x^*)$. 
Then, 
for $c^*=W x^*$, 
the system $h_{c^*,a^*}=0$ has a multiple root (namely, $x^*$).
So, by the definition of the mixed discriminant and the assumptions about the Jacobian matrix ${\rm Jac}(h_{\hat c^*, \hat a^*})(\hat x^*)$, we have 
$D(c^*, a^*)=0$.  
\end{proof}
The converse of Proposition~\ref{prop:discriminant} does not hold, as we see in the following example.

\begin{example}[Converse of Proposition~\ref{prop:discriminant} is false] \label{ex:discriminant-converse}
We revisit the network in Examples~\ref{ex:CFMW611} and \ref{CFMW612}. 
For the following choice of $\kappa^*$ and $c^*$:
	\begin{align*}
\kappa^*~&=~ 
\left(1, ~  \frac{1}{20},  ~ \frac{1}{20}, ~  1, ~  \frac{19}{2}, ~  \frac{1}{2}, ~ 1,~   5,~  5,~ 6,~ \frac{926}{823}, ~  \frac{4630}{823}\right)~ \quad {\rm and}\\
c^*~&=~
\left(  12, ~  \frac{2675}{926},  ~22,  ~ \frac{11935}{926} \right)~,
	\end{align*}
there are two positive steady states:
 \begin{align*}
x^{(1)} ~&=~ 
	\left( 1,~  1,~  1,~  10,~  10,~   1,~  10,~  1,~   1,~   \frac{823}{926} \right) \quad {\rm and }\\
x^{(2)}~&\approx~ \left(  2.48, ~ 0.60 ,~  0.31,~  13.08, ~6.84 , ~ 2.81 , ~7.82 , ~1.70,~ 0.78, ~ 1.51 \right)~.
\end{align*}
The first steady state, $x^{(1)}$, is degenerate, while the second, $x^{(2)}$, is nondegenerate.
Thus, for the parametrization~\eqref{eq:CFMW-param-2}, which we denote by $x \mapsto \phi(x)$, it follows that 
$C(x^{(1)})=0$ (because $x^{(1)}$ is degenerate).
Let $(a^{(1)}, x^{(1)})=\phi(x^{(1)})$ and $(a^{(2)}, x^{(2)})=\phi(x^{(2)})$.
 So, 
\[
D\left(W x^{(2)},  a^{(2)} \right) 
~=~
D\left( W x^{(1)}, a^{(1)}\right)
~=~
D\left( c^*, a^{(1)} \right)=0~,
\]
where we are using the fact that $a^{(1)}=a^{(2)}$
and
$W x^{(1)}=W x^{(2)}=c^*$. 
However, the second steady state, $x^{(2)}$, is nondegenerate, so
although $D\left(W x^{(2)},  a^{(2)} \right)=0$,
we nonetheless have $C(x^{(2)}) \neq 0$.
\end{example}

\section{Multistationarity for linearly binomial networks} \label{sec:binomial}
This section focuses on a class of networks that
appear surprisingly often in applications.  We called these 
networks {\em linearly binomial networks} in Definition~\ref{def:lbin} 
because their steady-state equations are equivalent to 
binomial equations, and these binomial equations are related to the polynomials arising from mass-action kinetics linearly  via multiplication by a matrix $M(\kappa)$. % (like in Example~\ref{CFMW612}).  
We prove for linearly binomial networks 
that multistationarity can be decided simply from inspecting the critical function
(Theorem~\ref{thm:NP}) and
that the regions of parameter space
allowing for multistationarity that we find via degree theory and the use of effective parameters are full-dimensional  
(Theorem~\ref{thm:openmap}). 

It follows from Definition~\ref{def:lbin} that
 {\em linearly binomial networks} 
are the networks for which the 
augmented system $f_{c,\kappa}$ can be simplified to some
$h_{c,a}$,
as in~\eqref{consys-h},
% (as in the Setup in Section~\ref{sec:param}: 
%$(h_{j_1}, h_{j_2}, \dots, h_{j_{s-d}})^T
% =  M(\kappa) ~  (f_{j_1}, f_{j_2}, \dots, f_{j_{s-d}})^T$)
with the following form: 
\begin{equation}\label{toric}
\augmentH(x)_i=%~a_i(\kappa)\cdot x^{\delta_i} -  x^{\gamma_i}, \;\;\; i \in \{1, \ldots, s\}\backslash \{i_1, \ldots, i_d\}
\begin{cases}
 x^{\gamma_k} - a_k\cdot x^{\delta_k}&~\text{if}~i %\in I \\
 	 ~=~j_k \in \{j_1,j_2,\dots, j_{s-d}\} \quad \text{(i.e., } i \notin I\text{)}\\
(Wx-c)_k &~\text{if}~i =i_k \in I~.
	%=~i_k\in  \{i_1, i_2,\ldots, i_d\}~,
\end{cases}
\end{equation}
where for $k  = 1,2, \dots,  s-d$, 
the coefficient $a_k$ %\in {\mathbb Q}(\kappa)$ 
is an indeterminate arising from an effective parameter $\bar a_i(\kappa)$,
%effective parameter
and
both  $x^{\gamma_k}$ and $x^{\delta_k}$  are monomials. 
The steady-state equations can be transformed into binomial equations by linear operations, as in \eqref{eq:linearchange}, 
via a matrix $M(\kappa)$ which has positive determinant for all $\kappa \in \mathbb{R}^m_{>0} $.

\begin{remark} \label{rem:tss}
Mass-action systems arising from linearly binomial networks have ``toric steady states'', as defined in~\cite{TSS}, 
for every choice of positive rate constants.  For networks with toric steady states, analyzing and finding witnesses for multistationarity has been investigated~\cite{signs,messi,TSS}.
\end{remark}

We use the binomials in $(\ref{toric})$ to solve for the indeterminates $a_k$:
%{\color{blue}the effective parameters}:
\begin{equation*} %\label{solvea}
a_k~=~x^{\gamma_k-\delta_k},\;\;\;\;\;k~=~1,2, \dots, s-d~,
\end{equation*}
%\[\phi: \mathbb{R}^{s}_{>0} \rightarrow   \mathbb{R}_{>0}^{s-d}\times  \mathbb{R}_{>0}^{s}\]
and this yields a steady-state parametrization $\phi(\hat a, \hat x)$, where $\hat a=\emptyset$ and $\hat x=x$:
\begin{align} \label{eq:phia}
	\phi:~ \mathbb{R}^{s}_{>0} &~\rightarrow~  \mathbb{R}_{>0}^{s-d}\times  \mathbb{R}_{>0}^{s}~\\ \nonumber
		(\hat a, \hat x) & ~\mapsto~ (\phi_a(x),x)
	\end{align}
given by
%\begin{equation} \label{eq:phi-a}
$\phi_a(x)~:=~\left( x^{\gamma_{1}-\delta_1}, ~ x^{\gamma_{2} -\delta_2}, ~\ldots~,~  x^{\gamma_{s-d}-\delta_{s-d}} \right)$.
%, ~x_1, x_2, \ldots,  x_s\right)~.
%\end{equation}
The parametrization $\phi(\hat a, \hat x)$, which we denote by $\phi(x)$, 
yields the following critical function:
	\begin{equation}\label{ecf}
C(x)~ =~\left(\det {\rm Jac}({\augmentH})\right)|_{a_{k}= x^{\gamma_{k}-\delta_k}, \;\;\;\;k~=1, \ldots, s-d}~.
	\end{equation}

\begin{example}[Phosphorylation of two substrates, continued] \label{ex:2-phos-continued}  
The network describing phosphorylation/dephosphorylation of two substrates, from 
Examples~\ref{ex:CFMW611} and~\ref{CFMW612},
is linearly binomial.  This can be seen from the system $\augmentH$ given in Example~\ref{CFMW612}.  
Moreover, the steady-state parametrization given there in~\eqref{eq:CFMW-param-2}, has the form in~\eqref{ecf}.
\end{example}

\begin{example}[ERK network, continued] \label{ex:eek-continued}
The ERK network from Example~\ref{erk} is MESSI (see Section \ref{sec:defmessi}) but not linearly binomial. 
This result can be checked using results from~\cite{messi}.
%{\color{violet} Prove this using results from~\cite{messi}?  
%Or: This follows from the fact, mentioned earlier (is it?), that the network is not toric.}
\end{example}

\subsection{Establishing multistationarity using the critical function} \label{sec:inspection}
For linearly binomial networks, we show that the critical function
$C(x)$
has the same sign as a  homogeneous polynomial $B(x)$ with square-free monomials
(Lemma~\ref{lm:ecfstructure}).
It follows that every monomial of $B(x)$ defines a vertex of the Newton polytope,
so establishing 
multistationarity for a linearly binomial network is easy:
simply check whether $B(x)$ has coefficients of both signs
(Theorem~\ref{thm:NP}).
%the critical function (\ref{ecf}) has a ``nice" structure such that one can easily determine whether $\widetilde{C}(x)$ changes sign over ${\mathbb R}_{>0}^s$. 

Recall that a polynomial $g(x)$ is {\em homogeneous of degree} $d$ if 
$g(\lambda x) = \lambda^d g(x)$ for all constants $\lambda \in {\mathbb C}$.
%In particular, the zero polynomial is homogeneous of degree $d$ for any $d$.

\begin{lemma}\label{lm:ecfstructure}
For a linearly binomial network $G$, 
the critical function $C(x)$ in (\ref{ecf}) has the form 
\begin{equation}\label{nice}
C(x)~=~\frac{x^{\alpha}}{x^{\beta}}\cdot B(x)~,
\end{equation}
where 
$x^{\alpha}$ and $x^{\beta}$ are monomials, and
$B(x)$ is either the zero polynomial or a homogeneous polynomial 
of %total 
degree $d=
s-rank(N)$ such that all monomials have
 exponents vectors with coordinates $0,1$. Here,
 $s$ denotes the number of species, and $N$ denotes the stoichiometric matrix of $G$.
 % (possibly, zero).
%\begin{itemize}
%\item $B(x)$ is a homogenous polynomial in $x$,
%\item $B(x)$ is either the zero polynomial or has total degree $d$, and
%\item all monomials in $B(x)$ are square-free. 
%\end{itemize}
%\textcolor{brown}{Hence, $\widetilde C$ changes sign if and only if $B$ contains both positive and negative terms.
\end{lemma}
\begin{proof}
First we reorder the polynomials $h_{c,a}$ as $b_1, b_2, \ldots, b_s$ such that the first $s-d$ polynomials are the binomials:
\begin{equation*} %\label{reorder1}
b_k~=~
x^{\gamma_{k}} - a_{k}\cdot x^{\delta_{k}}~, \;\;\;\;\;\;\; k = 1,2, \ldots, s-d~,
\end{equation*}
and the remaining $d$ polynomials are the linear equations from the conservation laws:
\begin{equation*}
b_{s-d+\ell}~=~(Wx-c)_{\ell}~, \;\;\;\;\;\;\;\;   \ell = 1, 2,\ldots, d~.
\end{equation*}
As $b$ is obtained by reordering the polynomials $h_{c,a}$, we have that
\[\det {\rm Jac}({\augmentH})~=~(-1)^\mu\cdot \det {\rm Jac}({b}),\]
for some integer $\mu$.
%\;\;\;\;\text{for some positive integer}\;n.\]
So, the critical function, as in (\ref{ecf}), is:
\[C(x)~ =~\left(\det {\rm Jac}({\augmentH})\right)|_{a_k= x^{\gamma_k-
\delta_k}}~=~ (-1)^\mu \cdot \left(\det {\rm Jac}(b)\right)|_{a_k=x^{\gamma_k-
\delta_k}}~.\]
%Simply, we use $J$ to denote the $s\times s$ matrix ${\rm Jac}(h')$. 
 %Note that by the construction of $h'$,  the $(k,  r)$-entry of ${\rm Jac}(h')$ is 
 %\[\]
 We first describe the first $s-d$ rows of ${\rm Jac}(b)$. 
 For every  $k =1, 2, \ldots, s-d$ and for every 
 $r=1,2, \ldots, s$, the entry
 in the $k$-th row and the $r$-th column of ${\rm Jac}(b)$ is
 \begin{equation*}% \label{entries}
 \frac{{\partial} b_{k}}{{\partial} x_r}~=~\gamma_{k, r}\frac{x^{\gamma_{k}}}{x_r} - a_{k} \delta_{k, r}\frac{x^{\delta_{k}}}{x_r}~.
 \end{equation*}
% Now, simply denote the matrix ${\rm Jac}(b)|_{a_k=\frac{x^{\gamma_k}}{x^{\delta_k}}}$ by $J$. 
 Thus, the entry
 in the $k$-th row and the $r$-th column of ${\rm Jac}(b)|_{a_k=x^{\gamma_k-
\delta_k}}$ is
\begin{equation*}%\label{entries}
\frac{{\partial} b_{k}}{{\partial} x_r}|_{a_{k}= x^{\gamma_{k} -\delta_{k}}}~=~
\frac{ x^{\gamma_{k}}}{x_r}\left(\gamma_{k, r}  - \delta_{k,r}\right)~.
\end{equation*}
  As for the last $d$ rows of ${\rm Jac}(b)$,
  this submatrix is exactly $W$, 
  the conservation-law matrix.
  So, 
  \[{\rm Jac}(b)|_{a_k=x^{\gamma_{k} -\delta_{k}}}~=~
  \left(\begin{array}{ccc} 
	    \frac{x^{\gamma_1}}{x_1}(\gamma_{1,1}-\delta_{1,1}) & \dots & \frac{x^{\gamma_1}}{x_s}(\gamma_{1,s}-\delta_{1,s})\\
	    \vdots & \ddots & \vdots \\
            \frac{x^{\gamma_{s-d}}}{x_1}(\gamma_{s-d,1}-\delta_{s-d,1}) & \dots & \frac{x^{\gamma_{s-d}}}{x_s}(\gamma_{s-d,s}-\delta_{s-d,s})\\
            \hline
            & & \\
            & W &\\
            & & 
\end{array}\right)~,
\]
and hence,
\begin{align} \label{eq:final-det}
\notag
  \left(\det {\rm Jac}(b)\right)|_{a_k=x^{\gamma_{k} -\delta_{k}}} 
  &~=~
  \underset{k=1}{\overset{s-d}{\prod}}x^{\gamma_k}\det\left(\begin{array}{ccc} 
	    \frac{1}{x_1}(\gamma_{1,1}-\delta_{1,1}) & \dots & \frac{1}{x_s}(\gamma_{1,s}-\delta_{1,s})\\
	    \vdots & \ddots & \vdots \\
            \frac{1}{x_1}(\gamma_{s-d,1}-\delta_{s-d,1}) & \dots & \frac{1}{x_s}(\gamma_{s-d,s}-\delta_{s-d,s})\\
            \hline
            & & \\
            & W &\\
            & & 
\end{array}\right)\\
&~=~
	\frac{\underset{k=1}{\overset{s-d}{\prod}}x^{\gamma_k}}{\underset{r=1}{\overset{s}{\prod}}x_r}\det\left(\begin{array}{ccc} 
	    (\gamma_{1,1}-\delta_{1,1}) & \dots & (\gamma_{1,s}-\delta_{1,s})\\
	    \vdots & \ddots & \vdots \\
            (\gamma_{s-d,1}-\delta_{s-d,1}) & \dots & (\gamma_{s-d,s}-\delta_{s-d,s})\\
            \hline
            & & \\
            & Wx &\\
            & & 
\end{array}\right),
\end{align}
where $Wx$ denotes the $d \times s$ matrix obtained from $W$ by multiplying
 the $j$-th column by $x_j$, for every $j=1,2, \dots,s$.

Therefore,
in the determinant in~\eqref{eq:final-det},
each term is a 
product of $d$ distinct $x_i$'s times a number, that is,  
we get a homogeneous polynomial $\widetilde B(x)$ of degree $d$ 
where each monomial has exponents vectors with coordinates $0,1$, as stated.
Hence,
 \[C(x)~ =~(-1)^\mu\cdot \left(\det {\rm Jac}(b)\right)|_{a_k=x^{\gamma_{k} -\delta_{k}}}~=~(-1)^\mu
 \cdot\frac{\Pi_{k=1}^{s-d}x^{\gamma_{k}}}{\Pi_{r=1}^{s}x_r} \widetilde B(x)~.\]
 Then, letting $B(x)=(-1)^\mu\widetilde B(x)$,
 it follows that $B(x)$ is as in~(\ref{nice}) and is homogeneous of degree $d$ 
with each monomial has exponents vectors with coordinates $0,1$.%
\end{proof}

Notice that $C(x)$ and $B(x)$, as in (\ref{nice}), have the same sign for every $x\in {\mathbb R}^s_{>0}$. 
This observation motivates the following definition.

\begin{definition}\label{ecfb}
The {\em critical polynomial} of a linearly binomial network 
is the polynomial
 $B(x)$ in (\ref{nice}).   
\end{definition}

In the remainder of this section, we make the following assumption: 
{\em B(x) is not the zero polynomial.}
Indeed, we do not know of any network for which the critical function is the zero function.  
For such a network, every steady state (in every compatibility class and for every choice of rate constants) would be degenerate.  
We suspect that no such networks exist, and in fact the non-existence of such networks is implied by the (open) 
Nondegeneracy Conjecture~\cite{Joshi:Shiu:Multistationary}.

%THEOREM: MSS FOR BINOMIAL NETWORKS
\begin{theorem}[Multistationarity for linearly binomial networks] \label{thm:NP}
Let $G$ be a linearly binomial network that is dissipative and has no boundary steady states in any compatibility class. 
Let $B(x)$ denote its critical polynomial. 
%, {\color{brown} and assume that $B(x)$ is not the zero polynomial.}
Then:
\begin{enumerate}[(A)]
% MULTI
\item  {\bf Multistationarity.}
 $G$ is multistationary 
if and only if $B(x)$ has a coefficient with sign equal to $(-1)^{\mathrm{rank}(N)+1}$, where 
 $N$ denotes the stoichiometric matrix of $G$.
 % WITNESS
\item  {\bf Witness to multistationarity.}
 Every $x^* \in \mathbb{R}^s_{>0}$ with 
${\rm sign} (B(x^*))=
(-1)^{\mathrm{rank}(N)+1}$ yields 
 a witness
to multistationarity $(\kappa^*, c^*)$ as follows:
let $c^* = W x^*$ %be the vector of conservation-law values defined by $x^*$
(where $W$ is the conservation-law matrix), and
let $\kappa^* \in \mathbb{R}_{>0}^m$ be such that 
%the specialization of the effective parameters arising from $\kappa^*$ satisfies
$\bar a(\kappa^*)=\phi_a(x^*)$ 
(where 
$\phi_a$ is as in~\eqref{eq:phia}). 
\end{enumerate}
\end{theorem}
% PROOF
\begin{proof}
Part (B) follows directly from Theorem~\ref{thm:c-general} and Lemma~\ref{lm:ecfstructure}.

For the forward direction of part (A), we proceed by contrapositive.  Assume that every coefficient of $B(x)$ 
has sign $(-1)^{\mathrm{rank}(N)}$ (at least one such coefficient exists, as $B(x)$ is nonzero).
So, $B(x)$, and thus $C(x)$ as well (by Lemma~\ref{lm:ecfstructure}), has sign
$(-1)^{\mathrm{rank}(N)}$ for all $x \in \mathbb{R}^s_{>0}$.  Thus, by Theorem~\ref{thm:c-general},
$G$ is monostationary.

For the backward direction of part (A), assume that 
$B(x)$ has a coefficient with sign equal to $(-1)^{\mathrm{rank}(N)+1}$. 
%Denote the corresponding term by $b x^{\sigma}$.  
By Lemma~\ref{lm:ecfstructure}, 
$B(x)$ is homogeneous with square-free monomials, so 
every monomial of $B(x)$ defines a vertex of the Newton polytope.
Then, by Lemma~\ref{lem:NP}, there exists
$x^* \in \mathbb{R}^s_{>0}$ with 
${\rm sign} (B(x^*))=
(-1)^{\mathrm{rank}(N)+1}$.  By again appealing to Lemma~\ref{lm:ecfstructure}, 
we conclude that $C(x)$ also takes that sign, so, by Theorem~\ref{thm:c-general},
$G$ is multistationary. 
\end{proof}

Theorem~\ref{thm:NP} and its proof yield the following procedure for obtaining, for linearly binomial networks, a witness to multistationarity.

%----------------------------
% PROCEDURE: witness to binomial networks 
%----------------------------
\begin{proc}[Witness to multistationarity for linearly binomial networks] \label{proc:witness}
~

 \noindent
{\bf Input:}  A  binomial network that is dissipative and has no boundary steady states, given by 
equations 
$\augmentH(x)$
as in~\eqref{toric}, 
arising from a reparametrization map $\bar a$ as in~\eqref{eq:ep}, 
with conservation-law matrix $W$ 
and steady-state parametrization $\phi(x)=(\phi_a(x),x)$ as in~\eqref{eq:phia}.\\
\noindent
{\bf Output:} ``No'' if $G$ is {\em not} multistationary; otherwise, a witness to multistationarity. \\
{\bf Steps:} 
\begin{enumerate}
	\item Does the critical polynomial $B(x)$ (Definition~\ref{def:Critical}) have a
	coefficient with sign equal to $(-1)^{\mathrm{rank}(N)+1}$ (e.g., if $B(x)$ has
	both positive and negative coefficients)?  If not, return ``No''.  If yes, pick any $x^* \in \mathbb{R}^s_{>0}$ 
	such that ${\rm sign}(B(x^*))=(-1)^{\mathrm{rank}(N)+1}$.  
	\item Let $c^*:=W x^*$.  %Let $(a^*,c^*):=(\phi_a(x), W x^*)$.
	\item Pick any $\kappa^* \in \mathbb{R}^{m}_{>0}$ such that
$\bar a (\kappa^*)=\phi_a(x^*)$. 
	\item
	 Return $(\kappa^*,c^*)$.
\end{enumerate}
\end{proc}

Recall from Remark~\ref{rmk:vertex} that one way to complete Step (1) (in the ``yes'' case) follows the proof 
of Lemma~\ref{lem:NP}.  
Namely, choose a monomial $x^{\alpha}$ of $B(x)$ with coefficient having sign $(-1)^{\mathrm{rank}(N)+1}$,
and 
then pick $x^* \in \mathbb{R}^s_{>0}$ as follows: let $ \lambda \gg 1$ and $x^*_i=\lambda$ when $i \in \alpha$ and  $x^*_i= 1$ when $i \notin \alpha$.  
It follows that $(x^*)^{\alpha} \gg (x^*)^{\beta}$ for every other monomial $x^{\beta}$ in $B(x)$, and hence 
 ${\rm sign}(B(x^*))=(-1)^{\mathrm{rank}(N)+1}$.  

We demonstrate Procedure~\ref{proc:witness} in the following example.

\begin{example}\label{cascade}

Consider the following network from \cite[Figure 1(k)]{FW2012}:

\begin{minipage}{0.25\textwidth}
 \begin{tikzpicture}[scale=0.7,node distance=0.5cm]
  \node[] at (1.1,-1.5) (dummy2) {};
  \node[left=of dummy2] (p0) {$P_0$};
  \node[right=of p0] (p1) {$P_1$}
    edge[->, bend left=45] node[below] {\textcolor{black!70}{\small{$F$}}} (p0)
    edge[<-, bend right=45] node[above] {} (p0);
  \node[] at (-1,0) (dummy) {};
  \node[left=of dummy] (s0) {$S_0$};
  \node[right=of s0] (s1) {\textcolor{blue}{$S_1$}}
    edge[->, bend left=45] node[below] {\textcolor{black!70}{\small{$F$}}} (s0)
    edge[<-, bend right=45] node[above] {\textcolor{black!70}{\small{$E$}}} (s0)
    edge[->,thick,color=blue, bend left=25] node[above] {} ($(p0.north)+(25pt,10pt)$);
 \end{tikzpicture}
\end{minipage}
\begin{minipage}{0.7\textwidth}
\begin{center}
\begin{tabular}{cc}
\ce{S_{0} + E
<=>[\kappa_1][\kappa_2]
ES_{0}
->[\kappa_3]
S_{1} + E},  &
\ce{S_{1} + F
<=>[\kappa_4][\kappa_5]
FS_{1}
->[\kappa_6]
S_{0} + F}\\
\ce{P_{0} + S_{1}
<=>[\kappa_7][\kappa_8]
S_{1}P_{0}
->[\kappa_9]
P_{1} + S_{1} }, 
&
\ce{P_{1} + F
<=>[\kappa_{10}][\kappa_{11}]
FP_{1}
->[\kappa_{12}]
P_{0} + F} \\

\end{tabular}
\end{center}
\end{minipage}

\noindent
This network describes a ``cascade motif'' with two layers;
each layer is a ``one-site modification cycle'', and the same
phosphatase ($F$) acts in each layer \cite[Figure 1(k)]{FW2012}.

The network has $s=10$ species:
\begin{center}
\begin{tabular}{llllll}
$X_1$=\ce{S_0},& $X_3$=\ce{S_1},& $X_5$=\ce{P_0},&$X_7$=\ce{P_1}, & $X_9$=\ce{E}, &  \\
$X_2$=\ce{ES_0},& $X_4$=\ce{FS_1},& $X_6$=\ce{S_1P_0},                             & $X_8$=\ce{FP_1},& $X_{10}$=\ce{F}~.\\
\end{tabular}
\end{center}

\noindent There are $d=4$ conservation laws, which arise from the total amounts of substrate $S$, enzyme $E$, enzyme $F$, and product $P$, respectively:
\begin{equation}\label{eq:ex45con}
%\begin{tabular}{llll}
x_1 + x_2 + x_3 + x_4 + x_6 = c_1, \;\; 
x_2 + x_9 = c_2,  \;\;
x_4 + x_8 + x_{10} = c_3, \;\;
x_5 + x_7 + x_6 + x_8 = c_4~.
%\end{tabular}
\end{equation}
% Augmented ODEs
The resulting augmented system is 
\begin{align*}
f_{c, \kappa,1} ~&=~ x_1 + x_2 + x_3 + x_4 + x_6 - c_1,   \\
f_{c, \kappa,2} ~&=~ x_2 + x_9 - c_2,  \\
f_{c, \kappa,3} ~&=~  \kappa_3x_2 - \kappa_4x_3x_{10} + \kappa_5 x_4 - \kappa_7x_3x_5 + \kappa_8x_6 + \kappa_9x_6, \\
f_{c, \kappa,4}~&=~ x_4 + x_8 + x_{10} - c_3,  \\
f_{c, \kappa,5} ~&=~  x_5 + x_7 + x_6 + x_8 - c_4,  \\
f_{c, \kappa,6} ~&=~ \kappa_7x_3x_5 - \kappa_8x_6 - \kappa_9x_6, \\
f_{c, \kappa,7} ~&=~ -\kappa_{10}x_7x_{10} + \kappa_{11}x_8 + \kappa_{9}x_6,  \\
f_{c, \kappa,8}~&=~ \kappa_{10}x_7x_{10}  - \kappa_{11} x_8 - \kappa_{12} x_8, \\
f_{c, \kappa,9}~&=~  -\kappa_1x_1x_9  +\kappa_2x_2  + \kappa_3x_2 ,  \\
f_{c, \kappa,10} ~&=~ -\kappa_4x_3x_{10} +\kappa_5x_4 +\kappa_6x_4 -\kappa_{10}x_7x_{10}  + \kappa_{11} x_8 + \kappa_{12} x_8 .  
\end{align*}
Consider the following upper-triangular matrix:
{\footnotesize
\[M(\kappa)=
\left(\begin{array}{cccccccccc}
\frac{1}{\kappa_3} &   \frac{1}{\kappa_3} & 0 & -\frac{1}{\kappa_3} & 0&-\frac{1}{\kappa_3}\\
 0&   \frac{1}{\kappa_7}& 0& 0  & 0& 0\\
 0 & 0 &  \frac{1}{\kappa_9}&   \frac{1}{\kappa_9}& 0 & 0\\
 0    &  0 & 0     &  \frac{1}{\kappa_{10}} & 0 & 0\\
 0    &  0 & 0   & 0   &  \frac{1}{\kappa_{1}}& 0\\
 0    & 0 & 0   & 0   & 0       &  \frac{1}{\kappa_{4}}
\end{array}
\right).
\]
} 
Note that ${\rm det} M(\kappa)>0$ for all $\kappa\in \mathbb{R}^{12}_{>0}$. 
Following~\eqref{eq:linearchange}--\eqref{consys-h}, the resulting system is:
\begin{center}
\begin{tabular}{lll}
$h_{c, a,1}  =  x_1 + x_2 + x_3 + x_4 + x_6 - c_1$,&
$h_{c, a,2} = x_2 + x_9 - c_2 $, &
$h_{c, a,3}  =  x_2 - a_1x_4$, \\
$h_{c, a,4}  =  x_4 + x_8 + x_{10} - c_3$, &
$h_{c, a,5}  =  x_5 + x_7 + x_6 + x_8 - c_4$, &
$h_{c, a,6}= x_3x_5 - a_2x_6$, \\
$h_{c, a,7} =x_6-a_3 x_8$,  &
$h_{c, a,8} = x_7x_{10}  -a_4 x_8$,& 
$h_{c, a,9}  =  -x_1x_9  +a_5x_2$, \\
$h_{c, a,10}  = -x_3x_{10} +a_6x_4$ &&
\end{tabular}
\end{center}
where the effective parameters are as follows:
\begin{equation}\label{eq:ep46}
\begin{tabular}{llllll}
$\bar a_1=\frac{\kappa_6}{\kappa_3}$,& $\bar a_2=\frac{\kappa_8 + \kappa_9}{\kappa_7}$,& $\bar a_3=
\frac{\kappa_{12}}{\kappa_9}$,&$\bar a_4=\frac{\kappa_{11} + \kappa_{12}}{\kappa_{10}}$, & $\bar a_5=
\frac{\kappa_{2} + \kappa_{3}}{\kappa_{1}}$, & $\bar a_6=\frac{\kappa_{5} + \kappa_{6}}{\kappa_{4}}~.$
\end{tabular}
\end{equation}
The associated map 
$\mathbb{R}^{12 }_{>0} \to \mathbb{R}^{6}_{>0}$ given by 
$ \kappa\mapsto \bar a(\kappa)$ is surjective. 

Note that the non-conservation-law equations in $\augmentH$ are binomials, so the network is linearly binomial. 
Hence, as in~\eqref{eq:phia}, we use the binomial equations 
$h_{c, a,3}=h_{c, a,6}=h_{c, a,7}=h_{c, a,8}=h_{c, a,9}=h_{c, a,10}=0$
to solve for the $a_k$'s: % {\color{blue}effective parameters}: 
\begin{equation}\label{eq:ex46}
\begin{tabular}{llllll}
$a_1=\frac{x_2}{x_4}$, &
$a_2=\frac{x_3x_5}{x_6}$, & 
$a_3=\frac{x_6}{x_8}$, & 
$a_4=\frac{x_7x_{10}}{x_8}$, & 
$a_5=\frac{x_1x_9}{x_2}$,&
$a_{6}=\frac{x_3x_{10}}{x_4}$~.\\
\end{tabular}
\end{equation}
Substituting (\ref{eq:ex46})  into $\det {\rm Jac}({\augmentH})$ -- that is, considering the 
parametrization $x \mapsto (\phi_a(x) , x)$ where $\phi_a(x)$ is given by~\eqref{eq:ex46}
 --
 yields the following critical function $C(x)$, as in $(\ref{ecf})$:
\[C(x)=\frac{x_3x_{10}}{x_4x_8}B(x)~,\] 
where $B(x)$ is the following critical polynomial:
\begin{center}
{$
x_1x_{10}x_2x_5+x_1x_{10}x_2x_6+x_1x_{10}x_2x_7+x_1x_{10}x_2x_8+x_1x_{10}x_5x_9+
x_1x_{10}x_6x_9+x_1x_{10}x_7x_9+x_1x_{10}x_8x_9\underline{-x_1x_2x_5x_8}+x_1x_2x_7x_8\underline{-x_1x_5x_8x_9}+
x_1x_7x_8x_9+x_{10}x_2x_5x_9+x_{10}x_2x_6x_9+
x_{10}x_2x_7x_9+x_{10}x_2x_8x_9+x_{10}x_3x_5x_9+x_{10}x_3x_6x_9+x_{10}x_3x_7x_9+
x_{10}x_3x_8x_9+x_{10}x_4x_5x_9+x_{10}x_4x_6x_9+
x_{10}x_4x_7x_9+x_{10}x_4x_8x_9+x_{10}x_5x_6x_9\underline{-x_2x_5x_8x_9}+x_2x_7x_8x_9+
x_3x_4x_5x_9+x_3x_4x_6x_9+x_3x_4x_7x_9+x_3x_4x_8x_9+
x_3x_7x_8x_9+x_4x_5x_6x_9\underline{-x_4x_5x_8x_9-x_4x_6x_7x_9}+x_4x_7x_8x_9
$~.
}
\end{center}
Consistent with Lemma \ref{ecfb}, $B(x)$ is homogeneous with total degree $d=4$ and 
square-free monomials.

From the conservation laws~\eqref{eq:ex45con}, 
we see that this network is conservative and hence dissipative. 
It is also straightforward to check (for instance, using criteria in \cite{SS2010,messi}) that it has no boundary steady states. 
Thus, we can follow Procedure~\ref{proc:witness} to find a witness as follows:

\underline{Step 1}.
We compute the following sign:
\[(-1)^{\mathrm{rank}(N)+1}~=~(-1)^{s-d+1}~=~(-1)^{10-4+1}~=~ -1~.\]
We see that $B(x)$ has five (underlined) terms with the above negative sign; one
such term is: \[-x_1x_2x_5x_8~.\] 
Accordingly, 
define 
$x^*\in {\mathbb R}^{10}$
with coordinates $\lambda$ (in indices $1,2,5,$ and $8$) and 1 (all others):
\begin{align} \label{eq:x-lambda}
x^*~=~
 \left(
  \lambda,~  \lambda, ~ 1, ~ 1, ~ \lambda, ~ 1, ~ 1, ~ \lambda, ~ 1,~  1
 \right)~.\end{align}
Then we have
\begin{align} \label{eq:B-phos}
B(x^*)~=~
-\lambda^4+\lambda^3 + 7\lambda^2 + 14\lambda + 5~.\end{align}
It follows that $B(x^*)<0$ if $\lambda$ is larger than the largest positive root of the polynomial~\eqref{eq:B-phos}. 
There are many well-known upper bounds for the real roots of a univariate  polynomial. Here, we use an elementary bound,  
the sum of the absolute values of all coefficients:
\[1+1+7+14+5 = 28~.\]
Let $\lambda = 29$; then, $B(x^*)|_{\lambda=29}=-676594<0$.

\underline{Step 2}.
To solve for $c^*$,
we substitute  $x^*|_{\lambda=29}$, as in~\eqref{eq:x-lambda},
into equation~\eqref{eq:ex45con}, which yields:
\begin{align*}
c^* ~=~ (61, 30, 31, 60)
\end{align*}

\underline{Steps 3--4}. 
We substitute  $x^*|_{\lambda=29}$, as in~\eqref{eq:x-lambda},
into equation~\eqref{eq:ex46}.  This yields:
\begin{align*}
\quad {\rm and} \quad
\phi_a(x^*)
 ~=~ (29,~ 29,~ 1/29,~ 1/29,~1,~1)~.
\end{align*}

Finally, we choose $\kappa^*$ for which $ \bar a (\kappa^*) = \phi_a(x^*)$, 
%a^*$, 
as in (\ref{eq:ep46}):
 \begin{align*}
\kappa^* ~=~ (2,~1,~1,~30,~1,~29,~1,~28,~1,~2,~ 1/29,~ 1/29)~.
\end{align*}
So, $(\kappa^*, c^*)$ is a witness to multistationarity. 
\end{example}

\subsection{Open multistationarity regions in parameter space} \label{sec:open-map}

We saw in Procedure~\ref{proc:witness} that, for linearly binomial networks,
 finding $x^*$ with ${\rm sign}(B(x^*))=(-1)^{\mathrm{rank}(N)+1}$
allows us to obtain a witness to multistationarity 
$(\kappa^*,c^*)$ satisfying $\varphi(\kappa^*,c^*)=\psi(x^*)$ for the maps:
\begin{align*}
	 \mathbb{R}^{s}_{>0} ~		&\overset{\psi}\to ~  \mathbb{R}_{>0}^{s-d} \times \mathbb{R}_{>0}^{d}
	~~ \overset{\varphi}\leftarrow ~
	\mathbb{R}_{>0}^m \times \mathbb{R}_{>0}^d~ \\ 
	x ~ &\mapsto 
		\quad \quad (a,c) \quad \quad
		 \mapsfrom ~ (\kappa,c)~,
	\end{align*}
where $\psi: x ~\mapsto~ (a,c)$ is given by:
\begin{equation}\label{eq:openmap}
\begin{cases}
a_k =~\psi_k~=~x^{\gamma_k-\delta_k}, &k~=~1,2, \dots, s-d\\
c_k =~\psi_{s-d+k} ~=~(Wx)_k, & k~=~1,2,  \ldots, d~,
\end{cases}
\end{equation}
and $\varphi$ is given by 
\begin{equation} \label{eq:phi}
\varphi(\kappa, c) ~:=~ (\bar a(\kappa),c)~.
\end{equation}
%(Recall by (\ref{eq:ep}) that the effective parameters $a_k$'s are rational functions in the rate constants $\kappa_i$'s.)

Therefore, the region in the parameter space 
$\mathbb{R}_{>0}^m \times \mathbb{R}_{>0}^d$ 
of the $(\kappa,c)$'s where degree theory guarantees multistationarity is
\begin{equation} \label{eq:region}
	\varphi^{-1}(\psi(U))~,
\end{equation}
where 
\begin{equation}\label{eq:u}
U~=~
	\left\{ x\in {\mathbb R}_{>0}^s\;|\; {\rm sign}(B(x))~=~(-1)^{\mathrm{rank}(N)+1} \right\}~.
\end{equation}
We will show that the set $\psi(U)$ is open (Theorem~\ref{thm:openmap}), 
and thus so is $\varphi^{-1}(\psi(U))$, our region of interest (Corollary~\ref{cor:openmap}).  
For ``typical'' networks, the multistationarity regions of parameter space are full-dimensional.
The interpretation for applications is that multistationarity persists  under
 small perturbations of the rate constants and the total-constant values (equivalently, initial values). Such robustness properties are desirable in biological systems. In our setting, it is not immediately evident that we find an open region in the parameter space of the variables $(c,\kappa)$ because we detect a multistationarity region in the $x$ space and not in the parameter space.

%\begin{equation}\label{eq:openmap2}
%	\varphi:~ \mathbb{R}^{m}_{>0}\times {\mathbb R}_{>0}^{d} ~\rightarrow~  \mathbb{R}_{>0}^{s}, \;\;\;\;
%		\varphi(\kappa, c) ~=~ (a(\kappa),c)
%\end{equation}

Before stating the results in this section,
we show in Figure~\ref{fig:openregion} 
a ``slice'' of the set $U$, as in~\eqref{eq:u}, arising from the network in 
 Example \ref{cascade}.  More precisely, starting with the critical polynomial $B$ from that example, we display the (open) region where $B<0$, under the 
 following specialization:
\begin{align} \label{eq:specialize-B}
x_3=x_4=1~,\quad 
 x_5=x_8=6~, \quad x_6=x_7=1~, \quad x_9=x_{10}=1~.
\end{align}
This region, in $(x_1,x_2)$-space, is guaranteed by degree theory to yield multistationarity.

\begin{center}
\begin{figure}[ht]
\includegraphics[width=0.75\linewidth]{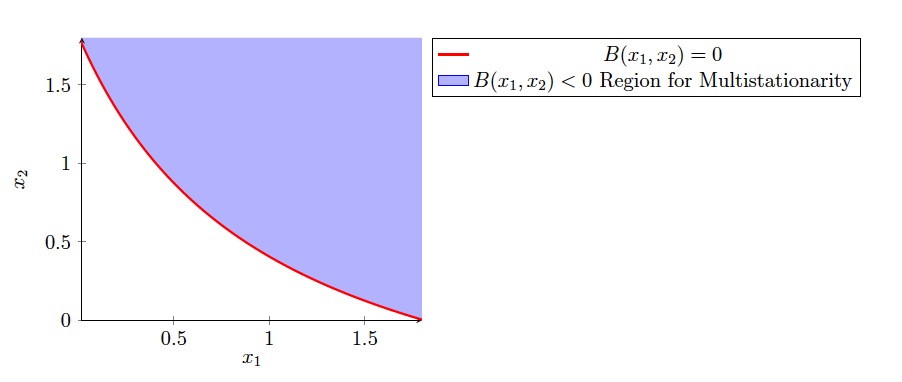}
\caption{The shaded region depicts a ``slice'' of the multistationarity region $U$, as in~\eqref{eq:u},
for the network in 
 Example \ref{cascade}.  The slice arises from the specialization~\eqref{eq:specialize-B}.
  \label{fig:openregion}}
\end{figure}
\end{center}

% TEMPORARILY COMMENTED THIS OUT B/C MY COMPUTER IS MISSING A PKG (ANNE)
\begin{comment}
\begin{center}
\begin{figure}[ht]
\begin{tikzpicture}
\begin{axis}[legend pos=outer north east,
    ymin=0, ymax=1.8,
    axis lines = left,
    xlabel = $x_1$,
    ylabel = {$x_2$},
]
\addplot[
name path=B, 
very thick,
	domain=0:1.8, 
    samples=100, 
    color=red
    ]{(-16*x+29)/(16*x+16)};
\addlegendentry{$B(x_1, x_2)=0$}
\path[name path=axis] (axis cs:1.8, 0) -- (axis cs:1.8, 2);
\addlegendentry{$B(x_1, x_2)<0$ Region for Multistationarity}
\end{axis}
\end{tikzpicture}
\caption{Specialize $B(x)$ in the Example \ref{cascade} with
 $x_3=x_4=1, x_5=x_8=6, x_9=x_{10}=1$ and get
$B(x_1, x_2)=16x_1x_2-16x_1-16x_2+29$. Then all the positive points $(x_1, x_2)$
admitting multistationarity are given by the
open set $\{(x_1, x_2)\in {\mathbb R}_{>0}^2 \mid B(x_1, x_2)<0\}$, see the colored region. 
  \label{fig:openregion}}
\end{figure}
\end{center}
% 	END OF COMMENT
\end{comment}

% LEMMA
\begin{lemma}\label{lm:openmap}
Let $G$ be
 a linearly binomial network with critical polynomial $B(x)$ and map $\psi$ 
 as in~(\ref{eq:openmap}). 
Then for every $x^*\in {\mathbb R}_{>0}^s$, we have
$\det {\rm Jac}~ \psi |_{x=x^*} 
\neq 0$ 
if and only if $B(x^*)\neq 0$. 
\end{lemma}
%PROOF
\begin{proof}
It is straightforward to verify, using equations (\ref{toric}) and (\ref{eq:openmap}),
 the following equality: 
\[ {\rm Jac}~ \psi ~ = ~ 
	\pm 
\left({\frac{1}{\Pi_{k=1}^{s-d}x^{\delta_{k}}}}\right) \cdot {\rm Jac}\left(\augmentH\right)|_{a_k=x^{\gamma_k-\delta_k}}, \;\;\;\;k~=1,2, \ldots, s-d. 
\]
Thus, for every $y\in {\mathbb R}_{>0}^s$ we obtain the first equality here:
\[
\det {\rm Jac}~ \psi |_{x=y} ~=~
	\pm 
\left({\frac{1}{\Pi_{k=1}^{s-d} y^{\delta_{k}}}}\right) \cdot \left(\det \left({\rm Jac}\left(\augmentH\right)\right)|_{a_k=x^{\gamma_k-\delta_k}}\right)|_{x=y}
~=~
	\pm 
{\frac{{C}(y)}{\Pi_{k=1}^{s-d} y^{\delta_{k}}}}~=~
	\pm 
 \frac{y^{\alpha}}{y^{\beta}}\cdot {\frac{B(y)}{\Pi_{k=1}^{s-d} y^{\delta_{k}}}}~, \]
and the second and third equalities arise, respectively,
from equation~(\ref{ecf}) and Lemma~\ref{lm:ecfstructure}.
Thus, for every $y \in {\mathbb R}_{>0}^s$, we have 
$\det {\rm Jac}~ \psi |_{x=y} 
\neq 0$ 
if and only if $B(y)\neq 0$. 
\end{proof}

%THEOREM
\begin{theorem}\label{thm:openmap}
Let $G$ be
 a  linearly binomial network with critical polynomial $B(x)$ and map $\psi$ 
 as in (\ref{eq:openmap}). 
Let $U$ be as in~\eqref{eq:u}. 
%\begin{equation}\label{eq:u}
%U~=~\{x\in {\mathbb R}_{>0}^s\;|\; sign(B(x))~=~(-1)^{\mathrm{rank}(N)+1}\}.
%\end{equation}
Then $\psi(U)$ is an open set in the $(c,a)$-space ${\mathbb R}_{>0}^{s}$. 
\end{theorem}
\begin{proof}
Note that $B(x^*) \neq 0$ for any $x^*\in U$. So, we deduce from
Lemma~\ref{lm:openmap} that the Jacobian of $\psi$ never vanishes on $U$. The
result now follows from the Inverse Function Theorem.
\end{proof}

% COROLLARY
\begin{corollary}\label{cor:openmap}
Let $G$ be a linearly binomial network, with maps
 $\psi$  and $\varphi$ 
 as in \eqref{eq:openmap}--\eqref{eq:phi}. 
Then the 
multistationarity 
region of parameter space  
$\varphi^{-1}\left(\psi(U)\right)$ as in~\eqref{eq:region}, 
is an open set in the $(\kappa, c)$-space ${\mathbb R}_{>0}^m\times {\mathbb R}_{>0}^{d}$. 
\end{corollary}
\begin{proof}
The map $\varphi$ is continuous, and, 
by Theorem \ref{thm:openmap}, $\psi(U)$ is open.
Hence, $\varphi^{-1}\left(\psi(U)\right)$ is open.
\end{proof}

\section{Many MESSI networks are linearly binomial networks}\label{sec:messi}
%\section{Binomials in MESSI networks}\label{sec:messi}
In this section, we show that a class of networks that includes many biological signaling networks 
are linearly binomial networks (Theorem~\ref{thm:binom}).  These networks are so-called MESSI networks~\cite{messi}, whose definition we recall below.

\subsection{Definition of MESSI systems}\label{sec:defmessi}

%\medskip

A chemical reaction network with species set
$\Sp$ is a {\it MESSI network} if there is a partition 
\begin{equation}\label{eq:S}
\Sp=\Sp^{(0)}\bigsqcup \Sp^{(1)} \bigsqcup \Sp^{(2)} \bigsqcup \dots \bigsqcup \Sp^{(m)},
\end{equation}
where $m \ge 1$ and $\bigsqcup$ denotes disjoint union, such that the complexes and reactions satisfy the 
conditions below.
Species in $\Sp^{(0)}$ and $\Sp_1:= \Sp \setminus\Sp^{(0)}$ are called, respectively, {\em intermediate} and {\em core}.
%When convenient, we will distinguish intermediate and core species in the notation in the following way:
%$\Sp^{(0)}=\{U_1,\dots, U_p\}$, $\Sp_1 =\{X_1,\dots,X_n\}$. 
Complexes are also partitioned into two disjoint sets: {\em intermediate complexes} and {\em core complexes} \cite{fw13}.
Each intermediate complex consists of a unique intermediate species.
% that appears in that complex. 

%The vector corresponding to the complex $U_i$ is denoted by $y_i$. 
Core complexes satisfy the following two conditions:
\begin{enumerate}[(i)]
\item They are molecular or bimolecular and consist of either one or two core species. 
%If the core complex consists of only the species $X_i$.
%, the corresponding vector will  be denoted by $y_{i0}$.  
\item If the core complex consists of two species $X_i, X_j$, they {\em must} belong to \emph{different} sets 
$\Sp^{(\alpha)}, \Sp^{(\beta)}$ (with $\alpha \neq \beta$ and $\alpha, \beta \geq1$).
% We also denote the complex $X_i+ X_j = X_j+ X_i$ by $y_{ij}=y_{ji}$.
\end{enumerate}

We say that complex $y$ reacts to complex $y'$ \emph{via intermediates} if either $y\to y'$ or
there exists a path of reactions from $y$ to $y'$ {\em only through intermediate complexes}.  
This is denoted  by $y \uri y'$.
The intermediate complexes of a MESSI network satisfy moreover the following condition. 
For every intermediate complex $y$, there exist core
complexes $y'$ and $y''$ such that $y'\uri y$ and $y\uri y''$.

% \begin{itemize}
% \item[$(\cond)$]  For every intermediate complex $y_k$, there exist core
% complexes $y_{ij}$ and $y_{\ell m}$ such that $y_{ij}\uri y_k$ and $y_k\uri y_{\ell m}$.
% \end{itemize}

\medskip

The reactions of MESSI networks are constrained by the following rules:
\begin{enumerate}[(i)]
\item  If three species are related by $X_i+ X_j \uri X_k$ or $X_k \uri X_i + X_j$, 
then $X_k$ is an intermediate species.
 \item If two core species $X_i, X_j$ 
 are related by $X_i\uri X_j$, then there exists $\alpha \ge 1$ such that both belong to $\Sp^{(\alpha)}$.
 \item If $X_i+X_j\uri X_k+X_\ell$, then 
 there exist $\alpha \neq \beta$ such that $X_i,X_k \in \Sp^{(\alpha)}$,
$X_j,X_\ell \in
\Sp^{(\beta)}$ or $X_i,X_\ell \in \Sp^{(\alpha)}$, $X_j,X_k \in
\Sp^{(\beta)}$.
\end{enumerate}

The partition~\eqref{eq:S} \emph{defines a MESSI structure} on
the network.
% \end{definition}
A {\em MESSI system} is the mass-action kinetics dynamical system~\eqref{sys} associated with a MESSI network. 

\subsection{The associated digraphs}\label{ssec:associated_digraphs}

%\smallskip
\medskip

We now present the digraphs $G_1$, $G_2^\circ$, and $G_E$ associated to a MESSI network $G$ through an example. 
For the actual definition we refer the reader to \cite{messi}.

\begin{example}[Two-layer cascade, continued] \label{ex:2-layer-again}
Recall the two-layer cascade network in Example~\ref{cascade}. 
We consider the following partition $\Sp^{(0)}=\{ES_0,FS_1,S_1P_0, FP_1\}$  (intermediate species), and
$\Sp^{(1)}=\{S_0,S_1\}$, $\Sp^{(2)}=\{P_0,P_1\}$, $\Sp^{(3)}=\{E\}$, $\Sp^{(4)}=\{F\}$ (partition of the core species).
The intermediate complexes correspond to the intermediate species, and the remaining complexes are core complexes. 
This partition defines a MESSI structure in the network.

We  first define a digraph $G_1$ by keeping the core complexes as vertices and considering the edges
$y\to y'$ if $y \uri y'$  in $G$. The labels assigned to these edges, $\Katau(\kappa)$,
are rational functions of the original rate constants $\kappa$, following~\cite[Theorem~3.1]{fw13}:
\begin{align}\label{eq:katau}
\Katau ~:~\mathbb{R}^m_{>0} & ~\to~ \mathbb{R}^{r}_{>0} \\ \notag
                \kappa & ~\mapsto~ (\tau_1(\kappa), \tau_2(\kappa), \dots, \tau_{r} (\kappa))~,
\end{align}
where $r$ denotes the number of edges in $G_1$.
We then define a new digraph $G_2$ where we ``hide'' the concentrations of some of the species in the labels.
We keep all monomolecular reactions $X_i\to X_j$ and for each reaction $X_i+X_\ell \overset{\tau}{\longrightarrow} X_j+X_m$,
with $X_i,X_j \in \Sp^{(\alpha)}$, $X_\ell,X_m \in \Sp^{(\beta)}$,
we consider two reactions $X_i \overset{\tau x_\ell}{\longrightarrow} X_j$
and $X_\ell \overset{\tau x_i}{\longrightarrow} X_m$.
We obtain a multidigraph $MG_2$ that may contain loops or parallel edges between some pairs of nodes
(i.e., directed edges with the same source and target nodes).
We define the digraph $G_2$ by collapsing into one edge all parallel edges in $MG_2$ and we define the
labels of each edge as the sum of the labels of the corresponding collapsed edges in $MG_2$.
Note that these labels might depend on some of the concentrations.
We will moreover denote by $G_2^\circ$ the digraph obtained from $G_2$ by deleting loops and isolated nodes.
We finally define the associated digraph $G_E$. The set of vertices of $G_E$ equals $\{\Sp^{(\alpha)}\mid \, \alpha \ge 1\}$.
The pair $(\Sp^{(\alpha)},\Sp^{(\beta)})$ is an edge of $G_E$ when there is a species
in $\Sp^{(\alpha)}$ in a label of an edge in $G_2^\circ$ between (distinct) species of $\Sp^{(\beta)}$.

The graphs $G_1$, $G_2^\circ$, and $G_E$ associated to this network are the following:

 \begin{tabular}{|>{\centering}m{1.3in} >{\centering}m{0.1in}  >{\centering}m{1.3in} >{\centering\arraybackslash}m{1.6in}|}
 \hline
 \multicolumn{1}{|l}{$G_1$:} &  & \multicolumn{1}{l}{$G_2^\circ$:} & \multicolumn{1}{l|}{$G_E$:}\\
 $\begin{array}{l}
  S_0+E \overset{\tau_1}{\rightarrow} S_1+E \\
  S_1+F \overset{\tau_2}{\rightarrow} S_0+F \\
  P_0+S_1 \overset{\tau_3}{\rightarrow} P_1+S_1\\
  P_1+F \overset{\tau_4}{\rightarrow} P_0+F
\end{array}$ & 
\text{\Large \textcolor{blue}{$\Rightarrow$}} & 
$\begin{array}{l}
S_0\underset{\tau_2x_{10}}{\overset{\tau_1x_{9}}{\rightleftarrows}} S_1\\
P_0\underset{\tau_4x_{10}}{\overset{\tau_3x_3}{\rightleftarrows}} P_1
\end{array}$ & 
\begin{tikzpicture}[ampersand replacement=\&] 
\matrix (m) [matrix of math nodes, row sep=1.5em, column sep=1em, text height=1.5ex, text depth=0.25ex]
{ \Sp^{(3)}\&  \Sp^{(1)} \& \Sp^{(2)} \\
\Sp^{(4)} \&  \& \\};
\draw[->]($(m-1-1)+(0.4,0)$) to node[below] (x) {} ($(m-1-2)+(-0.4,0)$);
\draw[->]($(m-2-1)+(0.4,0)$) to node[below] (x) {} ($(m-1-2)+(-0.35,-0.15)$);
\draw[->]($(m-1-2)+(0.4,0)$) to node[below] (x) {} ($(m-1-3)+(-0.4,0)$);
\draw[->]($(m-2-1)+(0.4,0)$) to node[below] (x) {} ($(m-1-3)+(-0.4,-0.1)$);
\end{tikzpicture}\\ 
& & & \\

\hline
\end{tabular}\\

Here, as defined in \eqref{eq:katau}, $\Katau: \mathbb{R}^{12}_{>0}  ~\to~ \mathbb{R}^{4}_{>0}$ is 
\begin{equation*}%\label{eq:tau_2}
\begin{tabular}{llll}
$\tau_1=\frac{\kappa_1}{\kappa_2+\kappa_3}$,\;\;& $\tau_2=\frac{\kappa_4}{\kappa_5+\kappa_6}$,\;\;& $\tau_3=
\frac{\kappa_7}{\kappa_8+\kappa_9}$,\;\;&$\tau_4=\frac{\kappa_{10}}{\kappa_{11}+\kappa_{12}}$.
\end{tabular}
\end{equation*}
\end{example}

\begin{remark}\label{rem:G1}
  An important fact is that for any MESSI network with digraph $G$, once the edges in $G_1$ are labeled with the 
  constants $\Katau(\kappa)$, the steady states of the mass-action chemical reaction system defined by $G$ and those of $G_1$ 
  are in one-to-one correspondence. Moreover, $G_1$ and $G_2^\circ$, together with the corresponding equations 
  of the intermediate species, define the whole variety of steady states of $G$.
\end{remark}

\subsection{Main result: conditions for MESSI linearly binomial networks}

 In this section we give sufficient conditions on a MESSI system
that ensure that the network is linearly binomial (Definition~\ref{def:lbin}).
We  define a further condition
$(\cond)$:  For every intermediate complex $y$ there exists a \emph{unique} core complex $y'$ such that $y'\uri y$ in $G$.

\begin{theorem}\label{thm:binom}
Let $G$ be the underlying digraph of a MESSI system with $m$  reactions ($m$ directed edges). Assume that $G$ satisfies condition~$(\cond)$ and that the associated digraph $G_E$ has no directed cycles, the underlying undirected graph 
of the associated graph $G_2^\circ$ is a forest (an acyclic graph), and $MG_2$ has no parallel edges. 
 Then, there exist:
\begin{enumerate}
\item $\bar a_1(\kappa), \bar a_2(\kappa), \dots, \bar a_{\bar m}(\kappa) \in \mathbb{Q}(\kappa)$ 
such that $\bar a_i(\kappa^*)$ is defined and, moreover, $\bar a_i(\kappa^*)>0$ for every 
$i=1,2,\dots,\bar m$ and for all $\kappa^* \in \mathbb{R}^m_{>0}$, and 
\item  an $(s-d)\times(s-d)$ matrix $M(\kappa)$ with $\det M(\kappa^*)>0$ for all 
$\kappa^* \in \mathbb{R}^m_{>0}$,
\end{enumerate}	
such that the functions %$\bar{h}_{j_1},\dots,\bar{h}_{j_{s-d}}$: %in~\eqref{eq:linearchange}
\begin{equation*}%\label{eq:linearchange}
(\bar h_{j_1}, \bar h_{j_2}, \dots, \bar h_{j_{s-d}})^{\top}
\quad := \quad M(\kappa) ~  (f_{j_1}, f_{j_2}, \dots, f_{j_{s-d}})^{\top}~,
\end{equation*}
 are binomials, where $f$ is the polynomial system obtained from the ODEs \eqref{sys} of 
 the network $G$,  and every nonconstant coefficient in $\bar h_{j_l}$ is equal to a rational-number multiple of some $\bar a_i(\kappa)$.
  Therefore, if the map $\bar a$ is surjective, then $\bar a$ is a reparametrization map, as in 
 \eqref{eq:ep}, and $\bar{h}_{j_1},\dots,\bar{h}_{j_{s-d}}$ is an effective steady-state function 
 $\augmentH(x)$ of $G$, as in \eqref{eq:effective-reln}--\eqref{consys-h}, and so 
 $G$ is a linearly binomial network. 
\end{theorem}

We prove Theorem~\ref{thm:binom} in Appendix~\ref{sec:pf-thm-binom}.

\begin{remark}
The extra hypothesis in Theorem~\ref{thm:binom} that guarantees that $G$ is a linearly binomial 
network -- namely, the condition that the map $\bar a$ is surjective -- holds for every example we have examined. 
See, for instance, Examples \ref{CFMW612} and \ref{cascade}.
\end{remark}

\begin{example}[Two-layer cascade, continued] \label{ex:2-layer-again-again}
æ Recall the two-layer cascade network in Examples~\ref{cascade} and~\ref{ex:2-layer-again}. Notice that $G$ satisfies
condition~$(\cond)$ and that the associated digraph $G_E$ has no directed cycles (see Example~\ref{ex:2-layer-again}). Moreover, the
underlying undirected graph of the associated graph $G_2^\circ$ is a forest:
æ \[S_0 - S_1 \qquad P_0 -æ P_1\]
æand $MG_2$ has no parallel edges.
Hence, Theorem~\ref{thm:binom} applies, and in fact the resulting
binomials are the ones we computed in Example~\ref{cascade}.
æ\end{example}

\section{Establishing multistationarity using triangular forms} \label{sec:triangular}
For networks that are non-dissipative or have boundary steady states, 
Theorems~\ref{thm:c-general} and~\ref{thm:NP}
do not apply. 
Accordingly, in this section, we propose a method to analyze such networks. 
The main idea is to find a degenerate positive steady state $x^*$ and then 
perturb the corresponding parameters.
The hope is that 
$x^*$ will break into two distinct positive steady states. 
 We prove that this approach will succeed under certain conditions 
 (Theorem~\ref{prop:triangular}).
Specifically, we require that the steady-state equations admit a ``triangular form'' (Definition~\ref{def:triangular}).  
Finally, we show that 
such a triangular form exists whenever the steady-state equations form a ``general zero-dimensional system'' 
(Corollary~\ref{cor:tri}).

Before stating our results, we introduce our running example, which we first show has boundary steady states (so, our earlier results do not apply).
\begin{example}[Calvin cycle]\label{cc}
Consider the following 
 ``elementary mode''
 of the Calvin cycle network 
proposed in~\cite[Fig.\ 4]{GAKKSN2011}:
\begin{center}
\begin{tabular}{ll}
\ce{RuBP + E_1
->[\kappa_1]
RuBPE_1
->[\kappa_2]
2\cdot PGA + E_1
},
&
\ce{
PGA + E_2
->[\kappa_3]
PGAE_2
->[\kappa_4]
DPGA + E_2},\\
\ce{
DPGA + E_3
->[\kappa_5]
DPGAE_3
->[\kappa_6]
GAP + E_3
},
&
\ce{
5\cdot GAP + E_4
->[\kappa_{7}]
GPAE_4 
->[\kappa_{8}]
3\cdot Ru5P + E_4
},\\
\ce{
Ru5P + E_5
->[\kappa_{9}]
Ru5PE_5
->[\kappa_{10}]
RuBP + E_5
},
&
\ce{
GAP + E_7
->[\kappa_{11}]
GAPE_7
->[\kappa_{12}]
E_7}.
\end{tabular}
\end{center}
This network is obtained by shutting down  $9$ transporter reactions from the original Calvin cycle network (see $``v_1^{EM}"$ in 
 \cite[page 218]{GAKKSN2011}). %, 
Let
\begin{center}
\begin{tabular}{llllll}
$X_1$=\ce{RuBP},& $X_2$=\ce{E_1},& $X_3$=\ce{RuBPE_1}, & $X_4$=\ce{PGA}, & $X_5$=\ce{E_2}, & $X_6$=\ce{PAGE_2},\\
$X_7$=\ce{DPGA}, & $X_8$=\ce{E_3}, & $X_9$=\ce{DPGAE_3},   & $X_{10}$=\ce{GAP}, & $X_{11}=\ce{E_4}$,   &$X_{12}=\ce{GAPE_4}$,\\
$X_{13}$=\ce{Ru5P},& $X_{14}$=\ce{E_5},& $X_{15}$=\ce{Ru5PE_5},                             & $X_{16}$=\ce{E_7},& $X_{17}$=\ce{GAPE_7}. &
\end{tabular}
\end{center}
The function $\F(x)$ is 
\begin{center}
\begin{tabular}{lll}
$f_{c, \kappa,1} = -\kappa_1x_1x_2+\kappa_{10}x_{15}$,  & &$f_{c, \kappa,2} = x_2 + x_3 - c_1$, \\
$f_{c, \kappa,3} =  \kappa_1x_1x_2-\kappa_2x_3$,  & &$f_{c, \kappa,4} = 2\kappa_2x_3 -\kappa_3x_4x_5$, \\
$f_{c, \kappa,5} = x_5 + x_6 - c_2$, &  &$f_{c, \kappa,6}= \kappa_3x_4x_5 - \kappa_4x_6$, \\
$f_{c, \kappa,7}= \kappa_4x_6-\kappa_5x_7x_8$, & &$f_{c, \kappa,8} = x_8 + x_9 - c_3 $, \\
$f_{c, \kappa,9} = \kappa_5x_7x_8 - \kappa_6x_9$, & &$f_{c, \kappa,10} =\kappa_6x_9 -5\kappa_{7}x_{10}^5x_{11}-\kappa_{11}x_{10}x_{16}$, \\
$f_{c, \kappa,11} =x_{11}+x_{12}-c_4 $, & &$f_{c, \kappa,12} =\kappa_{7}x_{10}^5x_{11}-\kappa_{8}x_{12}$,\\
$f_{c, \kappa,31} =3\kappa_{8}x_{12}-\kappa_{9}x_{13}x_{14}$, & &$f_{c, \kappa,14} =x_{14} + x_{15} - c_5$,\\
$f_{c, \kappa,15} =\kappa_{9}x_{13}x_{14}-\kappa_{10}x_{15}$, & &$f_{c, \kappa,16} =x_{16} + x_{17} - c_6$, \\
$f_{c, \kappa,17} =\kappa_{11}x_{10}x_{16}-\kappa_{12}x_{17}$.
\end{tabular}
\end{center}
If we set $x_1=x_3=x_4=x_6=x_7=x_9=x_{10}=x_{12}=x_{13}=x_{15}=x_{17}=0$, the polynomials $f_{c, \kappa, i}$ (for $i=1,2,\ldots, 17)$ above become
\[0,~ x_2-c_1,~ 0,~ 0,~ x_5-c_2,~ 0,~ 0,~ x_8 - c_3,~ 0,~ 0,~ x_{11}-c_4,~ 0,~ 0,~ x_{14}-c_5,~ 0,~ x_{16}-c_6,~ 0~.\]
Thus, for any rate-constant vector $\kappa$ and for any total-constant vector $c$, 
%$\F(x)=0$ has 
we have the following boundary steady state:
\[(0, c_1, 0, 0, c_2, 0, 0, c_3, 0, 0, c_4, 0, 0, c_5, 0, c_6, 0)~.\]
Thus, Theorem \ref{thm:c-general} does not apply to the Calvin cycle network. 
\end{example}

% SUBSECTION: TRI. FORM
\subsection{Steady-state equations that admit a triangular form}
In this section, we investigate multistationarity for networks whose steady-state equations admit a triangular form (Theorem~\ref{prop:triangular}).

\begin{definition} \label{def:triangular}
Let $G$ be a reaction network with $s$ species. 
The steady-state equations of~$G$ {\em admit a triangular form} if there exists a  system $\augmentH(x)=0$ \eqref{consys-h} for $G$ 
and there exist functions $T_1, T_2,\ldots, T_s$ of the following form:
		\begin{align*}
		T_s~&=~\theta_s (c, a, x_s),  \\
		T_{s-1} ~&=~ x_{s-1} - \theta_{s-1} (c, a, x_s), \\
		T_{s-2} ~&=~ x_{s-2} - \theta_{s-2} (c, a, x_{s-1}, x_s), \\
		&~~ \vdots \\
		T_1 ~&=~ x_1 - \theta_1 (c, a, x_2, \dots, x_s)~ , 
		\end{align*}
such that:
\begin{enumerate}[(i)]
	% ITEM
	\item each 
	$\theta_i: 
		\mathbb{R}^{d}_{>0} \times  \mathbb{R}^{\bar m}_{>0}  \times \mathbb{R}^{s-1}_{>0}
			 \to \mathbb{R} $ 
	is a $C^2$-function, 
 $\theta_s$ does not depend on $x_1,x_2,\dots, x_{s-1}$, and 
 	(for $1 \leq i \leq s-1$) $\theta_i$ does not depend on $x_1, x_2, \dots, x_i$; and
	% ITEM
	\item there exists a variety
	 $\mathcal{W} \subsetneq   \mathbb{C}^{d} \times  \mathbb{C}^{\bar m}$
	 such that  
	for all $(c^*,{a}^*)  \in \mathbb{R}_{>0}^{d} \times \mathbb{R}_{>0}^{\bar m} \setminus \mathcal{W} $, 
	the positive zeros  of $h_{c^*,a^*}$ coincide with the positive zeros
	of the system \[T_s(c^*, a^*, x_s),~\dots~,~ T_1 (c^*, a^*, x_1,x_2, \dots, x_{s}).\] 
	%$\mathbb{R}_{>0}^s$ holds: 
%\begin{align*} 
% \V_{>0}\left( h_{c^*,a^*}(x_1,x_2, \dots, x_s)\right) 
%	~=~ \V_{>0} \left(T_s(c^*, a^*, x_s),~\dots~,~ T_1 (c^*, a^*, x_1,x_2, \dots, x_{s}) \right)~.
%\end{align*}
%where $\V_{>0}(\cdot)$ denotes the positive zero-sets in the domain over which all the $C^2$-functions are defined. 
\end{enumerate}
\end{definition}

We recall the standard notion of singular point and we state a useful lemma (Lemma \ref{lem:4}) that we will prove in Appendix~\ref{sec:pf-lem-4} using the Implicit Function Theorem.

\begin{definition}\label{def:multiplicity}
 Consider a polynomial $f \in {\mathbb R}[b, z]$, where $(b, z)\in  {\mathbb R}^{n}\times {\mathbb R}$. We say $(b^*, z^*)\in {\mathbb R}^{n}\times {\mathbb R}$ is a {\em singular point} of $f$ if
$f(b^*, z^*)=0~, \;\;\frac{\partial f}{\partial z}(b^*, z^*)=0$, and 
$\frac{\partial f}{\partial b_i}(b^*, z^*)=0, \text{ for all } i=1, \dots, n$.
We say $(b^*, z^*)$ is a {\em regular point} of $f$ if
$f(b^*, z^*)=0$ and $(b^*, z^*)$ is not a singular point of $f$. 
Given a univariate polynomial $f \in {\mathbb R}[z]$, we say that $z^*\in {\mathbb R}$ is a {\em multiplicity-$2$ solution} of $f(z)=0$ if 
$f(z^*)=0~, \;  \frac{d f}{d z}(z^*)=0$, and $\frac{d^2 f}{d z^2}(z^*)\neq 0.$
\end{definition}

\begin{lemma} \label{lem:4}
Consider a  $C^2$-function $f: \mathbb{R}^{n+1} \to \mathbb{R}$. % \in {\mathbb R}[b_1,\dots,b_m, z]$.   
%We call the variables $(a,z)=(a_1,\dots, a_m,z)$.
Assume that $(a^*,z^*) = (a_1^*,a_2^*,\dots, a_n^*,z)\in \mathbb{R}^{n+1}$ satisfies the following:
	\begin{enumerate}
	\item $z^*\in {\mathbb R}$ is a multiplicity-$2$ solution of $f(a^*, z)=0$,  and 
% 	\item $(b^*, z^*)$ is a regular point of $f(b, z)=0$.
	\item  there exists an index $\ell \in [n]$ such that $\frac{\partial f}{\partial a_\ell}(a^*,z^*)\neq0$ {\em (}i.e. $(a^*,z^*)$ is 
	a regular point of $f${\em )}.
	\end{enumerate}
Then for every $\epsilon>0$,  there exists $\delta>0$ such that for all $\delta' \in (0,\delta)$, for either $a_\ell^{**}=a_\ell^*-\delta'$ 
or $a_\ell^{**}=a_\ell^*+\delta'$, the equation $f(a^{**}, z)=0$, where we set $a^{**}_i=a^*_i$ for all $i\neq \ell$, has two distinct real 
solutions $z^{(1)}$ and $z^{(2)}$, for which $|z^{(1)}-z^*|<\epsilon$ and $|z^{(2)}- z^*|<\epsilon$. 
\end{lemma}

We are ready to present a result that uses the existence of a triangular form to find witnesses of multistationarity. 

% THEOREM: MSS
\begin{theorem}[Multistationarity when steady-state equations admit a triangular form] \label{prop:triangular}
Let $G$ be a reaction network with $s$ species.  
Suppose that the steady-state equations of $G$ admit 
a simplified system $\augmentH$,
a triangular form 
via functions $T_1, T_2, \dots, T_s$,
 and 
a variety 	 ${\mathcal W} \subsetneq   \mathbb{C}^{d} \times \mathbb{C}^{\bar m}$
 (as in Definition~\ref{def:triangular}).
 Fix a total-constant vector ${ c}^* \in \mathbb{R}^{d}_{>0}$, 
 effective parameters $a^* \in \mathbb{R}^{\bar m}_{>0}$,  and $x^* \in \mathbb{R}^s_{>0}$.  Assume that:
\begin{enumerate}
	\item $T_1(c^*,a^*,x_2^*, \dots, x_s^*)=T_2(c^*,a^*,x_3^*, \dots, x_s^*)=\dots=T_s(c^*,a^*,x_s^*)=0$ %$\augmentH(x^*)=0$ 
	(i.e., $x^*$ is a positive steady state of the system defined by $G$, ${ c}^*$, and any $\kappa^*$ for which $\bar a(\kappa^*)=a^*$),
	\item $x_s^*$ is a multiplicity-2 solution of $ T_s(c^*, a^*,x_s) ~=~ 0$,
	\item there exists an index $i \in \{1, 2, \dots, \bar{m} \}$ such that $(a_i^*, x_s^*)$ is a regular point of 
		\[ T_s(c^*; a_1^*,\ldots, a_{i-1}^*, a_{i}, 
	a_{i+1}^*,  \ldots, a_{\bar m}^*, x_s) ~=~ 0~,\]
	\item $(c^* , a^*) \notin {\mathcal W}$.
\end{enumerate}
Then $G$ is multistationary.  Moreover, 
a witness to multistationarity is guaranteed as follows:
there exists $\delta>0$ such that for all $\delta' \in (0,\delta)$, for either 
$a_i^{**}=a_i^*-\delta'$ or $a_i^{**}=a_i^*+\delta'$, the mass-action system given by $G$, 
$c^*$, and any 
$\kappa^*$ for which 
$\bar a(\kappa^*)=(a_1^*,\ldots, a_{i-1}^*, a_{i}^{**}, a_{i+1}^*,  \ldots, a_{\bar{m}}^*)$
 has at least two positive steady states.
\end{theorem}
\begin{proof}
Straightforward from Lemma~\ref{lem:4} and Definition~\ref{def:triangular}.
\end{proof}

\begin{remark}\label{rmk:degeneratess}
Theorem~\ref{prop:triangular} suggests a procedure for finding a witness to multistationarity for networks that admit 
a triangular form.  Namely, find a degenerate positive steady state $x^*$ and associated values for $c^*$ and $a^*$, 
and then perturb some coordinate of $a^*$ by a small amount.  See Example~\ref{ex:ccc}.
However, in practice this perturbation method could fail if 
we can obtain only approximations for $x^*$, $c^*$, and $a^*$, and if the multistationary region in the space of $a^*$'s 
 is too small for such approximations to find a witness.  
 Hence, our approach is most promising when $x^*$, $c^*$, and $a^*$ 
 can be exactly chosen with rational-number coordinates (as in Example~\ref{ex:ccc}).
 Indeed, for linearly binomial networks, this is easier as it is enough to find a solution 
 $x^*$ of $B(x^*)=0$, and the critical function $B$ has the nice properties detailed in Lemma~\ref{lm:ecfstructure}.
 Nevertheless, for this case, we gave other methods for obtaining witnesses to multistationarity.
\end{remark}

\subsection{Sufficient conditions for a triangular form}

Theorem~\ref{prop:triangular} gave an approach to multistationarity for networks
whose steady-state equations admit a triangular form.  In turn, Corollary~\ref{cor:tri} below guarantees that a network 
admits such a triangular form as long as the steady-state equations form a ``general zero-dimensional system''. % (Definition \ref{def:gzd}).

\begin{definition}\label{def:gzd}
A set of $s$ polynomials \[H~=~\{h_1,h_2, \ldots, h_s\}
	~\subseteq~ {\mathbb C}[a_1, a_2, \ldots a_n, x_1, x_2, \dots, x_s] = {\mathbb C}[a,x]\]
forms a {\em general zero-dimensional system} if there exists a variety ${\mathcal W}\subsetneq {\mathbb C}^n$ such that for any 
$a^*=(a^*_1, a^*_2, \ldots, a^*_n)\in {\mathbb C}^n\backslash {\mathcal W}$, 
the system $h_1|_{a=a^*}=h_2|_{a=a^*}=\ldots=h_s|_{a=a^*}=0$  satisfies:
% the hypotheses below:
%
\begin{itemize}
\item[(A1)]the number of complex solutions is finite and nonzero;

\item[(A2)]for any distinct complex solutions $x^*=(x^*_1, \ldots, x^*_s)$ and ${y}^*=(y^*_1, \ldots, y^*_s)$,
$x^*_s\neq y^*_s$;

\item[(A3)]the ideal ${\mathcal I}(\{ h_1|_{a=a^*},h_2|_{a=a^*}, \dots, h_s|_{a=a^*} \})$ is radical. 
\end{itemize}
\end{definition}

\begin{remark}
%Under the hypothesis (A1), 
When (A1) holds, 
the assumption (A3) implies that each 
 complex solution $x^*=(x^*_1, x^*_2, \ldots, x^*_s)$ has multiplicity $1$~\cite[Page 150, Corollary 2.6]{CLOb}.
\end{remark}
\begin{remark}\label{rmk:verifygzd}
It is computationally expensive to verify whether a given system is
general  zero-dimensional. 
In practice, one can take a heuristic approach: 
pick random values of 
$a_1, a_2, \ldots, a_n$, solve approximately for the resulting complex solutions, and then check whether conditions (A1)--(A3) in Definition \ref{def:gzd} are satisfied.
 %all complex solutions of the system. 
% If we do it several rounds and the specialized  system
%satisfy the assumptions (A1--A3) in Definition \ref{def:gzd}, 
If this is the case, then we assume that the given system is general zero-dimensional and we try to look for the triangular form described in Theorem \ref{shape}. 
\end{remark}

\begin{remark}[Relation to Shape Lemma]\label{rmk:shape}
Theorem \ref{shape} can be viewed as a more general version of the Shape Lemma. 
The original Shape Lemma \cite{BMMT1994} pertains to a zero-dimensional ideal
arising from a polynomial system without parameters.
Later, a version for systems involving parameters was given by geometric resolutions \cite{GHMMP1998, GLS2000}. 
The main difference between Theorem~\ref{shape} and the results in \cite{GHMMP1998, GLS2000} is that
in Theorem \ref{shape}, a triangular system representing the solution set of a general zero-dimensional ideal is selected from a Gr\"obner basis, whereas in  \cite{GHMMP1998, GLS2000}, a triangular system is computed by an interpolation idea.
\end{remark}

\begin{theorem}%[A sufficient condition for triangular form]
\label{shape}
Suppose that $h_1, h_2, \ldots, h_s\in {\mathbb C}[a,x]$ form a general zero-dimensional system. 
Let $\mathcal{G}$ be a
Gr\"obner basis of the following ideal 
with respect to the lexicographic order 
$a_n<\cdots<a_2<a_1<x_{s}<\cdots <x_2<x_{1}$:
%${\mathrm Ideal}\{F_1, \ldots, F_s\}$ 
\[
	{\mathcal I}(\{h_1,h_2, \ldots, h_s\} )~\subseteq ~{\mathbb C}[a,x]~.
	%\langle F_1, \ldots, F_s \rangle~.
	\] 
Then there exist $g_1, g_2, \ldots, g_s\in \mathcal{G}$ such that:
%$\langle G_1, G_{2}, \ldots, G_{s}\rangle$
\begin{itemize}
\item[(1)]  $g_1, g_2, \ldots, g_s$ have the following triangular form: 
%$G_1\in {\mathbb Q}[K_1, \ldots, K_s,  y_{1}]$ is monic in $y_{1}$ and 
\begin{align*}
         g_{s}~&=~Q_{s, N}~x_s^N +Q_{s, N-1}~x_s^{N-1} + \ldots +Q_{s, 1}~x_s + Q_{s, 0}~, \\
	 g_{s-1}~ &= ~ Q_{s-1}~ x_{s-1} + R_{s-1}~,\\
	               &~~\vdots\\
	  g_{1}~ &= ~ Q_{1}~ x_{1} + R_{1}~,
\end{align*}
where $N>0$ and, for all 
$i \in \{0,1,\dots, N\}$ and 
$j \in \{1,2,\dots, s-1\}$, we have: 
%\begin{align*}
\[
Q_{s, i} \in {\mathbb C}[a]~, \quad  %, \\
Q_{j} \in {\mathbb C}[a]~, \quad %, \\
%{\rm and } \quad
R_{j}\in {\mathbb C}[a, x_{j+1}, \ldots, x_{s}]~.
\]
%\end{align*}
%and $\deg_{y_{1}}(R_{i})<\deg_{y_{1}}(G_1)$
\item[(2)]For any $a^*\in {\mathbb C}^n\backslash 
	V(Q_{s, N}\,  Q_1\,  Q_2 \cdots Q_{s-1})$, 
%	V(Q_{s, N}\Pi_{i=2}^N Q_i)$, 
the set $\{g_1|_{a=a^*}, g_2|_{a=a^*}, \ldots g_s|_{a=a^*} \}$ is a Gr\"obner basis of the
following ideal
with respect to the lexicographic order $x_{s}<\cdots<x_2<x_{1}$:
	\[
	{\mathcal I}(\{ h_1|_{a=a^*},h_2|_{a=a^*}, \dots, h_s|_{a=a^*} \})
~\subseteq~ {\mathbb C}[x_1, x_2, \ldots,  x_s]~.
\]
% ideal generated by $F|_{b=b^*}$ in ${\mathbb C}[x_1, x_2, \ldots,  x_s]$ .\;
\end{itemize}
\end{theorem}

The proof of Theorem \ref{shape} is in Appendix~\ref{sec:pf-shape-thm}.

\begin{corollary} \label{cor:tri}
Let $G$ be a reaction network with $s$ species and a simplified system $\augmentH(x)$. If 
	$h_{c, a,1},  h_{c, a,2}, \dots,  h_{c, a,s}$
	form a general zero-dimensional system, then the steady-state equations of~$G$ admit a triangular form.
\end{corollary}

\begin{proof}
Straightforward from Definition~\ref{def:triangular}, Theorem~\ref{shape}, 
and the fact that if $(c^*,a^*)$ is not in $W$, then the same is true for any sufficiently small perturbation of  $(c^*,a^*)$.
\end{proof}

We end this section by 
showing, through the Calvin cycle example, how to 
use Theorem~\ref{prop:triangular} to find a witness to multistationarity.

\begin{example}[Calvin cycle, continued]\label{ex:ccc}
%Although Theorem \ref{thm:c-general} does not apply, we can 
%use Theorem \ref{prop:triangular} to find a witness to multistationarity.  
%We first simplify the steady-state equations $\F(x)$ as $\augmentH(x)$. 
We return to the network in Example~\ref{cc}, which is known to be multistationarity \cite{GAKKSN2011}.  
% Grimbs {\em et al.}
Here we find a witness to multistationarity.

Consider the following upper-triangular matrix:
{\footnotesize
\[M(\kappa)=
\left(\begin{array}{ccccccccccc}
1&  1& 0& 0  & 0&0 & 0& 0&0& 0&0\\
 0&   \frac{1}{\kappa_1}& 0& 0  & 0& 0 & 0& 0& 0& 0 & 0\\
 0&  0& 1& 1  & 0&0 & 0& 0&0& 0&0\\
 0 &  0 & 0     &  \frac{1}{\kappa_3} & 0 & 0&0& 0&0&0 &0\\
0&  0& 0& 0  & 1&1 & 0& 0&0& 0&0\\
  0&  0& 0& 0  & 0&\frac{1}{\kappa_5} & 0& 0&0& 0&0\\
0&  0& 0& 0  & 0&0 & 1& 5&0& 0&1\\
0&  0& 0& 0  & 0&0 & 0& \frac{1}{\kappa_7}&0& 0&0\\
 0&  0& 0& 0  & 0&0 & 0& 0&1& 1&0\\
0&  0& 0& 0  & 0&0 & 0& 0&0& \frac{1}{\kappa_{9}}&0\\
  0&  0& 0& 0  & 0&0 & 0& 0&0& 0&\frac{1}{\kappa_{11}}
  \end{array}
\right).
\]
}
Note that ${\rm det} M(\kappa)>0$ for $\kappa\in \mathbb{R}^{10}_{>0}$. 
Let the effective parameters be: 
\begin{align}\label{tab:excc}
\bar a_1 =\kappa_2, \;\;  &\bar a_2=\kappa_4,  \;\;\bar a_3=\kappa_6, \;\; \bar a_4=\kappa_{8}, 
 \;\; \bar a_5=\kappa_{10},  \;\; \bar a_6=\kappa_{12},   \\ \notag
\bar a_7=\frac{\kappa_2}{\kappa_1}, \;\; & \bar a_8=\frac{\kappa_4}{\kappa_3}, \;\;  \bar a_9=
\frac{\kappa_6}{\kappa_5}, \;\;\bar a_{10}=\frac{\kappa_{8}}{\kappa_{7}}, \;\; \bar a_{11}=
\frac{\kappa_{10}}{\kappa_{9}}, \;\; \bar a_{12}=\frac{\kappa_{12}}{\kappa_{11}}. 
\end{align}
From the above effective parameters \eqref{tab:excc} and equations \eqref{eq:linearchange}--\eqref{consys-h}, the resulting system $\augmentH(x)$ is:
\begin{center}
\begin{tabular}{lll}
$h_{c, a,1}= -a_1x_3+a_{5}x_{15}$,  & &$h_{c, a,2}= x_2 + x_3 - c_1$, \\
$h_{c, a,3}=  x_1x_2-a_7x_3$,  & &$h_{c, a,4}= 2a_1x_3 -a_2x_6$, \\
$h_{c, a,5} = x_5 + x_6 - c_2$, &  &$h_{c, a,6}= x_4x_5 - a_8x_6$, \\
$h_{c, a,7}= a_2x_6-a_3x_9$, & &$h_{c, a,8} = x_8 + x_9 - c_3 $, \\
$h_{c, a,9} = x_7x_8 - a_9x_9$, & &$h_{c, a,10} =a_3x_9 -5a_4x_{12}-a_{6}x_{17}$, \\
$h_{c, a,11} =x_{11}+x_{12}-c_4 $, & &$h_{c, a,12} =x_{10}^5x_{11}-a_{10}x_{12}$,\\
$h_{c, a,13} =3a_{4}x_{12}-a_{5}x_{15}$, & &$h_{c, a,14}=x_{14} + x_{15} - c_5$,\\
$h_{c, a,15} =x_{13}x_{14}-a_{11}x_{15}$, & &$h_{c, a,16} =x_{16} + x_{17} - c_6$, \\
$h_{c, a,17} =x_{10}x_{16}-a_{12}x_{17}$.
\end{tabular}
\end{center}

With an eye toward applying 
Theorem \ref{prop:triangular}, we find a triangular form for our steady-state equations, as follows. 
Following the heuristic proposed in Remark \ref{rmk:verifygzd}, 
we verified for several values of the $a_i$'s and $c_i$'s that 
 (A1), (A2), and (A3) in Definition \ref{def:gzd} are satisfied.
 Hence, we assume that the $h_{c,a,i}$'s
 form a general zero-dimensional system.
If so, Corollary~\ref{cor:tri} would imply that 
the steady-state equations admit a triangular form.  More precisely,
we would obtain 
such a triangular set as a subset of a Gr\"obner basis 
 (by Theorem \ref{shape}).  
 
Accordingly, we compute a Gr\"obner
basis of  the generated ideal $\langle h_{c,a,i} \rangle$ 
with respect to the lexicographic order $c<a<x_{17}<\ldots<x_{1}$, via {\tt Maple} \cite{maple}, which consists of $86$ polynomials. 
Following the proof of Theorem \ref{shape} (see Appendix~\ref{sec:pf-shape-thm}), 
we select for each $i=1,2, \ldots, 17$, 
some $g_i$ among these $86$ polynomials for which the leading monomial of $g_i$, where $g_i$ viewed in $\mathbb{Q}(a_i,c_i)[x]$, 
% with respect to $x$ 
has the form $x_i^{N_i}$ for some $N_i \geq 0$.
This yields the following $17$ polynomials $g_i$, which form a triangular set:
\begin{align*}
	g_{17}~&=~a_6(a_{10}-a_{12}^5)x_{17}^6+
	(a_4a_{12}^5c_4-5a_6a_{10}c_6)x_{17}^5
	+
	 10{a_{6}a_{10}c_6^2}x_{17}^4 \\
	& \quad \quad 
	-
	10{a_6a_{10}c_6^3}x_{17}^3 
	+
	5{a_{6}a_{10}c_6^4}x_{17}^2-
	{a_{6}a_{10}c_6^5}{x_{17}}, \\
g_{16}~&=~{x_{16}}+x_{17}-{c_6}, \\
g_{15}~&=~{a_{5}}{x_{15}}-3a_{6}x_{17}, \\
g_{14}~&=~{x_{14}}+x_{15}-{c_5}, \\
	&~ \vdots \\
g_{3}~&=~ {a_1}{x_3}-3a_{6}x_{17}, \\
 g_{2}~&=~{x_2} + x_3 - {c_1}, \\
 g_{1}~&=~Q_{1}({a, c})~ {x_{1}} + R_{1}({a, c}, x_{17})~,
\end{align*}
 where $Q_{1}(a,  c)$ and  $R_{1}({a,  c}, x_{17})$ are polynomials. 
 So, by 
 Theorem \ref{shape} (2), 
the steady-state equations admit the following triangular form:
\[T_{17} = g_{17}(a, c, x_{17}), \;\; T_{16}=x_{16} - (c_6-x_{17}),\;\; T_{15}=x_{15}-3\frac{a_{6}}{a_{5}}x_{17}, \;\;T_{14}=x_{14} - (c_5-x_{15} )\]
\[\ldots\]
\[T_{3}=x_{3} - 3\frac{a_{6}}{a_1}x_{17}, \;\;T_{2}={x_2} - ({c_1}-x_3 ), \;\;T_{1}=x_1 - \frac{R_{1}({a, c}, x_{17})}{Q_{1}({a, c})}~. \]
This triangular form is valid
as long as the leading coefficients of $g_i$'s do not vanish.  In other words,
the variety $\mathcal W$,
as in Definition~\ref{def:triangular}, can be defined by the vanishing set of those coefficients.

Next, we aim to find a degenerate positive steady state $x^*$ and corresponding
 parameters $(a^*, c^*)$. 
 Let $\hat a=a_6$ and $\hat x = x$. 
The idea is to 
compute the critical function $C(\hat a, \hat x)$ and then 
find a positive point where the function vanishes. 
%We first derive a critical function. 
By solving the equations $h_{c, a,1}=h_{c, a,3}=h_{c, a,4}=h_{c, a,6}=
h_{c, a,7}=h_{c, a,8}=h_{c, a,9}=h_{c, a,10}=h_{c, a,12}=
h_{c, a,13}=h_{c, a,15}=h_{c, a,17}=0$ in the unknowns $a_1, a_2, a_3, a_4, a_5, a_7, a_8, a_9, a_{10}, a_{11}, a_{12}$, we obtain:
\begin{center}
\begin{tabular}{llllll}
$a_1=\frac{3a_6x_{17}}{x_3}$, &
$a_2=\frac{6a_6x_{17}}{x_6}$, & 
$a_3=\frac{6a_6x_{17}}{x_9}$, &
$a_4=\frac{a_6x_{17}}{x_{12}}$, & 
$a_5=\frac{3a_6x_{17}}{x_{15}}$, &\\
$a_7=\frac{x_1x_2}{x_3}$, &
$a_8=\frac{x_4x_5}{x_6}$, & 
$a_{9}=\frac{x_7x_8}{x_9}$, &
$a_{10}=\frac{x^5_{10}x_{11}}{x_{12}}$, &
$a_{11}=\frac{x_{13}x_{14}}{x_{15}}$, &
$a_{12}=\frac{x_{10}x_{16}}{x_{17}}$. 
\end{tabular}
\end{center}
Thus, we obtain the following steady-state parametrization (outside $\mathcal W$) 
	$\phi: \mathbb{R}^{18}_{>0} \rightarrow   \mathbb{R}_{>0}^{12}\times  \mathbb{R}_{>0}^{17}$, 
where 
$\phi(\hat a, \hat x)=\phi(a_6; x) $ is defined as 
\[\left(\frac{3a_6x_{17}}{x_3}, 
\frac{6a_6x_{17}}{x_6}, 
\frac{6a_6x_{17}}{x_9},
\frac{a_6x_{17}}{x_{12}},
\frac{3a_6x_{17}}{x_{15}},
a_6,
\frac{x_1x_2}{x_3},
\frac{x_4x_5}{x_6},  
\frac{x_7x_8}{x_9}, 
\frac{x^5_{10}x_{11}}{x_{12}}, 
\frac{x_{13}x_{14}}{x_{15}}, 
\frac{x_{10}x_{16}}{x_{17}};
 x\right)~.
\]
%By Definition \ref{def:pp}, we see that the network $G$ admits the positive parametrization $\phi$.  

The resulting critical function is:
\[C(\hat a, \hat x)~=~C(a_6,x) 
		~=~	
-324a_6^5x_2x_5x_8x_{10}^5x_{14}x_{17}^4\left(4x_{11}x_{16}+5x_{11}x_{17}-x_{12}x_{16}\right)~. \]
It is straightforward to find a positive rational vector $(a_6^*,x^*) $ such that $ C(a_6^*, x^*)$ vanishes. For instance, one can simply choose 
\begin{align*}
a^*_6 ~=~ 1 \quad {\rm and } \quad 
x^* ~=~ \left(
1, \;1, \;1,\;1,\;1, \;1, \;
1,\;1,\; 1,\;1,\;1, \; 9,\; 1, \;
1,\; 1,\; 1,\;1
\right).
\end{align*}
Plugging $u^*$ into $\phi(u)$, we compute:
 \begin{align*}
a^*~=~ \left(3, \;
 6 , \; 
 6, \; 
 \frac{1}{9},\;
 3, \;
 1,\;
 1, \;
 1, \;
 1, \; 
 \frac{1}{9}, \;
 1, \;
 1\right).
 \end{align*}
 Also, by plugging $x^*$ into the conservation laws, we obtain:
 \begin{align*}
c^*~=~ \left(2, \;
 2, \;
  2, \; 
 10, \;
 2, \;
 2 \right).
\end{align*}
Then $x^*$ is a degenerate positive steady state for $(a^*, c^*)$. 
In other words, we have verified the first hypothesis of Theorem~\ref{prop:triangular}.
It is straightforward to check the validity of the remaining three hypotheses of this theorem. 
 Thus, by Theorem \ref{prop:triangular}, there exists a small positive number $\delta$ such that for any $0 < \delta' < \delta$
\[(a^*_1, \ldots, a^*_3, a^*_4+\delta', a^*_5, \ldots, a^*_{12}, c^*)
\quad \quad 
{\rm 
or }
\quad \quad 
	(a^*_1, \ldots, a^*_3, a^*_4-\delta', a^*_5, \ldots, a^*_{12}, c^*)\]
generates multistationarity. Indeed, one can check that the system 
\[\augmentH(a^{**}, c^*, x)=0, \;\;\;\text{where}\;a^{**}=\left( a^*_1, \ldots, a^*_3, a^*_4-\frac{1}{1000}, a^*_5, \ldots, a^*_{12} \right)~,\] 
has two distinct positive steady states, which are approximately equal to:
 \begin{align*}
x^{(1)}\approx &(0.96, 1.01,  0.98, 0.96, 1.01, 0.98,  0.96, 1.01, 0.98,
  0.96, 1.14, 8.85, 0.96, 1.01,  0.98, 1.01,  0.98), and\\
x^{(2)}\approx&(1.02,  0.98, 1.01, 1.02,  0.98, 1.01, 1.02, 0.98, 1.01, 
 1.02,  0.87,  9.12,  1.02, 0.98, 1.01,  0.98,  1.01).
 \end{align*}
Finally, any $\kappa^* \in \mathbb{R}^{12}_{>0}$ for which $\bar a(\kappa^*) = a^{**}$,
as in~\eqref{tab:excc}, yields a witness to multistationarity. One such $\kappa^*$ is:
 \begin{align*}
 \kappa^*~=~\left(
 3, \;
  3, \; 
 \frac{5999}{1000}, \; 
 \frac{5999}{1000},\;
 6, \;
 6, \;
 1, \;
 \frac{1}{9}, \;
3, \; 
 3, \;
  1, \;
 1
 \right)~. 
 \end{align*}
% \begin{align*}
%(x_1 &= 0.96, x_2 = 1.01, x_3 = 0.98, x_4 = 0.96, x_5 = 1.01, x_6 = 0.98, x_7 = 0.96, x_8 = 1.01, x_9 = 0.98,\\
  %x_{10}& = 0.96, x_{11} = 1.14,  x_{12} = 8.85, x_{13} = 0.96, x_{14} = 1.01, x_{15} = 0.98, x_{16} = 1.01, x_{17} = 0.98),
  %\end{align*}
 % \begin{align*}
%(x_1 &= 1.02, x_2 = 0.98, x_3 = 1.01, x_4 = 1.02, x_5 = 0.98, x_6 = 1.01, x_7 = 1.02, x_8 = 0.98, x_9 = 1.01, \\
%x_{10} &= 1.02, x_{11} = 0.87, x_{12} = 9.12, x_{13} = 1.02, x_{14} = 0.98, x_{15} = 1.01, x_{16} = 0.98, x_{17} = 1.01).
%\end{align*}
\end{example}

\section{Discussion}\label{sec:discussion}
Steady-state parametrizations have been shown in recent years to 
be very useful for analyzing chemical reaction networks
and more specifically their capacity for multistationarity. 
Our first main results are in this vein, discerning multistationarity and finding witnesses to 
multistationarity for networks with steady-state parametrizations
(Theorem~\ref{thm:c-general}), including linearly binomial networks (Theorem~\ref{thm:NP}). Furthermore, we characterize a class of MESSI networks that are linearly binomial,  which provides an interesting range of application for our results. 
In a complementary direction, 
we also showed how to obtain witnesses to multistationarity when a network's steady-state equations 
admit a triangular form and a degenerate steady state exists (Theorem \ref{prop:triangular}).  

Hence, the main contribution of our work is a suite of new tools for tackling the important but difficult 
problem of assessing and analyzing multistationarity.  Moreover, our results can decide multistationarity for networks from biology 
that previously could not be handled systematically.

Finally, our work used novel approaches that we expect to be useful in the future.
For instance, we used steady-state parametrizations in which the rate constants depend on the species variables,
 thereby simplifying the subsequent analyses.
We also analyzed reaction networks based on results on specialization of Gr\"obner bases, adapting a general technique used in several applications.   
We expect similar algebraic techniques to allow us in the future
to go beyond multistationarity to study topics such as 
stability of steady states and Hopf bifurcations.

%*******************************************************************
%Acknowledgements
%*******************************************************************
\subsection*{Acknowledgements}
The authors thank Frank Sottile for helpful discussions, Alan Rendall for pointing us to the Calvin Cycle model, and
Carsten Conradi for helpful discussions on the ERK network. 
The authors also thank three conscientious referees whose comments helped improve our work.
AD and MPM were partially supported by UBACYT 20020170100048BA, 
CONICET PIP 11220150100473 and 11220150100483, and ANPCyT PICT 2016-0398, Argentina.
AS partially supported by the NSF (DMS-1513364 and DMS-1752672)
and the Simons Foundation (\#521874).
XT was partially supported by the NSF (DMS-1752672).

%\bibliographystyle{plain}
%\bibliography{mss.bib}

\appendix

\section{Proof of Theorem~\ref{thm:binom}} \label{sec:pf-thm-binom}
% APPENDIX
We now prove Theorem~\ref{thm:binom}. We also illustrate the proof in Example~\ref{ex:toy}.
We assume the reader is familiar with the notion of the Laplacian $\Lap(G)$ of a 
digraph $G$ and its main properties. One important observation is that mass-action 
kinetics associated with a digraph $G$ with vertices labeled by variables $x_1, \dots, x_s$ 
equals $\dot x =\Lap(G) x$. 
Another important observation is that when $G$ is strongly connected, the kernel of $\Lap(G)$ 
has dimension one and there is a known generator $\rho(G)$ with positive entries described as follows.
Recall that an {\it $i$-tree} $T$ of  a digraph is a 
spanning tree where the $i$-th vertex is its unique sink (equivalently, the $i$-th is the only vertex
of the tree with no edges leaving from it) and  we call $k^{T}$ 
the product of the labels of all the edges of $T$. Then, the $i$-th coordinate of $\rho(G)$ equals
\begin{equation}\label{eq:rhoi}
\rho_i(G)=\underset{T\; an \; i-tree}{\sum}k^{T}.
\end{equation}
We refer the reader to~\cite{MiGu13,Tutte} for a detailed account.

% PROOF
\begin{proof}[Proof of Theorem~\ref{thm:binom}]
Recall that in a MESSI network there are two types of species: intermediate and core. 
Our proof proceeds by performing (invertible) linear operations on the steady-state equations, which in the end yield (equivalent) binomial equations.

We begin by operating on the steady-state equations of intermediate species.
%We start by proving that the steady state equations corresponding to intermediate species can be replaced by binomial equations via linear operations. 
Given a core complex $y$, we consider the following set of 
intermediate complexes: 
$I_{y}=\{y' \text{ intermediate} : y\uri y'\}$.
Following the reasoning in \cite{fw13}, we build a labeled directed graph denoted by
$G_{y}$, with node set $I_{y}\cup \{y\}$, and labeled directed edges as in $G$, except that any 
reaction of the form $y'\overset{\kappa}{\rightarrow} y''$, where $y'\in I_y$ and $y''$ is any  
core complex, is replaced by $y'\overset{\kappa}{\rightarrow} y$, with the same rate constant 
(if there are several core complexes to which $y'$ reacts, the edges are collapsed and the 
label equals the sum of the labels of the corresponding collapsed edges).
Note that, as all
intermediate complexes in MESSI networks react via intermediates to some core complex, the graph $G_y$ is
strongly connected.

Number the species in $I_y$ and denote by $x_1, \dots, x_{n_y}$ the corresponding concentrations. 
Then, the mass-action ODEs corresponding to them in the given network,  
are given by $f_\ell=\dot{x}_\ell=(\Lap(G_{y}))_\ell {\mathbf x}^{\top}$, 
where  $(\Lap(G_y))_\ell$ is the $\ell$th row of the Laplacian $\Lap(G_{y})$ of $G_{y}$,
${\bf x}=(x_1, \dots, x_{n_y},m(y))$, and $m(y)$ is  the monomial associated with the complex $y$.
Call $\rho_\ell=\rho(G_{y})_\ell$ for $1\leq \ell \leq n_y+1$. 
From the Matrix-Tree Theorem we know that, up to sign, the determinant of the first 
$n_y\times n_y$ principal minor of 
$\Lap(G_{y})$ equals $\rho_{n_y+1}\neq 0$. Call $A$ the $(n_y+1)\times(n_y+1)$ block matrix 
\[A=\left(\begin{array}{c|c}A_1 & 0 \\ \hline 0 & 1\end{array}\right),\]
where $A_1$ is the $n_y\times n_y$ is such that $A\cdot \Lap(G_{y})$ has the form:  
\[A\cdot \Lap(G_{y})=\left(\begin{array}{ccc|c}1 & & 0 &- \alpha_1\\ & \ddots & & \vdots\\ 
0 & &1& - \alpha_{n_y}\\ \hline * & \dots & * & * \end{array}\right).\]
Such a matrix $A$ exists since the $n_y \times n_y$ first principal minor
of $\Lap(G_{y})$ is invertible.
Observe that $\ker(\Lap(G_{y}))$ is generated by 
$(\rho_1,\rho_2,\dots,\rho_{n_y},\rho_{n_y+1})$, so $\rho_\ell- \alpha_\ell\rho_{n_y+1}=0$,
for $1\leq \ell\leq n_y$, and then $\alpha_\ell=\frac{\rho_\ell}{\rho_{n_y+1}}$ 
(and $\alpha_\ell\neq 0$). 
By multiplying $\Lap(G_{y})$ on the left by $A$, which is equivalent to operating  
linearly on the original equations $f_\ell=0$, we deduce the binomial equations 
$x_\ell-\alpha_\ell \, m(y)=0$, for the intermediate species $x_\ell$ at steady state, 
for $1\leq \ell \leq n_y$. Note that, as $\rho_{n_y+1}\neq 0$ for all $\kappa\in\mathbb{R}^m_{>0}$, 
these operations are well defined for all $\kappa\in\mathbb{R}^m_{>0}$.

We can afterwards substitute the steady-state value of  the intermediate species $x_\ell \in I_y$ 
into the original steady-state equations, for core species, of $G$. We moreover show that this 
substitution can be achieved via linear operations. 
Indeed, for a core species $X_k$, write the corresponding ODE as:
\[f_k~=~\dot{x}_k~=~p_k+\sum_{\ell=1}^{n_y} \kappa_\ell x_{\ell}~,\]
where $p_k$ is a polynomial that does not depend on $x_1,\dots, x_{n_y}$ and 
$\kappa_\ell\ge 0$ is positive precisely when $X_{\ell}$ reacts with rate constant $\kappa_\ell$ 
to a core complex that involves $X_k$.
We now subtract a linear combination of the intermediate-species binomials $(x_\ell-\alpha_\ell \, m(y))$: 
\[f_k-\sum_{\ell=1}^{n_y} \kappa_\ell (x_{\ell}-\alpha_{\ell} m(y))
~=~
p_k+ \sum_{\ell=1}^{n_y} (\kappa_\ell \alpha_{\ell}) m(y)~,
\]
and so we replace $f_k=0$ by $p_k+ \sum_{\ell=1}^{n_y} (\kappa_\ell \alpha_{\ell}) m(y)=0$, where 
all the intermediates in $I_y$ have been eliminated by performing linear operations 
on the original steady-state equations (for core species) and the new binomials (for the intermediates).

A key observation 
pertaining to how we obtained binomial equations for all intermediate species is that, by 
condition~$(\cond)$, the set of intermediate complexes can be written as the disjoint union of 
sets $I_y$ for a certain (finite) number of core complexes $y$. By the natural bijection between 
intermediate complexes and intermediate species, we can then obtain the corresponding binomials 
by operating linearly on the original equations. We also, as described above, eliminated all the intermediate 
species from the core-species equations by linear operations. 
We will denote this procedure as follows. First assume, without loss of generality, that the 
intermediate species are the last $s-n$ species. Then we can assert that there exists an 
invertible matrix $M_1=M_1(\kappa)\in \Q(\kappa)^{s\times s}$  which is well defined for all 
$\kappa\in\mathbb{R}^m_{>0}$ such that 
$M_1(f_1,\dots,f_s)^{\top}=(\tilde{f}_1,\dots,\tilde{f}_n,\tilde{h}_{n+1},\dots,\tilde{h}_s)^{\top}$,
where $\tilde{f}_1,\dots,\tilde{f}_n$
% have the form $p_k+ \sum_{\ell=1}^{n_y} (\kappa_\ell \alpha_{\ell}) m(y)$ - but repeatedly applied... - and therefore 
do not depend on the intermediate-species concentrations 
$x_{n+1},\dots, x_{s}$, 
and $\tilde{h}_{n+\ell}$ is the binomial for the intermediate species 
$x_{n+\ell}$ ($1\leq \ell \leq s-n$) and its form is 
$x_{n+\ell}-\alpha_{n+\ell} \, m(y)=0$.

Before we continue operating on the core-species equations (the $\tilde{f}_i$'s are not binomials), we describe the map $\Katau=\Katau(\kappa)$ 
mentioned in \eqref{eq:katau}. For each $X_i+X_j\uri X_\ell+X_m$ in $G$, the 
reaction constant $\tau$ in $G_1$ which gives the label 
$X_i+X_j\overset{\tau}{\longrightarrow} X_\ell+X_m$ 
has the form
\begin{equation*}% \label{eq:tau}
\tau=\kappa+\overset{s-n}{\underset{k=1}{\sum}}\kappa_k\alpha_k,
\end{equation*}
where $\kappa\ge 0$ is positive when $X_i+X_j\overset{\kappa}{\longrightarrow} X_\ell+X_m$ in $G$ 
(and $\kappa=0$ otherwise), and $\kappa_k\ge 0$ is positive if there is a reaction from the 
intermediate species $X_{n+k}\overset{\kappa_k}{\longrightarrow} X_\ell+X_m$ and 
$X_i+X_j\uri X_{n+k}$ in $G$ (and $\kappa_k=0$ otherwise). 
As we pointed out in Remark~\ref{rem:G1}, the steady states of $G$ are in one-to-one 
correspondence with those of $G_1$ and, in fact, the polynomials $\tilde{f}_i$ can be read 
from the digraph $G_1$ (see Theorem~3.2 in~\cite{fw13}).

What we show now is that we can 
operate linearly on the core-species equations 
$\tilde{f}_i=0$ (for $1\leq i\leq n$) in order to 
obtain equivalent binomial equations.
To avoid unnecessary notation, we will assume in what follows that the partition of $\Sp$ 
is minimal.
Recall that a vertex in a directed graph has \emph{outdegree zero} if it is not the tail of any 
directed edge. 
Let us define subsets of indices based on the graph $G_E$:
\begin{align*}
  L_0=&\{\beta \ge 1 : \text{outdegree of }\Sp^{(\beta)}\text{ is }0\},\\
  L_k=&\{\beta \ge 1: \text{for any edge }  \Sp^{(\beta)}\to\Sp^{(\gamma)}\text{ in }G_E \text{ it holds that }
  \gamma \in L_t, \text{ with }  t<k \}\backslash \underset{t=0}{\overset{k-1}{\bigcup}} L_t, ~ k \ge1.
\end{align*}
As  $\Sp$ is finite and there are no directed cycles in $G_E$, there must exist a subset 
$\Sp^{(\beta)}$ with $1\le \beta\le m$ such that its outdegree in $G_E$ is zero. 
This means that $L_0\neq \emptyset$.

Consider $\alpha\in L_0$. 
By the assumption of minimality of the partition, 
there is a connected component of $G^\circ_2$, 
which we denote by $H_\alpha$, 
 with vertices 
the species in $\Sp^{(\alpha)}$.
Let $\widetilde{H}_\alpha$ be the corresponding underlying 
undirected graph.
As $\widetilde{H}_\alpha$ is a tree, consider $X_i$, a leaf of the tree (this is, a vertex of 
degree one). $X_i$ is only connected to one vertex $X_j$, so  $\tilde{f}_i$ is already a 
binomial of the form
\[\tilde{f}_i=\tau_{ji}x_jx_\ell-\tau_{ij}x_ix_h,\]
for some species $X_\ell\in \Sp^{(\beta)}, X_h\in \Sp^{(\gamma)}$, $\beta, \gamma$ in levels 
strictly greater than $0$. Moreover, $\tilde{f}_j$ is 
\[\tilde{f}_j=p_j+\tau_{ij}x_ix_h-\tau_{ji} x_jx_\ell=p_j-\tilde{f}_i,\]
with $p_j$ a polynomial that does not depend on $x_i$.
Then, $\tilde{f}_j+\tilde{f}_i=p_j$.
And we can replace $\tilde{f}_j$ with $p_j$.
As the associated digraph $G_E$ has no directed cycles, all the reactions are enzymatic. 
This means that the reactions that correspond to $\tau_{ij}$ and $\tau_{ji}$ in $G_1$ are 
$X_i+X_h\overset{\tau_{ij}}{\longrightarrow}X_j+X_h$ and  
$X_j+X_\ell\overset{\tau_{ji}}{\longrightarrow}X_i+X_\ell$, 
respectively, and none of these reactions affect either $\tilde{f}_h$ or $\tilde{f}_\ell$. 
Moreover, as $X_i\in \Sp^{(\alpha)}$ with $\alpha\in L_0$ and $\Sp^{(\alpha)}$ has outdegree 
zero in $G_E$, $x_i$  only appears in $\tilde{f}_i$ and $\tilde{f}_j$.
We have then eliminated by linear operations the variable $x_i$ from all the equations other 
than $\tilde{f}_i$. And this can be done for all the species whose vertices are leaves of 
$\widetilde{H}_\alpha$. 
We can then erase all the leaves from $\widetilde{H}_\alpha$ and, by an inductive argument we 
see that we can operate linearly, with integer coefficients, on 
$\tilde{f}_1,\dots,\tilde{f}_n$ to obtain binomials for the species in 
$\Sp^{(\alpha)}$. This argument holds for any $\alpha \in L_0$. As the species in the 
$\Sp^{(\alpha)}$'s for $\alpha\in L_0$ do not appear in any label of the $L_i$'s for $i\geq 1$ 
and all the reactions in all the $H_\alpha$'s with $\alpha\in L_0$ do not affect the equations 
of those species that appear on its labels, by an inductive argument we can 
complete the proof to obtain binomial equations by operating linearly on 
the equations $\tilde{f}_j=0$ for the species in $\Sp^{(\alpha)}$ for $\alpha \in L_i$, $i \geq 1$.

We have then proved that there exists a matrix with integer entries 
$\widetilde{M}_2\in\Q^{n\times n}$, and an invertible block matrix $M_2\in \Q^{s\times s}$ 
of the form
\[M_2=\left(\begin{array}{c|c} \widetilde{M}_2 & 0\\ \hline 0 & Id_{s-n} \end{array}\right),\]
such that 
$M_2~(\tilde{f}_1,\dots,\tilde{f}_n,\tilde{h}_{n+1},\dots,\tilde{h}_s)^{\top}=
(\tilde{h}_1,\dots,\tilde{h}_s)^{\top}$,
where $\tilde{h}_1,\dots, \tilde{h}_s$ are binomials.

We have so far that there are invertible matrices $M_1\in\Q(\kappa)^{s\times s}$ and 
$M_2\in\Q^{s\times s}$, with $M_1$ well defined for all $\kappa\in\R^m_{>0}$, such that 
$M_2M_1(f_1,\dots,f_s)^{\top}=(\tilde{h}_1,\dots,\tilde{h}_s)^{\top}$, with $\det(M_2M_1)\neq 0$. 
If  $f_{j_1},\dots,f_{j_{s-d}}$ is a basis of the $\Q(\kappa)$-linear subspace generated by 
$f_1,\dots,f_s$ there must exist a set of $s-d$ binomials  
$\{\bar{h}_{j_1},\dots,\bar{h}_{j_{s-d}}\}\subseteq\{\tilde{h}_1,\dots,\tilde{h}_s\}$ 
and an invertible matrix $M\in \Q(\kappa)^{(s-d)\times(s-d)}$, which is well defined for all 
$\kappa \in \mathbb{R}^m_{>0}$, such that 
$M(f_{j_1},\dots,f_{j_{s-d}})^{\top}=(\bar{h}_{j_1},\dots,\bar{h}_{j_{s-d}})^{\top}$, as we wanted to prove.
\end{proof}

\begin{example}\label{ex:toy}
Consider the following network:
\begin{center}
\begin{tabular}{c}
\ce{S_{0} + E
<=>[\kappa_1][\kappa_2]
ES_{0}
->[\kappa_3]
S_{1} + E
->[\kappa_4]
S_{2} + E
->[\kappa_5]
S_{3} + E}
\\
\ce{S_{3}
->[\kappa_6]
S_{2}
->[\kappa_7]
S_{1}
->[\kappa_8]
S_{0}
},\\
\end{tabular}
\end{center}
which has $s=6$ species:
\begin{center}
\begin{tabular}{llllll}
$X_{1}$=\ce{S_{0}}, & $X_2$=\ce{S_{1}},& $X_3$=\ce{S_{2}},&$X_{4}$=\ce{S_{3}},  & $X_5$=\ce{E}, & $X_{6}$=\ce{ES_{0}}.\\
\end{tabular}
\end{center}
There are $2$ conservation laws: 
\begin{align*}
 x_1+x_2+x_3+x_4+x_6 &=c_1\\
 x_5+x_6 &=c_2.
\end{align*}
The equations $f_{c,\kappa}(x)$ are

\begin{minipage}{0.47\textwidth}
 \noindent $f_{c,\kappa,1}=x_1+x_2+x_3+x_4+x_6-c_1$,\\
$f_{c,\kappa,3}=\kappa_4x_2x_5-\kappa_5x_3x_5+\kappa_6x_4-\kappa_7x_3$,\\
$f_{c,\kappa,5}=x_5+x_6-c_2$,
\end{minipage}
\begin{minipage}{0.47\textwidth}
 \noindent $f_{c,\kappa,2}=\kappa_3x_6-\kappa_4x_2x_5+\kappa_7x_3-\kappa_8x_2$,\\
$f_{c,\kappa,4}=\kappa_5x_3x_5-\kappa_6x_4$,\\
$f_{c,\kappa,6}=\kappa_1x_1x_5-(\kappa_2+\kappa_3)x_6$.
\end{minipage}

\medskip

The matrix $M_1(\kappa)$ is the product of two matrices: the first one multiplied by 
$(f_1,\dots,f_6)^{\top}$ gives binomials for the intermediate species equations; the second one 
eliminates the intermediate species from the core species equations.
{\footnotesize
\[M_1=M_1(\kappa)=\left(\begin{array}{ccc|c}
 &  &  & -\kappa_2\\
 &  &  & -\kappa_3\\
 & Id_5 &  & 0 \\
 &  &  & 0 \\
 &  &  & -\kappa_2-\kappa_3\\
 \hline
 0 & \dots & 0 & 1
\end{array}\right)\cdot\left(\begin{array}{ccc|c}
 & & & 0\\
 & & & \\
 & Id_5 & & \vdots \\
 & & & \\
 & & & 0\\
 \hline
 0 & \dots & 0 & -\dfrac{\vphantom{A^A}1}{\kappa_2+\kappa_3} 
\end{array}\right),\]
}
where $Id_5$ is the $5\times 5$ identity matrix. This leads to:
{\small
\begin{align}
\nonumber M_1(f_1,\dots,f_6)^{\top} =&(\kappa_8x_2-\frac{\kappa_1\kappa_3}{\kappa_2+\kappa_3}x_1x_5,~ 
-\kappa_4x_2x_5+\kappa_7x_3-\kappa_8x_2+\frac{\kappa_1\kappa_3}{\kappa_2+\kappa_3}x_1x_5,\\ \label{eq:tilde}
&f_3,~f_4,~0,~-\frac{\kappa_1}{\kappa_2+\kappa_3}x_1x_5+x_6)^{\top} =(\tilde{f}_1,\dots,\tilde{f}_5,\tilde{h}_6).
\end{align}
}
The corresponding digraph $G_2^\circ$ defined in \S\ref{ssec:associated_digraphs} is:
\begin{center}
\begin{tabular}{c}
\ce{S_{0}
<=>[\frac{\kappa_1}{\kappa_2+\kappa_3}x_5][\kappa_8]
S_{1}
<=>[\kappa_4x_5][\kappa_7]
S_{2}
<=>[\kappa_5x_5][\kappa_6]
S_{3}
},\\
\end{tabular}
\end{center}
and the underlying undirected graph is a tree with leaves \ce{S_{0}} and \ce{S_{3}}. 
Then $\tilde{f}_1$ and $\tilde{f}_4$ are already binomials, and  
$p_2=\tilde{f}_2-\tilde{f}_1=-\kappa_4x_2x_5+\kappa_7x_3$ (from~\eqref{eq:tilde}), and 
$p_3=\tilde{f}_3+\tilde{f}_4=(\kappa_4x_2x_5-\kappa_5x_3x_5+\kappa_6x_4-\kappa_7x_3)+(\kappa_5x_3x_5-\kappa_6x_4)=\kappa_4x_2x_5-\kappa_7x_3$, 
which are binomials. The matrix $M_2$ is
{\small
\[M_2=\left(\begin{array}{cccccc}
1 & 0 & 0 & 0 & 0 & 0\\
1 & 1 & 0 & 0 & 0 & 0\\
0 & 0 & 1 & 1 & 0 & 0 \\
0 & 0 & 0 & 1 & 0 & 0 \\
0 & 0 & 0 & 0 & 1 & 0\\
0 & 0 & 0 & 0 & 0 & 1
\end{array}\right),\]
}
and $M_2~(\tilde{f}_1,\dots,\tilde{f}_5,\tilde{h}_6)^{\top}=(\tilde{h}_1,\dots,\tilde{h}_6)^{\top}$.

In order to obtain $h_{c,\kappa,2}$, $h_{c,\kappa,3}$, $h_{c,\kappa,4}$, 
$h_{c,\kappa,6}$ we need to build $f_1$ and $f_5$ from 
$f_{c,\kappa,2}$, $f_{c,\kappa,3}$, $f_{c,\kappa,4}$, $f_{c,\kappa,6}$, 
multiply by the product $M_2M_1$, and then pick $4$ linearly independent binomials:
{\small
\[M'=\left(\begin{array}{cccccc}
1 & 0 & 0 & 0 & 0 & 0\\
0 & 1 & 0 & 0 & 0 & 0\\
0 & 0 & 0 & 1 & 0 & 0 \\
0 & 0 & 0 & 0 & 0 & 1
\end{array}\right)M_2M_1
\left(\begin{array}{rrrr}
-1 & -1 & -1 & -1 \\
1 & 0 & 0 & 0\\
0 & 1 & 0 & 0 \\
0 & 0 & 1 & 0  \\
0 & 0 & 0 & -1 \\
0 & 0 & 0 & 1 
\end{array}\right)=
\left(\begin{array}{rrrr}
-1 & -1 & -1 & -\frac{\kappa_3}{\kappa_2+\kappa_3} \\
0 & -1 & -1 & 0\\
0 & 0 & 1 & 0 \\
0 & 0 & 0 & -\frac{1}{\kappa_2+\kappa_3} 
\end{array}\right)
,\]}

\noindent
which leads to $M'(f_{c,\kappa,2},f_{c,\kappa,3},f_{c,\kappa,4},f_{c,\kappa,6})^{\top}=
(\kappa_8x_2-\frac{\kappa_1\kappa_3}{\kappa_2+\kappa_3}x_1x_5,-\kappa_4x_2x_5+\kappa_7x_3,
\kappa_5x_3x_5-\kappa_6x_4,-\frac{\kappa_1}{\kappa_2+\kappa_3}x_1x_5+x_6)^{\top}$.
Notice that, as $M_1(\kappa)$ is well defined for all $\kappa\in\mathbb{R}_{>0}^8$, then $M'$ is 
also well defined for all $\kappa\in\mathbb{R}_{>0}^8$.

We now choose the matrix $M=M(\kappa)$:
\[M=\left(\begin{array}{rrrr}
\frac{1}{\kappa_8} & 0 & 0 & 0\\
0 & \frac{1}{\kappa_7} & 0 & 0 \\
0 & 0  & -\frac{1}{\kappa_6} & 0  \\
0 & 0 & 0 & 1
\end{array}\right)M',\]
and the effective parameters $\bar{a}_1=\frac{\kappa_1\kappa_3}{\kappa_8(\kappa_2+\kappa_3)}$, 
$\bar{a}_2=\frac{\kappa_4}{\kappa_7}$,  $\bar{a}_3=\frac{\kappa_5}{\kappa_6}$, and  
$\bar{a}_4=\frac{\kappa_1}{\kappa_2+\kappa_3}$.  Then, the reparametrization map $\bar{a}$ is 
surjective, $M(\kappa)$ is well defined for all $\kappa\in\mathbb{R}_{>0}^8$, and 
$\det(M)=(-\frac{1}{\kappa_6\kappa_7\kappa_8})\cdot (-\frac{1}{\kappa_2+\kappa_3})>0$.
\end{example}

\section{Proof of Lemma \ref{lem:4}} \label{sec:pf-lem-4}

In this appendix, we prove Lemma \ref{lem:4}.

\noindent
\begin{proof}[Proof of Lemma \ref{lem:4}]
Let $\varepsilon >0$.  
By hypothesis, $f(a, z)$ and $(a^*,z^*) \in \mathbb{R}^{n+1}$ satisfy:
	\begin{enumerate}[(I)]
	\item $f(a^*,z^*)=0$,
	\item  $\frac{\partial f}{\partial z}(a^*,z^*)=0$,
	\item  $\frac{\partial^2 f}{\partial z^2}(a^*,z^*) \neq 0$, and
	\item  $\frac{\partial f}{\partial a_\ell}(a^*,z^*)\neq0$.
	\end{enumerate}
By assumptions (I) and (IV), the implicit function theorem applies. 
So, there exists a function 
	\[
	\tilde{\beta}:B_{\varepsilon'}(a_1,\dots,a_{\ell-1},a_{\ell+1},\dots,a_n,z)\to \mathbb{R} \]
defined on a ball of some radius $\varepsilon' \le \varepsilon$ in $\mathbb{R}^n$ 
	 such that 
$\tilde{\beta}(a_1^*,\dots,a_{\ell-1}^*,a_{\ell+1}^*,\dots,a_n^*,z^*)=a_\ell^*$
and, near the point of interest $(a^*, z^*)$, the $f=0$ locus is the graph of $\tilde{\beta}$. 
%In particular,
%the following equality holds in the ball
%$B_{\epsilon'}(a_1,\dots,a_{\ell-1},a_{\ell+1},\dots,a_n,z)$:
%\begin{align*}
%f\left( a_1,\dots,a_{\ell-1},\tilde{\beta}(a_1,\dots,a_{\ell-1},a_{\ell+1},\dots,a_n,z),a_{\ell+1},\dots,a_n,z\right)~=~0 ~.
%\end{align*}

Call $\beta: (z^*-\epsilon', z^*+\epsilon') \to \mathbb{R}$ the restriction $\beta(z) =\tilde{\beta}(a_1^*,\dots,a_{\ell-1}^*,a_{\ell+1}^*,\dots,a_n^*,z)$, 
and call $\hat{\beta}(z)=(a_1^*,\dots,a_{l-1}^*,\beta(z),a_{\ell+1}^*,\dots, a_n^*)$;
 then near $(a^*, z^*)$ we have %, the $f=0$ locus is the graph of $\beta$, i.e.:
\begin{align} \label{eq:f=0}
f(\hat{\beta}(z),z)~=~0 \quad {\rm for~} z \in (z^*-\varepsilon', z^*+\varepsilon')~.
\end{align}

Take the derivative of equation~\eqref{eq:f=0}, via the chain rule:
\begin{align} \label{eq:deriv-1}
\pd{f}{a_\ell} \left( \hat{\beta}(z),z \right) \beta'(z) + \pd{f}{z}(\hat{\beta}(z),z) ~=~0~.
\end{align}
Evaluating this equation at $z=z^*$, and recalling that $\beta(z^*)=a_\ell^*$ and %that 
$\hat{\beta}(z^*)=a^*$, we obtain:
\begin{align*}
\pd{f}{a_\ell}(a^*,z^*) \beta'(z^*) + \pd{f}{z}(a^*,z^*) ~=~0~.
\end{align*}
Recall that $\pd{f}{a_\ell}(a^*,z^*)\neq 0$, by hypothesis (IV), and $\pd{f}{z}(a^*,z^*)=0$, by (II).  Thus, $\beta'(z^*)=0$.
Next,  take another derivative, applying the chain rule to equation~\eqref{eq:deriv-1}, and then evaluate at $z=z^*$:
\begin{align*}
\pd{f}{a_\ell}(a^*,z^*) \beta''(z^*) +\pd{^2 f}{a_\ell^2} (a^*,z^*) \left(  \beta'(z^*)  \right)^2
	+ 2\pd{^2 f}{z \partial a_\ell} (a^*,z^*) \beta'(z^*) + \pd{^2 f}{z^2} (a^*,z^*)	
~=~0~.
\end{align*}
Thus, we deduce from (III) and (IV) that  $\beta''(z^*) \ne 0$. It follows that the univariate function $\beta$ has a maximum or a minimum at $z^*$ (depending
on the sign of the second derivative).
Hence, there exists a sufficiently small $\delta>0$ such that for all 
$\delta' \in (0,\delta)$, either $b_\ell^{**}=a_\ell^*-\delta'$ (if $a_\ell^*$ is a local maximum) or $a_\ell^{**}=a_\ell^*+\delta'$ 
(if $a_\ell^*$ is a local minimum) yields $f(a^{**}, z)=0$, for $a^{**}=(a_1^*,\dots,a_{\ell-1}^*,a_\ell^{**},a_{\ell+1}^*,\dots,a_n^*)$, 
with two distinct real solutions within distance $\varepsilon$ of $z^*$. %$\Box$
\end{proof}

\section{Proof of Theorem \ref{shape}} \label{sec:pf-shape-thm}
%\section*{Appendix B: Proof of Theorem \ref{shape}}
The goal of this appendix is to prove Theorem \ref{shape}. % in the main text. 
We first recall some definitions and a result of Kapur, Sun, and Wang
from the theory of comprehensive Gr\"obner bases~\cite{SunYao2010}. 
For basic concepts from computational algebraic geometry, 
see %we direct our readers to 
the books
\cite{CLOb, CLO}. 

%Denote the field of complex numbers by ${\mathbb C}$. 
Let $h\in {\mathbb C}[a, x] := {\mathbb C}[a_1, a_2, \ldots, a_n, x_1, x_2, \ldots,  x_s]$.
We denote
by 
\[
{\rm lpp}_x(h) \quad \quad {\rm and}  \quad \quad {\rm lc}_x(h)~,
\]
the leading monomial (or ``leading power product'') and leading coefficient of $h$, respectively,
when $h$ is viewed in 
	$\mathbb{C}(a)[ x]$
 taken with the lexicographic order $x_{s}<\cdots<x_2<x_{1}$.
For instance, if $h=a_1^2x_1 + x_2$, then ${\rm lpp}_x(h)=x_1$ and ${\rm lc}_x(h)=a_1^2$. 

\begin{definition} \cite[Definition 4.1]{SunYao2010} \label{noncomp}
Given  $H\subseteq {\mathbb C}[a, x]$, a subset $H'$ of $H$ is a
{\em noncomparable subset of} $H$ if 
\begin{enumerate}
	\item for every $h\in H$, there exists $g\in H'$ such that ${\rm lpp}_x(h)$ is a multiple of ${\rm lpp}_x(g)$, and
	\item for every $g_1, g_2\in H'$, with $g_1 \neq g_2$, the leading monomial 
	${\rm lpp}_x(g_1)$ is \underline{not} a multiple of ${\rm lpp}_x(g_2)$, and 
${\rm lpp}_x(g_2)$ is \underline{not} a multiple of ${\rm lpp}_x(g_1)$. 
	\end{enumerate} 
\end{definition}

\begin{example}
Consider $H=\{a_2x_2^2-1, a_1x_1-1, (a_1+1)x_1-x_2, (a_1+1)x_2-a_1\}$.
Let $H'=\{a_1x_1-1, (a_1+1)x_2-a_1\}~ (\subseteq H)$. We verify that $H'$ is a noncomparable subset of $H$:
\begin{enumerate}
\item Note that $\{{\rm lpp}_x(h) \mid h \in H\}=\{x_2^2, x_2, x_1\}$ and 
$\{{\rm lpp}_x(g) \mid g \in H'\}=\{x_2, x_1\}$. So, every monomial in the first set is a multiple of
some monomial in the second set. 
%That means the condition 1 in  Definition \ref{noncomp} is satisfied. 
\item For the two polynomials in $H'$, their leading monomials are, 
respectively, $x_1$ and $x_2$. 
We see that $x_1$ is not a multiple of $x_2$, and $x_2$ is not a multiple of $x_1$.
%So 
%neither $x_1$ is a multiple of $x_2$ nor $x_2$ is a multiple of $x_1$. 
%That means the condition 2 in Definition \ref{noncomp} is satisfied. 
\end{enumerate} 
\end{example}

The following straightforward lemma shows that noncomparable subsets always exist and explains how (in theory) to effectively find one.

\begin{lemma}[Existence of noncomparable subsets] \label{lem:noncomparable-set}
Let $H$
be a finite, nonempty subset of $ {\mathbb C}[a,x]$.  The following procedure yields a noncomparable subset of $H$:

\begin{enumerate}
\item  Let $M=\{{\rm lpp}_x(h) \mid h\in H\}$, and let $D=\emptyset$.

\item  Pick a monomial from $M$, say $m_1$. 
Let $d:=m_1$. Search $M$. If there exists $m \in M$ such that $m|d$ and $m\neq d$, then set $d := m$. 
Continue searching until there is no $m$ in $M$ such that $m|d$ and $m\neq d$. 
Then add $d$ into $D$.

\item Let $M' := \{ m \in M: d|m\}$, and let $M := M \backslash M'$. 

\item Repeat steps 2--3 until $M$ is empty.

\item For each $d \in D$, pick some $f_d \in F$ for which ${\rm lpp}_x(f)=d$.  Output $F_D:=\{f_d \mid d \in D \}$.  
\end{enumerate}
\end{lemma}

Now consider a set of $s$ polynomials 
$H=\{h_1,h_2, \ldots, h_s\}$ in 
	${\mathbb C}[a,x]$. 
Denote by ${\mathcal I}(H)$ the ideal generated by $H$ in the polynomial ring. 
Denote by $V(H)$ the %affine 
variety generated by $H$ (or, equivalently generated by ${\mathcal I}(H)$) in ${\mathbb C}^{n+s}$:
\[
	V(H) ~:=~
	\{(a;x)=(a_1, a_2, \ldots a_n, x_1, x_2, \ldots,  x_s)\in {\mathbb C}^{n+s} \mid h(a;x)=0 ~{\rm for~all}~h\in {\mathcal I}(H)\}~.
	\]

Below, we show that if the system $H$ is a general zero-dimensional system (Definition \ref{def:gzd}),  
then $H$ admits a triangular form (see the proof of Theorem \ref{shape}). 
More 
specifically, we prove that any noncomparable subset of a Gr\"obner basis of ${\mathcal I}(H)$ with respect to
the lexicographic order $a_n<\cdots<a_2<a_1<x_{s}<\cdots<x_2<x_{1}$ has 
the desired triangular form. 
The proof requires the following result, due to Kapur, Sun, and Wang \cite[Theorem 4.3]{SunYao2010}, 
which relates Gr\"obner bases of ${\mathcal I}(H)$ to those of the specialized ideal ${\mathcal I}(H|_{b=b^*})$: 

\begin{proposition}[Specialization of Gr\"obner bases~\cite{SunYao2010}] \label{sg}
Consider $H \subseteq {\mathbb C}[a, x]$,  and let ${\mathcal G}$ be a Gr\"obner basis of the
 ideal ${\mathcal I}(H) \subseteq 
{\mathbb C}[a,x]$ with respect to the lexicographic order 
$a_n<\cdots<a_2<a_1<x_{s}<\cdots<x_2<x_{1}$.
Let ${\mathcal G}_{\cap}={\mathcal G}\cap {\mathbb C}[a]$, 
let ${\mathcal G}_m$ be a noncomparable subset of ${\mathcal G}\backslash {\mathcal G}_{\cap}$, and
let $h=\Pi_{g
\in {\mathcal G}_m}{\rm lc}_{x}(g)$. For any $a^*=(a_1^*, a_1^*, \ldots, a_n^*)\in {\mathbb C}^n$,  if 
	${\mathcal G}_{\cap}|_{a=a^*} \subseteq \{0\}$ 
	and $h|_{a=a^*}\neq 0$,  then 
${\mathcal G}_m|_{a=a^*}$ is a Gr\"obner basis of 
the ideal ${\mathcal I}(H|_{a=a^*})
~\subseteq~ {\mathbb C}[x]$ with respect to the lexicographic order $x_{s}<\cdots<x_2<x_{1}${\bf.}

\end{proposition}

To prove Theorem \ref{shape}, we also need the following lemma:

\begin{lemma}\label{nonempty}
For a general zero-dimensional system
$H=\{h_1,h_2, \ldots, h_s\}\subseteq
	{\mathbb C}[a,x]$, we have 
${\mathcal I}(H)\cap  {\mathbb C}[a] = \{0\}.$ 
\end{lemma}
\begin{proof}
Since $H$ is a general zero-dimensional system (Definition \ref{def:gzd}), 
there exists a proper variety ${\mathcal W} \subsetneq \mathbb{C}^n$
such that for every $a^* \in \mathbb{C}^n \setminus {\mathcal W}$ 
the specialized system $H|_{a=a^*}$ has at least one complex solution.  
It follows that ${\mathbb C}^n\backslash {\mathcal W}\subseteq \pi(V(H))$, 
where
$\pi: \mathbb{C}^{n+s} \to \mathbb{C}^n$ denotes the standard projection 	
given by $(a,x) \mapsto   a$.  Thus,
\[
\overline{{\mathbb C}^n\backslash {\mathcal W}}
	~\subseteq~ \overline{\pi(V(H))}
	~\subseteq~ \mathbb{C}^n~. 
\]	
Note that
$\overline{{\mathbb C}^n\backslash {\mathcal W}}={\mathbb C}^n$ 
(as ${\mathcal W} \subsetneq \mathbb{C}^n$), so 
$\overline{\pi(V(H)}={\mathbb C}^n$. 
By \cite[pg.\ 193, Thm.\ 3]{CLO}, we know 
that 
$\overline{\pi(V(H))}
=
V({\mathcal I}(H)\cap {\mathbb C}[a])$.
So, $V({\mathcal I}(H)\cap {\mathbb C}[a])={\mathbb C}^n$. 
The only ideal that generates the variety ${\mathbb C}^n$ is the zero ideal.
% $\{0\}$. 
So, ${\mathcal I}(H)\cap {\mathbb C}[a] = \{0\}$.
\end{proof}

\begin{proof} [{Proof of Theorem \ref{shape}}]
Let $H=\{h_1,h_2, \ldots, h_s\}$. 
Let ${\mathcal G}_m$ 
be a noncomparable subset 
of ${\mathcal G}$ (which exists by
Lemma~\ref{lem:noncomparable-set}).  
By Lemma \ref{nonempty}, we have:
\[
	{\mathcal G}_m\cap {\mathbb C}[a]
		~\subseteq ~
		{\mathcal I}(H)\cap {\mathbb C}[a]
		~=~\{0\}~.
\] 
%	${\mathcal G}\cap {\mathbb C}[a_1, a_2, \ldots, a_n]=\emptyset.
In fact, $0 \notin {\mathcal G}_m$ (by the definition of noncomparable subset and because $\mathcal G \neq \{0\}$), so 
${\mathcal G}_m\cap {\mathbb C}[a]
=\emptyset $. 
So, 
by Proposition~\ref{sg}, for every 
 $a^*\in %V(\emptyset) \backslash V(h)=
 {\mathbb C}^n\backslash V(h)$, where $h=\Pi_{g
\in {\mathcal G}_m}{\rm lc}_{x}(g)$,
the set 
 ${\mathcal G}_m|_{a=a^*}$ is a Gr\"obner basis of 
%the ideal 
$	{\mathcal I}(H|_{a=a^*})
\subseteq {\mathbb C}[x]$
with respect to %the lexicographic order 
$x_{s}<\cdots<x_{1}$.

(1) We show that a subset $\{g_1, g_2, \ldots , g_s\}$ of ${\mathcal G}_m$ has the
required triangular form.
  %Note the variety $W\cap V(\Pi_{g\in G_m} {\rm lc}_{x}(g))$ is still nontrivial. 
 As $H$ is a general zero-dimensional system, let $\mathcal W$ be the variety in ${\mathbb C}^n$ 
 such that 
$H$ satisfies the hypotheses (A1)--(A3) in Definition \ref{def:gzd}. 
  Let $c^*\in {\mathbb C}^n\backslash \left({\mathcal W}\cup V(h)\right)$. 
  Then 
 by (A1) in Definition \ref{def:gzd},
 we know that $V({\mathcal G}_m|_{a=c^*})=V(H|_{a=c^*})$ is a nonempty finite set in ${\mathbb C}^s$. 
 Hence, by \cite[pg.\ 234, % Chapter 5, \S 3, 
 Thm.\ 6 (i) and (iii)]{CLO},
 for $i\in\{1, 2, \ldots, s\}$, there exists $g_i \in {\mathcal G}_m$ such that 
 the leading monomial of 
 $g_i|_{a=c^*}$ 
 %${\rm lpp}_x(g_i|_{a=c^*})$ 
 has the form
 $x_i^{N_i}$, where $N_i$ is a non-negative integer.  
  In particular, 
 	$g_s|_{a=c^*}\in \mathbb C[x_{s}]$. 
 
 It follows that
 ${\rm lpp}_x(g_i) = x_i^{N_i}$, because $c^*\notin  V(h)$ and so ${\rm lc}_x(g_i)|_{a=c^*} \neq 0$.  
 Hence, if we show that $N_2=N_3=\dots=N_s=1$, then, by the definition of the
 lexicographic order, $g_1, g_2, \ldots , g_s$ have the forms shown in Theorem~\ref{shape} (1).

 Hence, to finish proving (1), we need only show that $N_2=N_3=\dots=N_s=1$.
Let $N:= \lvert V(H|_{a=c^*}) \rvert$.  
%, and note that $0 < N < \infty$, by (A1) in Definition \ref{def:gzd}.
Then every $x^* \in V(H|_{a=c^*})$ has a distinct $x_s$-coordinate
(by (A2) in Definition \ref{def:gzd}), and every such coordinate is a root of 
$g_s|_{a=c^*}\in \mathbb C[x_{s}]$.  Hence, 
\begin{align} \label{eq:bound-N}
	N_s ~=~ {\rm deg}(g_s|_{a=c^*} ) ~\geq ~ N~.
\end{align}
Next, by (A3) in Definition \ref{def:gzd} and 
\cite[pg.\ 235, Prop.\ 8(ii)]{CLO}, 
%\cite[Page 235, Chapter 5, \S3, Proposition 8 (ii)]{CLO}, 
we know that 
	\begin{align} \label{eq:bound-N-2}
	N~=~N_1N_2 \cdots N_s~. %\Pi_{i=1}^s N_i~. 
	\end{align}
So, by~\eqref{eq:bound-N} and~\eqref{eq:bound-N-2}, we have $N_s=N$ and $N_2 = N_3 =\cdots=N_{s-1}=1$.

 (2) Consider the following claim: \\
 \noindent
 {\bf Claim:} ${\mathcal G}_m= \{g_1, ~g_2, ~\ldots, ~g_s\}$.\\
 This claim implies that $h=Q_{s, N} \, Q_1 \, Q_2 \cdots Q_{s-1}$ 
 and so, by what we saw earlier, for every 
  $a^*\in {\mathbb C}^n\backslash 
	V(Q_{s, N} \, Q_1\,  Q_2 \cdots Q_{s-1})$,
	the set 
 $
 \{g_1|_{a=a^*},~ g_2|_{a=a^*},~ \ldots,~ g_s|_{a=a^*} \}$ is a Gr\"obner basis of 
the ideal 
$	{\mathcal I}(H|_{a=a^*})$
with respect to %the lexicographic order 
$x_{s}<\cdots<x_{1}$.
So, to complete the proof, we need only prove the Claim.
   
   First, the containment 
  ${\mathcal G}_m\supseteq \{g_1, g_2, \ldots g_s\}$ follows from the fact that
  the  $g_i$'s were selected from ${\mathcal G}_m$.
  % so we only need to show 
  Next, we show the containment ${\mathcal G}_m\subseteq \{g_1, g_2, \ldots , g_s\}$ by proving the following equality by induction on $i$:
 \begin{equation} \label{eq:prove-ind}
 	{\mathcal G}_m\cap {\mathbb C}[a, x_{i},x_{i+1} \ldots, x_{s}]
	~\subseteq ~
	 \{g_{i},g_{i+1}, \ldots, g_s\}~, \quad {\rm for~all}~ i\in \{1, 2, \ldots, s\}~.
 \end{equation} % by induction on $i$. 
 For $i=s$, assume that $g \in 
	{\mathcal G}_m\cap {\mathbb C}[a, x_{s}]$.
Then ${\rm lpp}_x(g)= x_s^{\tilde{N}}$ for some $\tilde{N} \geq 0$, and also
recall that ${\rm lpp}_x(g_s)= x_s^N$  for some $N >0$.  However,
${\mathcal G}_m$ is noncomparable, so $g=g_s$.

For the inductive step, assume that the containment~\eqref{eq:prove-ind} holds for all $i \in \{j, j+1, \dots, s \}$ (for some $j\leq s$). 
For $i=j-1$, let $g\in {\mathcal G}_m\cap {\mathbb C}[a_1, a_2, \ldots , a_n, x_{j-1}, x_j,\ldots, x_{s}]$,
and write ${\rm lpp}_x(g) = x_s^{a_s}\cdots x_j^{a_j}x_{j-1}^{a_{j-1}}$. 
If $a_{j-1}=0$, then 
$g\in {\mathcal G}_m\cap {\mathbb C}[a_1, a_2, \ldots , a_n, x_j, x_{j+1},\ldots, x_s]$, 
 so by the induction 
 hypothesis, 
 $g\in \{g_j, g_{j+1},\ldots, g_s\}$. 
 If $a_{j-1}>0$, then 
 ${\rm lpp}_x(g_{j-1})=x_{j-1} | {\rm lpp}_x(g)$.
Hence,  $g=g_{j-1}$ (because $\mathcal G_m$ is noncomparable).
\end{proof}

%\begin{remark} 
%~
%{\color{red} Delete this remark?}
%If the polynomial system $F\subseteq {\mathbb Q}[a_1, a_2, \ldots a_n, x_1, x_2, \ldots,  x_s]$, by the Buchberger's algorithm \cite{CLO}, we know 
%a Gr\"obner basis ${\mathcal G}$ of ${\mathcal I}(F)$ is also a subset of ${\mathbb Q}[a_1, a_2, \ldots a_n, x_1, x_2, \ldots,  x_s]$ and hence 
%the polynomials $g_1, g_2, \ldots, g_s$ stated in Theorem \ref{shape} will all have rational coefficients. 
%\end{remark}

\end{document}